\documentclass[11pt]{amsart}
\usepackage[utf8]{inputenc}
\pdfoutput=1	

\usepackage{fullpage}
\usepackage{bbm}

\usepackage{ graphicx, amsmath, amssymb,csquotes}

\usepackage{tikz}
\usepackage{tikz-dependency}

\usetikzlibrary{arrows.meta}
\usetikzlibrary{decorations}
\usetikzlibrary{decorations.markings}
\usetikzlibrary{decorations.pathreplacing}

\tikzset{  fullVertex/.style={circle, draw=black, thick, fill= black,  minimum size =2.5mm, inner sep=0mm},
	diffVertex/.style={circle, draw=black, thick, fill= white,  minimum size =2.5mm},
	point/.style={circle, draw=black, thick, fill= black,  minimum size =1.0mm, inner sep=0mm},
	treeVertex/.style={regular polygon, regular polygon sides=4,   draw, thick, fill= white,  minimum size =3.8mm, inner sep=0mm}, 
	counterVertex/.style={draw=black, line width=.2mm, circle, path picture={ 
			\draw[line width=.2mm] (-1.3mm,-1.3mm) -- (1.3mm,1.3mm) ;
			\draw[line width=.2mm] (-1.3mm,1.3mm) -- (1.3mm,-1.3mm);
	}},
	treeCounterVertex/.style={regular polygon, regular polygon sides=4,   draw=black, line width=.2mm, fill= white,  minimum size =4.5mm, inner sep=0mm, path picture={ 
			\draw[line width=.2mm] (-1.5mm,-1.5mm) -- (1.5mm, 1.5mm) ;
			\draw[line width=.2mm] (-1.5mm, 1.5mm) -- (1.5mm,-1.5mm);
	}},
	-|-/.style={decoration={markings, 	mark=at position .5 with {\arrow{|}}},postaction={decorate}},
	every picture/.style=thick
}

\tikzset{
	ncbar angle/.initial=90,
	ncbar/.style={
		to path=(\tikztostart)
		-- ($(\tikztostart)!#1!\pgfkeysvalueof{/tikz/ncbar angle}:(\tikztotarget)$)
		-- ($(\tikztotarget)!($(\tikztostart)!#1!\pgfkeysvalueof{/tikz/ncbar angle}:(\tikztotarget)$)!\pgfkeysvalueof{/tikz/ncbar angle}:(\tikztostart)$)
		-- (\tikztotarget)
	},
	ncbar/.default=0.5cm,
}

\tikzset{square left bracket/.style={ncbar=0.5cm}}
\tikzset{square right bracket/.style={ncbar=-0.5cm}}

\usepackage{hyperref}
\usepackage[nameinlink]{cleveref}	
\crefname{enumi}{rule}{rules}

\makeatletter
\newcommand{\mylabel}[2]{#2\def\@currentlabel{#2}\label{#1}}
\makeatother

\usepackage{placeins} 
\usepackage{flafter} 

\newtheorem{theorem}{Theorem}[section]
\newtheorem{lemma}[theorem]{Lemma}

\theoremstyle{definition}
\newtheorem{definition}[theorem]{Definition}
\newtheorem{example}[theorem]{Example}

\numberwithin{equation}{section}

\usepackage[
backend=bibtex,
]{biblatex}
\bibliography{diffeomorphism2.bib}

\newcommand{\renop} { \mathcal {R}  }
\newcommand{\renopk}[1] { \renop \left[ #1 \right]  }
\renewcommand{\d}{\textnormal{d}}
\newcommand{\abs}[1]{\left | #1 \right |}

\title{Propagator-cancelling scalar fields}
\author{Paul-Hermann Balduf}
\date{\today}

\thanks{The author thanks Dirk Kreimer for helpful discussion as well as Daniel Reiche and Alexandra Glück for proofreading.}

\begin{document}

\begin{abstract}
  We examine a large class of scalar quantum field theories where vertices are able to cancel adjacent propagators. These theories are obtained as diffeomorphisms of the field variable of a free field. Their connected correlations functions can be computed efficiently with an algebraic procedure  dubbed connected perspective in earlier work.    Specifically,
  
  (1) We compare the Feynman rules in momentum-space to their position-space counterparts, they agree.  This is a a-posteriori justification for the Feynman rules of the connected perspective.
  
  (2) We extend the connected perspective to also incorporate counterterm vertices.
  
  (3) We examine a specific choice of diffeomorphism that assumes the vertices at different valence proportional to each other. We find that these relations lead to various simplifications.

  (4) We perturbatively compute the connected 2-point function to all orders for an arbitrary diffeomorphism of a massless field. We thereby give a systematic perturbative derivation of the exponential superpropagator known from the  literature.

  (5) We compute a number of  one-  and two-loop counterterms. We find that for the specific diffeomorphism, the one-loop counterterms arise from the bare vertices by a non-linear redefinition  of the momentum. 
  
  (6) We show that every diffeomorphism fulfils a set of infinitely many Slavnov-Taylor-like identities which express diffeomorphism-invariance of the S-matrix at loop-level.
  
  Finally, we comment on surprising similarities between the structure of propagator cancelling scalar theories and gauge theories .
\end{abstract}

\maketitle

\section{Introduction}

\subsection{Motivation}
This paper concerns propagator-cancelling scalar theories, that is scalar quantum field theories  where vertex Feynman rules are able to cancel adjacent propagators.  We restrict ourselves to those theories which can be obtained from a free field theory  by a non-linear redefinition of the field variable.  The implications of the latter restriction depend on the precise form of the propagator. If the propagator is quadratic in momenta and massless, then  all propagator-cancelling theories fall in this class, see \cref{thm:generality}. In the following, the word "diffeomorphism" will always refer to a transformation of the field variable, not to be confused with spacetime diffeomorphisms. There are several reasons to study propagator-cancelling theories and especially diffeomorphisms:
\begin{enumerate}
	\item  Non-linear field redefinitions do not alter the $S$-matrix. This is frequently used to simplify the Lagrangian of a quantum field theory,  e.g. when treating gauge invariance   \cite[ch. 6.3]{lee_gauge_1981} or non-local interactions \cite{chebotarev_smatrix_2019}. The invariance of the $S$-matrix has been demonstrated in the path integral formalism \cite{pointPathIntegral,apfeldorf} and also using graph theory \cite{thooft_diagrammar_1973,velenich,KY17,balduf_perturbation_2020,mahmoud_diffeomorphisms_2020,mahmoud_enumerative_2020}. The behaviour of correlation functions of the field diffeomorphism has never been addressed in detail to the author's knowledge apart from the statement that they vanish onshell.
	
	\item It is well known   that linear shifts in the field variable of an interacting field alter the type of interaction,e.g. \cite{peskin_introduction_1995}. Remarkably, they do not destroy renormalizability of a theory, which makes it possible to formulate theories with spontaneous symmetry breaking. One can ask whether a similar mechanism is at work in non-linear field transformations.
	
	\item There are two different connections with quantum Einstein gravity. Firstly, in gravity the  renormalization Hopf algebra equals the core Hopf algebra  \cite{gravity}. This is also true for a propagator-cancelling theory, which qualifies the latter as a model for the algebraic behaviour of gravity. 
	Secondly, although gravity is not a diffeomorphism of a free field \cite{berends_spin_1984}, non-linear redefinitions of the field variable have been considered already decades ago \cite{isham_quantum_1973}. By now, several different choices have been proposed, all of which lead to a non-renormalizable quantum theory. A better general understanding of field diffeomorphisms might clarify what can and what cannot be altered by them, apart from mere onshell invariance.  
	
	\item Historically, field redefinitions were prominently used to treat non-polynomial interaction terms, see \cref{subs:history}, but these computations often relied on ad-hoc prescriptions. They concentrated on the 2-point-function, leaving open the question how they relate to the modern formalism of perturbation theory.	
\end{enumerate}

\subsection{Historical background}\label{subs:history}

Perturbative renormalization of quantum fields is possible   if  the bare Lagrangian contains only  coupling constants with non-negative mass dimensions \cite{sakata_structure_1952,heisenberg_zur_1936}. 
Since this classification is based on order-by-order perturbation theory, it is conceivable that non-renormalizable interactions still give finite results either by a non-perturbative treatment \cite{okubo_note_1954}, by analytic properties of the interaction term \cite{blomer_zeromass_1971}, or by non-trivial cancellation effects of higher order terms in the perturbation expansion \cite{kamefuchi_structure_1953,dewitt_gravity_1964,lazarides_highenergy_1972}. An often studied example of a non-renormalizable interaction are theories of Liouville-type \cite{liouville_equation_1853}
\begin{align}\label{L_liouville}
\mathcal L = \frac 12 \partial_\mu \phi \partial^\mu \phi -  \exp(g \phi).
\end{align}
In two dimensions, Liouville theory is solved by mapping its solutions to modes of a free field  via B\"acklund transformation \cite{backlund_zur_1880} which enjoyed significant attention in the 1980s \cite{braaten_exact_1983,dhoker_classical_1982} for its connection to string theory \cite{polyakov_quantum_1981}. In four dimensions, it has been treated, e.g., in  \cite{osipov_feynman_1981,efimov_formulation_1965}.

A suitable field diffeomorphism, i.e. a global redefinition
\begin{align}\label{def_diffeomorphism}
\phi(x) &= \sum_{j=0}^\infty a_j \rho^{j+1}(x), \qquad a_0=1,
\end{align}
of the field variable $\phi$ in terms of another field $\rho$ by constant coefficients $a_j\in \mathbbm C$,  can possibly turn a theory with non-polynomial interaction term into a theory with polynomial interaction, but non-standard propagator \cite{volkov_quantum_1968}. This motivated the study of such non-standard propagators, where the \emph{exponential superpropagator}  
\begin{align}\label{superpropagator}
G_2(x) &:= \exp \left( -g^2 G_F(x) \right) , \qquad G_F (x) \text{ a suitable free propagator}
\end{align}
is the most prominent example. Contrary to ordinary propagators, it is not a tempered distribution in position space \cite{jaffe_form_1966} and defining its Fourier transform to momentum space requires additional assumptions or constraints. This has produced a number of results, e.g. \cite{okubo_note_1954,lehmann_superpropagator_1971,blomer_zeromass_1971}, which differ by finite terms $\delta_k$ in each order in $p^2$. With the conventions of \cite{bollini_exponential_1974},
\begin{align}\label{superpropagator_mom}
&G_2(p) = \frac{i g^24\pi ^2}{p^2}  +i\pi^2 g^4\sum_{k=0}^\infty \frac{4^{-k}g^{2k} (p^2)^k  \left(\ln \left( \frac{g^2 p^2}{4} \right)  + \delta_k \right) }{ \Gamma(k+1) \Gamma(k+2) \Gamma(k+3)}.
\end{align}

\subsection{Notation and conventions}

Let $\hat s$ be the differential operator defining a free field theory via its Lagrangian
\begin{align}\label{Lfree}
\mathcal L &= \frac 12 \phi \hat s \phi .
\end{align}
The Fourier transform $s_p$ of the field differential operator, 
\begin{align}\label{def_offshellvariable} 
	s_p \cdot e^{ipx} :=  \hat s e^{ipx},
\end{align}
is called the \emph{offshell variable} of the theory. A momentum $p$ is said to be \emph{onshell} if $s_p=0$.

\begin{example}[Standard scalar theories]
	The most common free scalar theory has the Lagrangian
	\begin{align*}
	\mathcal L = -\frac 12 \phi \partial_\mu \partial^\mu \phi -\frac 12 m^2 \phi^2, 
	\end{align*}
 field differential operator $\hat s = -\partial_\mu \partial^\mu- m^2$ and    offshell variable 
	\begin{align*}
		 s_p = p^2-m^2.
	\end{align*}
	The corresponding massless theory is obtained by setting $m=0$, i.e.
\begin{align*}
	\mathcal L &= -\frac 12 \phi \partial_\mu \partial^\mu \phi,  \qquad 	\hat s = -\partial_\mu \partial^\mu,\qquad 	s_p = p^2.
\end{align*}
\end{example}

In the following, we will make reference to a specific free Lagrangian of type \cref{Lfree} by its corresponding offshell variable. Addition of indices is understood as acting on the momenta, i.e. if there are numbered momenta $p_1, p_2, \ldots$ then $s_{1+2} := s_{p_1 + p_2} $ which in general is not equal $s_1+s_2$. If an edge $e$ of a Feynman graph carries momentum $p_e$ then we similarly write $s_e\equiv s_{p_e}$ for the corresponding offshell variable.

Many statements of the present paper are purely combinatoric in nature and can be formulated in any dimension of spacetime. For the sake of clarity, in all explicit calculations we use dimensional regularization \cite{bollini_dimensional_1972,thooft_regularization_1972}, working in $D = 4-2\epsilon $ dimensions. In this way, a Feynman graph evaluates to a Laurent series in the regularization parameter $\epsilon$. The part of this series with negative exponents of $\epsilon$ diverges in the physical limit $D \rightarrow 4$. it represents the divergent part of the graph under consideration. 
For a graph $\Gamma$, we define $\renopk \Gamma$ to be the projection of the   amplitude onto its divergent part,  
\begin{align}\label{def_R}
	\Gamma &= \sum_{k=-n}^\infty \epsilon^k c_k \qquad \Rightarrow \qquad \renopk{\Gamma} := \sum_{k=-n}^{-1} \epsilon^k c_k.
\end{align}

The propagator $G_F(z)$ of a theory is defined as the Green function of the corresponding field differential operator \cref{Lfree}, $\hat s G_F(z) = i\delta(z)$. In momentum space it is given by
\begin{align}\label{GF}
	G_F(p) &= \frac{i}{s_p + i0},
\end{align}
where, in the following, the causal prescription $+i0$ will be implicit.  An edge $e$ in a Feynman graph in momentum space thus contributes to the  amplitude with a factor $\frac{i}{s_e}$. The 4-dimensional Fourier transforms of the propagator $\frac i {s_p}$ for a massless field  reads, e.g.  \cite[p.30]{huang_quantum_1998}, 
\begin{align}\label{propagators}
	s&=p^2 : \quad &G_F (x)&= \int \frac{\d^D k}{(2 \pi)^D}\; \frac i {k^2}e^{-ikx}=   \frac{i}{ (2 \pi)^2 x^2 } \left( 1+ \epsilon \left( \gamma_E + \ln (\pi x^2)\right) + \mathcal O (\epsilon^2)  \right)   .
\end{align}

We will often encounter symmetric sums over all permutations of certain terms, therefore we  introduce a shorthand notation:
\begin{definition}\label{def_permutations}
	The expression $\langle k \rangle f(x_1,x_2, \ldots, x_n)$ denotes the sum over all $k$ different permutations of arguments in the function $f(x_1, \ldots, x_n)$.
\end{definition}
\begin{example}[4-valent tree amplitude]
There are six different ways to choose two out of four offshell variables but external momentum conservation of a 4-valent graph identifies them pair-wise like in $s_{1+2} = s_{3+4}$. Consequently there are only three actually different ways of building such offshell variables:
	\begin{align*}
		\langle 3 \rangle \frac{1}{s_{1+2}} = \frac{1}{s_{1+2}} + \frac{1}{s_{1+3}}  +\frac{1}{s_{1+4}}
	\end{align*}
\end{example}

\subsection{Field diffeomorphisms}

Diffeomorphisms of free fields have been studied in detail recently \cite{velenich,KY17,balduf_perturbation_2020,mahmoud_diffeomorphisms_2020}. In this section we list some central results. If a diffeomorphism \cref{def_diffeomorphism} $\phi(x) = \sum_{j=1}^\infty a_{j-1} \rho^{j}(x)$ is applied to a free Lagrangian \cref{Lfree}, one obtains a theory with an infinite set of $n$-valent interaction vertices with Feynman rules \cite{balduf_perturbation_2020}
\begin{align}\label{vn_general}
i v_n &= i \frac 12 \sum_{k=1}^{n-1} a_{n-k-1} a_{k-1} (n-k)!k! \sum_{P\in Q^{(n,k)}} s_P
\end{align}
where $s_P$ is the offshell variable \cref{def_offshellvariable} of a set of momenta $P$ and $Q^{(n,k)}$ is the set of all possibilities to choose $k$ out of $n$ external edges without distinguishing the order. Especially, for $k=1$ and $k=n-1$, the summands are proportional to the offshell variables adjacent to the vertex,
\begin{align}\label{vn_leading}
iv_n &= i(n-1)! a_{n-2}  \left( s_1 + \ldots + s_n \right) + \mathcal O\left( s_{i+j}  \right) .
\end{align}
 For $s=p^2-m^2$, the sum can be computed explicitly to produce \cite{KY17}
\begin{align}\label{vn}
i v_n &= i f_n \cdot \left( s_1 + s_2 + \ldots + s_n \right)  + i g_n m^2 , \\
\text {where} \qquad f_n &= B_{n-2,1} \left( 2! a_1, 3! a_2, \ldots \right) + B_{n-2,2} (2! a_1, 3! a_2, \ldots), \nonumber \\
g_n &= \frac 12 n (n-2)! \sum_{k=0}^{n-2} a_{n-k-2} a_k (n-k-2)k. \nonumber
\end{align}
Here, $B_{n,k}(x_1, \ldots)$ are the incomplete Bell polynomials, see e.g. \cite{comtet_advanced_1974}. 
If additionally $m=0$, such that $s_p = p^2$, then by momentum-conservation all sums of partitions of momenta can be rewritten in terms of the external momenta  and
\begin{align}\label{vn_massless}
iv_n &= \frac  {i(n-2)!}2\cdot \left( s_1+s_2 + \ldots + s_n \right) \cdot \sum_{k=0}^{n-2}a_{n-2-k} a_k (n-k-1)(k+1)   .
\end{align}

For scalar fields where the propagator is of quadratic order in momentum, so $s_p=p^2$ or $s_p=p^2+m^2$, essentially all propagator-cancelling theories are diffeomorphisms of a free field:
\begin{theorem} \label{thm:generality}
	Consider a scalar field theory $\rho$ with propagator quadratic in momentum.
	\begin{enumerate}
		\item If $\rho$ has interaction vertices $iv_n = k_n \cdot (p_1^2 + \ldots + p_n^2) + r_n \ \forall n>2$ where $k_n, r_n\in \mathbbm R$, then $\rho$ is a unique diffeomorphism of a field $\phi$ such that the vertices of $\phi$ are independent of momenta, $iv'_n = ir'_n $ where $r'_n\in \mathbbm R$.
		\item If $\rho$ is additionally a massless field and has interaction vertices $iv_n = k_n \cdot (p_1^2 + \ldots + p_n^2)$, then it is a diffeomorphism of a free massless scalar field $\phi$.
		\item There is no diffeomorphism between two power-counting renormalizable theories.
	\end{enumerate}
\end{theorem}
\begin{proof}
	For the first two cases, the most general interaction vertex which can be obtained from the corresponding $\phi$ can be brought into the required form, see \cite{balduf_perturbation_2020}. For (2), see \cref{vn_massless} and note that by tuning all parameters $a_n$, the proportionality constants of each vertex can be adjusted independently to match any given $k_n$.
	
	(3) follows from the fact that a theory can only be power-counting renormalizable if no vertex is proportional to squared momenta. By (1), there is a unique diffeomorphism to produce this form. The resulting theory might or might not be renormalizable, depending on valence $n$ of the vertices where the remaining coefficients are $r'_n \neq 0$.
\end{proof}

\subsection{Propagator cancellation and tree sums}

The vertices \cref{vn_general} are capable of cancelling a propagator $\frac{i}{s_e}$ where $e$ is an edge adjacent to the vertex by means of $s_e \cdot \frac{i}{s_e}=i$. This means that a vertex $iv_n$ residing in  a graph $\Gamma$ will generally change the topology of $\Gamma$ which poses a challenge to computations as well as combinatorial arguments in perturbation theory. 

It turns out that the vertex Feynman rules fulfil - for any choice of the diffeomorphism parameters $\{a_i\}_i$ - an infinite set of identities, as is seen most clearly in the form \cref{vn_general}: The momentum-dependent factors $s_P$ are in one to one correspondence with the possibilities to construct the $n$-valent vertex by joining two vertices of lower valence with an internal propagator. Since each vertex and propagator comes with an imaginary factor $i$, a graph with two joined vertices carries an overall minus sign compared to a single vertex at the same position. In this way, the contributions of $s_P$ to a $n$-point-vertex cancel out the contributions of the same $s_P$ arising from two joined vertices of lower valence. All that remains from $iv_n$ are those terms where $s_P = s_e$ is the offshell variable of an external edge of $iv_n$ and we have
\begin{align}\label{ordinary_recursive}
\sum_{v_j \star v_k = v_n} iv_j \frac{i}{s_P} iv_j \Big|_{\text{Terms not proportional to an external $s_e$}} &= -iv_n \Big|_{\text{Terms not proportional to an external $s_e$}}.
\end{align}
The product $v_j \star v_k$ means summation over all ways to choose the valences $j,k$ and also all possible permutations of external edges. 
On the right hand side the restriction implies that both vertices $v_j$ and $v_k$ cancel the intermediate propagator $\frac{i}{s_P}$ such that an overall factor $s_P$ arises, which, if $P$ is a non-trivial partition, is not proportional to any external $s_e$. From \cref{ordinary_recursive} it follows that in the connected tree-level $n$-point amplitude, the only remaining contributions are proportional to the external offshell variables $s_e$, that is, this amplitude has the form 
\begin{align}\label{Vn}
	-ib_n (s_1 + \ldots + s_n) =: iV_n
\end{align}
where $b_n$ is a c-number independent of kinematics.

\begin{definition}[Tree sum $b_n$]\label{def_bn}
	The tree sums $b_n$ for $n\geq 3$ are defined as the sum of all connected tree-level Feynman graphs of the field $\rho$  with a total of $n$ external edges, where  $n-1$ external edges  are onshell (i.e. $s_e=0$ for these edges $e$) and the last external edge is offshell. The propagator $\frac{i}{s_e}$ of this offshell edge $e$ is included in $b_n$. Finally, $b_2 :=1$.
\end{definition}
It can be shown \cite {KY17,balduf_perturbation_2020} that, regardless of the concrete form of $s_p$,
\begin{align}\label{thm_bn}
	b_{n+2} &= \sum_{k=1}^{n}\frac{(n+k)!}{n!}B_{n,k} \left( -1! a_1, -2! a_2, \ldots, -n! a_n\right) 
\end{align}
where $B_{n,k}$ are the incomplete Bell polynomials. 
Moreover, the same $b_n$ also are the coefficients of the inverse diffeomorphism of \cref{def_diffeomorphism},
\begin{align}\label{def_diffeomorphism_inverse} 
	\rho(x) &=   \sum_{n=1}^\infty  \frac{b_{n+1}}{n!} \phi^n(x).
\end{align}
This implies
\begin{align}\label{anbn}
	a_n &= \frac{1}{(n+1)!}\sum_{k=1}^n B_{n+k,k} \left( 0,-b_3,-b_4,-b_5,\ldots  \right) .
\end{align}

\subsection{The connected perspective}\label{sec:connectedperspective}

The fact that the tree sums \cref{def_bn} are mere numbers without any remaining internal propagators motivates to use these tree sums as \emph{metavertices} in computing connected correlation functions. This approach is dubbed \emph{connected perspective}, as opposed to the ordinary perspective with vertex Feynman rules \cref{vn_general},   and works according to the following Feynman rules: 
\begin{theorem}[Feynman rules of the connected perspective of a free field diffeomorphism]\label{thm_feynmanrules}
	~ \\ 
	Assuming vanishing of tadpoles, the $n$-point connected   amplitude  is obtained by summing over all graphs $\Gamma$ such that
	\begin{enumerate}
		\item Each internal edge $e\in \Gamma $ contributes a propagator factor $\frac{i}{s_e}$ .
		\item $\Gamma$ is built from $(k>2)$-valent metavertices with amplitude $iV_k = -ib_k (s_1 + s_2 + \ldots + s_k)$. Keeping a summand $s_e$ in this amplitude amounts to cancelling the adjacent edge $e$. 
		\item The metavertices do not cancel internal edges of $\Gamma$.
		\item There are no internal metavertices.
	\end{enumerate}
\end{theorem}
Point $(4)$ is redundant and kept for clarity. An internal metavertex is one that is not adjacent to any external edge of the graph. 

These Feynman rules are slightly unconventional, but they have the advantage that no more cancelling of edges occurs, as opposed to \cref{vn}. The so-constructed graphs $\Gamma$ coincide with the topologies which eventually are left if the Feynman amplitudes are constructed naively from vertices $v_n$. For further explanations and examples see \cite{balduf_perturbation_2020}.

If tadpole graphs vanish then all - tree-level and loop-level - connected correlation functions of the diffeomorphism field $\rho$ differ from the respective ones of $\phi$ only by terms proportional to some offshell variable $s_e$ for an external edge $e$. The onshell connected correlation functions, that is, the elements of the $S$-matrix, are unaltered by a diffeomorphism.

Note that the construction of connected tree-level amplitudes with \cref{thm_feynmanrules} reminds of onshell methods from $S$-matrix theory. For example, the tree-level amplitude with precisely two external legs offshell has the general form $\sum i V_j \cdot \frac{i}{s_e} \cdot iV_k$ where the internal edge $e$ crucially is not cancelled. That means, from the perspective of the metavertices $V_{j,k}$, this edge appears \enquote{onshell} in the sense that the terms $\propto s_e$ vanish in $V_{j,k}$. In other words, the connected tree-level amplitude has poles only arising from internal propagators as asserted in BCFW relations  \cite{britto_new_2005,britto_direct_2005}. However, in the present case the individual constituents $iV_j$ are not completely onshell because then they would vanish identically. But at least they behave like onshell amplitudes with regard to the internal edges.

\subsection{Content}

\Cref{sec:positionspace} gives an interpretation of field diffeomorphisms in position space and thereby a posteriori motivates the peculiar Feynman rules of the connected perspective, \cref{thm_feynmanrules}. 
In \cref{sec:twopoint}  we derive the connected 2-point-function in momentum space.
In \cref{sec_divergences}, we extend the connected perspective to also incorporate counterterm vertices. We compute several of these meta-counterterms and show how to use them to remove subdivergences. 
\Cref{sec:rec} concerns a class of diffeomorphisms, where $b_n = \lambda^{n-2}$, called exponential diffeomorphism. We compute the explicit functional form of all such diffeomorphisms in position- and momentum space as well as the 2-point-functions. Further, we demonstrate that for this class of diffeomorphisms, all connected amplitudes can be computed from 2-point functions.
In \cref{sec_1PI}, we extract 1PI counterterms from the previously computed meta-counterterms. In particular, we show that for the exponential diffeomorphism,   the one-loop counterterms are structurally equal to the bare vertices $iv_n$ and that this is not the case at two-loop order. Finally, we derive the 1PI counterterm of the 2-point-function to all orders in perturbation theory. 
In \cref{sec:st}, we show that diffeomorphism-invariance of the $S$-matrix implies infinitely many Slavnov-Taylor identities between the 1PI counterterms. 
In \cref{sec:mhv} we explore the structural similarities between a field diffeomorphism and a nonabelian gauge theory. We argue that the recursion relations which define the connected tree-level amplitudes in the scalar case are   the equivalent of Berends-Giele relations in  QCD.

\section{Correlation functions in position space}\label{sec:positionspace}

\subsection{Structure of correlation functions}\label{sec:positionspace_func}
The   transformed field $\rho(x)$ is related to the   underlying free field $\phi$ via the inverse diffeomorphism \cref{def_diffeomorphism_inverse},
\begin{align}\label{def_diffeomorphism_invers2}
\rho(x) &= \sum_{n=1}^\infty \frac{b_{n+1}}{n!}\phi^n(x).
\end{align}
This allows for a straightforward computation of the correlation functions of $\rho$ in position space in terms of the correlation functions of $\phi$ by expanding each field operator $\rho(x)$ according to \cref{def_diffeomorphism_invers2}.  
The 1-point function is just the expectation value of $\rho(x)$.  Wick's Theorem \cite{wick} resolves $\langle \phi^n(x)  \rangle$ into a sum of all possible products of position space propagators \cref{propagators}. Since there is only one spacetime point, they are $G_F(0)$, which corresponds to a position-space tadpole graph and is independent of momenta. We  assume that all tadpole-graphs vanish,
\begin{align}\label{pos_tadpole}
	G_F(0) &\overset != 0.
\end{align}
Consequently the 1-point-function vanishes as well $	\left \langle \rho(x) \right \rangle =0$.

For the 2-point function, the Wick expansion reads
\begin{align*}
\left \langle \rho(x) \rho(y)\right \rangle &= \sum_{t_1=1}^\infty \sum_{t_2=1}^\infty \frac{b_{t_1+1} b_{t_2+1}}{t_1!t_2!} \left \langle \phi^{t_1}(x) \phi^{t_2} (y) \right \rangle. 
\end{align*}
The right hand side are correlation functions of a free field $\phi$. By Wick's theorem and using \cref{pos_tadpole}, the terms contributing to the 2-point function consist of an arbitrary number of edges between the two spacetime points $x,y$. Especially, as they do not involve any other vertex than these two, they can be interpreted as Feynman diagrams on two external vertices,  see graph $A$ in \cref{fig_posspace}. There are $t_1!$ different Wick contractions for each summand and
\begin{align}\label{posspace_2point}
\left \langle \rho(x) \rho(y) \right \rangle &= \sum_{t_1=1}^\infty  \frac{b_{t_1+1}^2 }{t_1!t_1! } t_1!G_F^{t_1} (x-y) = \sum_{t=1}^\infty  \frac{b_{t+1}^2 }{t! } G_F^{t} (x-y) . 
\end{align}
Conversely, this position-space representation of the 2-point-function allows to compute massless multiedges via Fourier transform as done in \cite{bollini_dimensional_1996}.

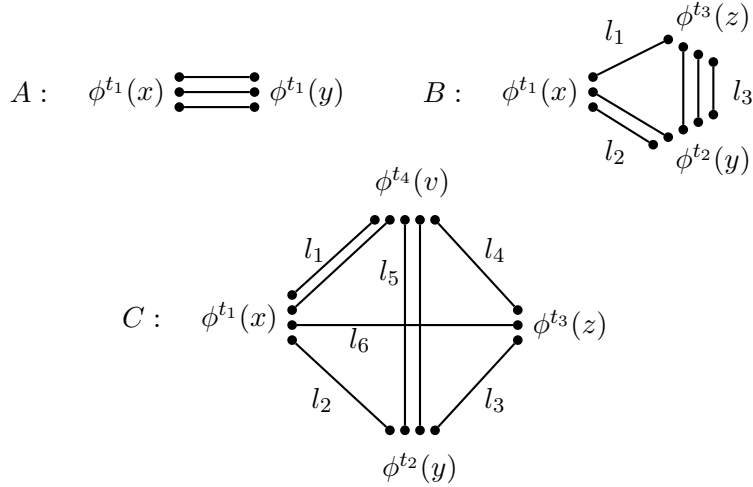
\begin{figure}[htbp]
	\centering
	\begin{tikzpicture}
	
	\node at (-1,0) {$A:$};
	
	\node at (.3,0){$\phi^{t_1}(x)$};
	
	\node [point] (x1) at (1,.2) {};
	\node [point] (x2) at (1,0) {};
	\node [point] (x3) at (1,-.2) {};
	
	\node at (2.7,0) {$\phi^{t_1}(y)$};
	\node [point] (y1) at (2,.2) {};
	\node [point] (y2) at (2,0) {};
	\node [point] (y3) at (2,-.2) {};
	
	\draw   (x1)--(y1);
	\draw   (x2)--(y2);
	\draw   (x3)--(y3);

	\node at (4.5,0) {$B:$};
	
	\node at (5.8,0) {$\phi^{t_1}(x)$};
	
	\node [point] (x1) at (6.5,-.2) {};
	\node [point] (x2) at (6.5,0) {};
	\node [point] (x3) at (6.5,.2) {};
	
	\node at (8.1,1) {$\phi^{t_3}(z)$};
	
	\node [point] (y1) at (7.5,.7) {};
	\node [point] (y2) at (7.7,.6) {};
	\node [point] (y3) at (7.9,.5) {};
	\node [point] (y4) at (8.1,.4) {};
	
	\node at (8.1,-.9) {$\phi^{t_2}(y)$};
	
	\node [point] (z1) at (7.3,-.7) {};
	\node [point] (z2) at (7.5,-.6) {};
	\node [point] (z3) at (7.7,-.5) {};
	\node [point] (z4) at (7.9,-.4) {};
	\node [point] (z5) at (8.1,-.3) {};
	
	\draw  (x3)--(y1);
	\draw  (x2)--(z2);
	\draw  (x1)--(z1);
	\draw  (y2)--(z3);
	\draw  (y3)--(z4);
	\draw  (y4)--(z5);
	
	\node at (6.8,-.8) {$l_2$};
	\node at (6.8,.8) {$l_1$};
	\node at (8.5,0) {$l_3$};

	\node at (.5,-3) {$C:$};
	
	\node at (1.8,-3) {$\phi^{t_1}(x)$};
	
	\node [point] (x1) at (2.5,-2.7) {};
	\node [point] (x2) at (2.5,-2.9) {};
	\node [point] (x3) at (2.5,-3.1) {};
	\node [point] (x4) at (2.5,-3.3) {};
	
	\node at (4.1,-1.2) {$\phi^{t_4}(v)$};
	
	\node [point] (y1) at (3.6,-1.7) {};
	\node [point] (y2) at (3.8,-1.7) {};
	\node [point] (y3) at (4,-1.7) {};
	\node [point] (y4) at (4.2,-1.7) {};
	\node [point] (y5) at (4.4,-1.7) {};
	
	\node at (6.2,-3.1) {$\phi^{t_3}(z)$};
	
	\node [point] (z1) at (5.5,-2.9) {};
	\node [point] (z2) at (5.5,-3.1) {};
	\node [point] (z3) at (5.5,-3.3) {};
	
	\node at (4.2,-5) {$\phi^{t_2}(y)$};
	
	\node [point] (v1) at (3.8,-4.5) {};
	\node [point] (v2) at (4,-4.5) {};
	\node [point] (v3) at (4.2,-4.5) {};
	\node [point] (v4) at (4.4,-4.5) {};
	
	\draw (x1) -- (y1);
	\draw (x2) -- (y2);
	\draw (y3) -- (v2);
	\draw (y4) -- (v3);
	\draw (y5) -- (z1);
	\draw (x3) -- (z2);
	\draw (x4) -- (v1);
	\draw (v4) -- (z3);
	
	\node at (2.8,-2.1) {$l_1$};
	\node at (2.9,-4.1) {$l_2$};
	\node at (5.2,-4.1) {$l_3$};
	\node at (5.2,-2.1) {$l_4$};
	\node at (3.8,-2.4) {$l_5$};
	\node at (3.4,-3.3) {$l_6$};
	
	\end{tikzpicture}
	\caption{Contributions to connected correlation functions in position space. Each dot represents a factor of $\phi$. Graphs where dots of the same monomial $\phi^j$ (i.e. same spacetime-point) are connected, are excluded due to \cref{pos_tadpole}. $A:$ 2-point function, $B:$ 3-point function, $C:$ 4-point function.}
	\label{fig_posspace}
\end{figure}

The 3-point function is sketched in graph $B$ in \cref{fig_posspace} and can be written as
\begin{align*}
\left \langle \rho(x) \rho(y) \rho(z) \right \rangle &= \sum_{t_1=1}^\infty \sum_{t_2=1}^\infty \sum_{t_3=1}^\infty \frac{b_{t_1+1} b_{t_2+1} b_{t_3+1}}{t_1! t_2! t_3!} \left \langle \phi^{t_1}(x) \phi^{t_2} (y) \phi^{t_3}(z)\right \rangle. 
\end{align*}
Let $l_1, l_2, l_3$ be the number of propagators between the points, then
\begin{align*}
&t_1 = l_1+l_2, \qquad t_2= l_2 + l_3, \qquad t_3= l_1 + l_3 \\
\Leftrightarrow \quad &l_1 = \frac 12 (t_1+t_3-t_2), \quad l_2 = \frac 12 (t_1+t_2-t_3), \quad l_3 = \frac 12 (t_2+t_3-t_1).
\end{align*}
We can rewrite the sums over $t_i$ in terms of $l_i$ where $l_i \geq 0$ under the condition that  $t_j\geq 1$.   After working out the combinatoric prefacotrs, the non-tadpole part of the  3-point function reads 
\begin{align*}
\left \langle \rho(  x)\rho(  y) \rho( z)\right \rangle &= \sum_{\stackrel{l_j\in \mathbbm N_0}{l_1+l_2+l_3\geq 2}}  \frac{b_{t_1+1}b_{t_2+1}b_{t_3+1}}{l_1! l_2! l_3! } G_F^{l_1}(  x-  y)^{l_2} G_F (  y-  z) G_F^{l_3} (  x- z).
\end{align*}
This structure continues in all $n$-point functions, see e.g. $C$ in \cref{fig_posspace}. For the 4-point-function, the relations between $t_j$ and $l_j$ read
\begin{align*}
t_1 &= l_1+l_2+l_6, \qquad t_2 = l_2+ l_3 + l_5 , \qquad t_3 = l_3 + l_4 + l_6, \qquad t_4 = l_1+l_4+l_5.
\end{align*}
Summing over all permutations, one then obtains
\begin{align}\label{pos_4point}
\left \langle \rho(x) \rho(y) \rho(z) \rho(v)\right \rangle &= \sum_{\stackrel{l_1, \ldots, l_6\in \mathbbm N_0}{t_j \geq 1\; \forall j}} \frac{b_{t_1+1} b_{t_2+1}  b_{t_3+1}  b_{t_4+1}}{l_1! l_2!l_3!l_4!l_5!l_6!} G_F^{l_1} (x-v) G_F^{l_2}(x-y) \cdots G_F^{l_6}(x-z).
\end{align}
The requirement $t_j\geq 1$ ensures that only complete contractions, that is terms without any disconnected spacetime point, appear. Consequently, we obtain a complete correlation function, not just a connected one. For example $l_1=1, l_3=1$ and all other $l_j$ zero is a valid contribution to \cref{pos_4point}, yet it is not a connected correlation function since it is actually a product of two 2-point-functions.

\begin{theorem}\label{thm_npoint_position}
	If tadpoles vanish, then the complete $n$-point amplitude in position space is 
	\begin{align*}
	\left \langle \rho(x_1) \cdots \rho(x_n)\right \rangle = \sum_{\stackrel{l_1, \ldots, l_k\in \mathbbm N_0}{t_j \geq 1 \; \forall j}} \frac{b_{t_1+1} \cdots b_{t_n+1}}{l_1! \cdots l_k!}G_F^{l_1}(x_1-x_2)G_F^{l_2}(x_1-x_3)\cdots G_F^{l_k}(x_{n-1}-x_n)
	\end{align*}
	where $k=\frac{n(n-1)}{2}$ is the number of ways to form pairs $x_i-x_j$.  $t_j$ are sums of $n-1$ indices $l_i$ labelling the $n-1$ edges incident to a vertex $j$ in a completely connected graph on $n$ vertices. Especially, each $l_i$ appears in precisely two distinct $t_j$ and each pair $\{t_i,t_j\}$ shares precisely one $l_i$.
\end{theorem}
\begin{proof}
	By Wick's theorem, the complete amplitude is a sum of all possible contractions. It remains to show that the combinatorial factors take the claimed form. 
	
	Consider a completely connected graph on $n$ vertices and choose an arbitrary labeling of its $\frac{n(n-1)}{2}=:k$   edges with non-negative integers $\{l_i\}$. Each vertex $j$ is incident to precisely $n-1$ edges, the sum of these edge labels defines the corresponding $t_j$. Each edge connects two vertices, so each edge label $l_i$ contributes to precisely two distinct $t_j$. Also, each pair of vertices $\{i,j\}$ is connected by precisely one edge, hence the corresponding $\{t_i,t_j\}$ have precisely one edge label $l_i$ in common.  
	
	The diffeomorphism \cref{def_diffeomorphism_invers2} implies $t_j\geq 1$ for each $j$, hence the summation over $\{l_i\}$ must ensure these conditions but is otherwise unconstrained and $l_i=0$ is allowed. Now fix some $j$ and examine the vertex $j$. It has connections to $(n-1)$ remaining vertices with, say, multiplicities $l_1, \ldots, l_{n-1}$. There are $\binom{t_j}{l_1}$ ways to choose the $l_1$ strands leading to the first neighbor vertex, what remains are $t_j-l_1$ unused points at $j$. So there are $\binom{t_j-l_1}{l_2}$ choices to connect to the second neighbor vertex and so on. The same holds true for any other vertex, so there arises an overall symmetry factor of
	\begin{align*}
	&\binom{t_1}{l_1}\binom{t_1-l_1}{l_2}\cdots \binom{l_{n-1}}{l_{n-1}} \cdot \binom{t_2}{l_1} \binom{t_2-l_1}{l_n}\cdots \binom{l_{2n-3}}{l_{2n-3}} \cdots  
	&=  \frac{t_1! \cdots t_n!}{l_1! l_1! l_2! l_2! \cdots l_k! l_k!}.
	\end{align*}
	This factor counts the possibilities to permute how vertices are connected, but not the additional possibilities to permute the $l_i$ internal lines of such a connection, which induces another factor $l_i!$ for each connection. Finally, each vertex $j$ comes with a factor $\frac{b_{t_j+1}}{t_j!}$ and hence the overall factor of the summands is as claimed.
	
	Each index $l_j$ counts the number of Wick contractions between the corresponding pair $\{i,j\}$ of vertices. This translates to a factor $G_F^{l_j}(x_i-x_j)$ in the summand.
	
\end{proof}

Note that the combinatoric prefactor is excactly what one would expect for Feynman graphs where a $j$-valent vertex has the amplitude $\frac{b_j}{j!}$. These are the ordinary Feynman amplitudes in position space. They do not involve integrals since integration would only be necessary for undeterminded inner points, not for loops as in momentum-space.

Note further that by $b_{t_j}$ depending on $(n-1)$ of the indices $l_j$, the $k$ sums in \cref{thm_npoint_position} are not independent from each other as long as the coefficients $\{b_{t_j}\}$ are unspecified.

\subsection{Interpretation}\label{sec:pos_interpretation}
The position-space correlation functions as computed from Wick's theorem in \cref{sec:positionspace_func} allow for a very transparent interpretation of the momentum-space Feynman rules of the connected perspective \cref{thm_feynmanrules}. Namely, the fact that the metavertices cancel an adjacent propagator in momentum space is equivalent to this vertex being an external (i.e. at the position of an argument of the $n$-point function), not an inner (i.e. at an unspecified position to be integrated over) vertex in position space. To see this, consider first the $n$-valent metavertex $iV_n$ in momentum space, including its adjacent propagators. This is supposed to be the tree-level contribution to the connected $n$-point amplitude. We set $s_1 \neq 0$ and $s_2 = \ldots = s_n=0$. 
\begin{align}\label{npoint_mom}
 \langle \rho(p_1) \cdots \rho(p_n)\rangle  &= -i b_n s_1 \frac{i^n}{s_1 \cdots s_n}  =  b_n \prod_{j=2}^n \frac{i}{s_j} .
\end{align}
Now consider the summand in the connected $n$-point amplitude in position space which is proportional to $b_n$. By \cref{thm_npoint_position}, it is the summand where all $b_{t_j+1}$ except one are 1 and consequently $t_j=1$  . The sum over $\{l_j\}$ then collapses to $n$ terms, namely, each $l_j$ is either 0 or 1 such that the sum of the $l_j$ is $n-1$. The term corresponding to \cref{npoint_mom} is the one where $t_1 = n-1$, 
\begin{align}\label{npoint_pos}
\langle \rho(x_1) \cdots \rho(x_n)\rangle  &= \frac{b_n}{1}\prod_{j=2}^n G_F^1 (x_j-x_1)  .
\end{align}
Such terms can be matched one by one to corresponding terms in \cref{npoint_mom}: The term with position-space propagators originating from one position $x_j$ corresponds to the term where the momentum-space propagator $\frac{i}{s_j}$ is cancelled, as can be seen by an explicit Fourier transform  of \cref{npoint_pos}, using the momentum-space propagators \cref{propagators}:
\begin{align*}
\langle \rho(p_1) \cdots \rho(p_n)\rangle &=\prod_{k=1}^n \int  \d^4 x_k   \; b_n \prod_{j=2}^n G_F(x_j-x_1)   \prod_{i=1}^n e^{ip_i x_k} =b_n \delta(p_1 + p_2 + \ldots + p_n) \prod_{j=2}^n  \frac{i}{s_j}    .
\end{align*}
The result equals \cref{npoint_mom} up to an overall delta function, which is not written in the momentum space functions by convention. The relation between both terms is represented graphically in the first row of \cref{fig_pos_mom}.

\begin{figure}[htbp]
	\centering
	\begin{tikzpicture}

	\node at (0,2.5) {position-space:};

\node [] at (-5,0) {$\propto b_5$};
	
	\node [] (c) at (0,0){};
	
	\node [label=left:$\phi^4(x_1)$] (x1) at ($ (c)+(180:1.5)$) {};
	\node [point] (x12) at ($ (x1)+(0,.3)$) {};
	\node [point] (x13) at ($ (x1)+(0,.1)$) {};
	\node [point] (x14) at ($ (x1)+(0,-.1)$) {};
	\node [point] (x15) at ($ (x1)+(0,-.3)$) {};
	
	\node [point,label=below:$\phi(x_5)$] (x5) at ($ (c)+(252:1.5)$) {};
	\node [point,label=right:$\phi(x_4)$] (x4) at ($ (c)+(324:1.5)$) {};
	\node [point,label=right:$\phi(x_3)$] (x3) at ($ (c)+(396:1.5)$) {};
	\node [point,label=above:$\phi(x_2)$] (x2) at ($ (c)+(468:1.5)$) {};
	
	\draw (x12) -- (x2);
	\draw (x13) -- (x3);
	\draw (x14) -- (x4);
	\draw (x15) -- (x5);

	\node at (6,2.5) {momentum-space:};
	\node[treeVertex] (c) at (6,0){};
	\draw [>-] (c) -- ++ (180:1) node [label=left:$p_1$]{};
	\draw [-] (c)-- ++(252:1) node [label=below:$p_5$]{};
	\draw [-] (c)-- ++(324:1) node [label=right:$p_4$]{};
	\draw [-] (c)-- ++(396:1) node [label=right:$p_3$]{};
	\draw [-] (c)-- ++(468:1) node [label=above:$p_2$]{};

	\node [] at (-5,-5) {$\propto b_3 b_4$};
	
	\node [] (c) at (0,-5){};
	
	\node [label=left:$\phi^3(x_1)$] (x1) at ($ (c)+(180:1.5)$) {};
	\node [point] (x12) at ($ (x1)+(0,.2)$) {};
	\node [point] (x13) at ($ (x1)+(0,0)$) {};
	\node [point] (x14) at ($ (x1)+(0,-.2)$) {};
	
	\node [label=above:$\phi^2(x_2)$] (x2) at ($ (c)+(108:1.5)$) {};
	\node [point] (x22) at ($ (x2)+(-.1,-.05)$) {};
	\node [point] (x23) at ($ (x2)+(.1,.05)$) {};
	
	\node [point,label=below:$\phi(x_5)$] (x5) at ($ (c)+(252:1.5)$) {};
	\node [point,label=right:$\phi(x_4)$] (x4) at ($ (c)+(324:1.5)$) {};
	\node [point,label=right:$\phi(x_3)$] (x3) at ($ (c)+(396:1.5)$) {};

	\draw (x12) -- (x22);
	\draw (x13) -- (x4);
	\draw (x14) -- (x5);
	\draw (x23) -- (x3);

	\node[treeVertex] (c1) at (6,-5.3){};
	\node[treeVertex] (c2) at (6.2,-4.6){};
	
	\draw[-|-] (c1)--(c2);
	
	\draw [>-] (c1) -- ++ (180:1) node [label=left:$p_1$]{};
	\draw [-] (c1)-- ++(252:1) node [label=below:$p_5$]{};
	\draw [-] (c1)-- ++(324:1) node [label=right:$p_4$]{};
	\draw [-] (c2)-- ++(396:1) node [label=right:$p_3$]{};
	\draw [>-] (c2)-- ++(468:1) node [label=above:$p_2$]{};

	\node [] at (-5,-10) {$\propto b_6 b_4b_3$};
	
	\node [] (c) at (0,-10){};
	
	\node [label=left:$\phi^5(x_1)$] (x1) at ($ (c)+(180:1.5)$) {};
	\node [point] (x12) at ($ (x1)+(0,.4)$) {};
	\node [point] (x13) at ($ (x1)+(0,.2)$) {};
	\node [point] (x14) at ($ (x1)+(0,0)$) {};
	\node [point] (x15) at ($ (x1)+(0,-.2)$) {};
	\node [point] (x16) at ($ (x1)+(0,-.4)$) {};
	
	\node [label=above:$\phi^3(x_2)$] (x2) at ($ (c)+(108:1.5)$) {};
	\node [point] (x22) at ($ (x2)+(-.2,-.05)$) {};
	\node [point] (x23) at ($ (x2)+(0,0)$) {};
	\node [point] (x24) at ($ (x2)+(.2,.05)$) {};
	
	\node [label=right:$\phi^2(x_3)$] (x3) at ($ (c)+(36:1.5)$) {};
	\node [point] (x32) at ($ (x3)+(-.05,0)$) {};
	\node [point] (x33) at ($ (x3)+(.05,-.2)$) {};
	
	\node [point,label=below:$\phi(x_5)$] (x5) at ($ (c)+(252:1.5)$) {};
	\node [point,label=right:$\phi(x_4)$] (x4) at ($ (c)+(324:1.5)$) {};

	\draw (x12) -- (x22);
	\draw (x13) -- (x23);
	\draw (x14) -- (x33);
	\draw (x15) -- (x4);
	\draw (x16) -- (x5);
	\draw (x24) -- (x32);

	\node[treeVertex] (c1) at (6,-10.5){};
	\node[treeVertex] (c2) at (6,-9.4){};
	\node[treeVertex] (c3) at (7,-10){};
	
	\draw [-,bend angle=30, bend left] (c1) to (c2);
	\draw [-, bend angle=30, bend right] (c1) to (c2);
	\draw [-] (c1) -- (c3);
	\draw [-] (c2) -- (c3);
	
	\draw [>-] (c1) -- ++ (180:.8) node [label=left:$p_1$]{};
	\draw [-] (c1)-- ++(252:.8) node [label=below:$p_5$]{};
	\draw [-] (c1)-- ++(324:.8) node [label=right:$p_4$]{};
	\draw [>-] (c3)-- ++(10:.8) node [label=right:$p_3$]{};
	\draw [>-] (c2)-- ++(110:.8) node [label=above:$p_2$]{};

	\end{tikzpicture}
	\caption{Correspondence between momentum-space and position-space Feynman rules for selected graphs of the 5-point function. If a vertex cancels an external propagator in momentum space, it is turned from an internal to an external vertex in position space. A momentum $p_i$ is the Fourier transform of a position $x_i$. The edge which is cancelled by a metavertex is marked with an arrowhead.}
	\label{fig_pos_mom}
\end{figure}
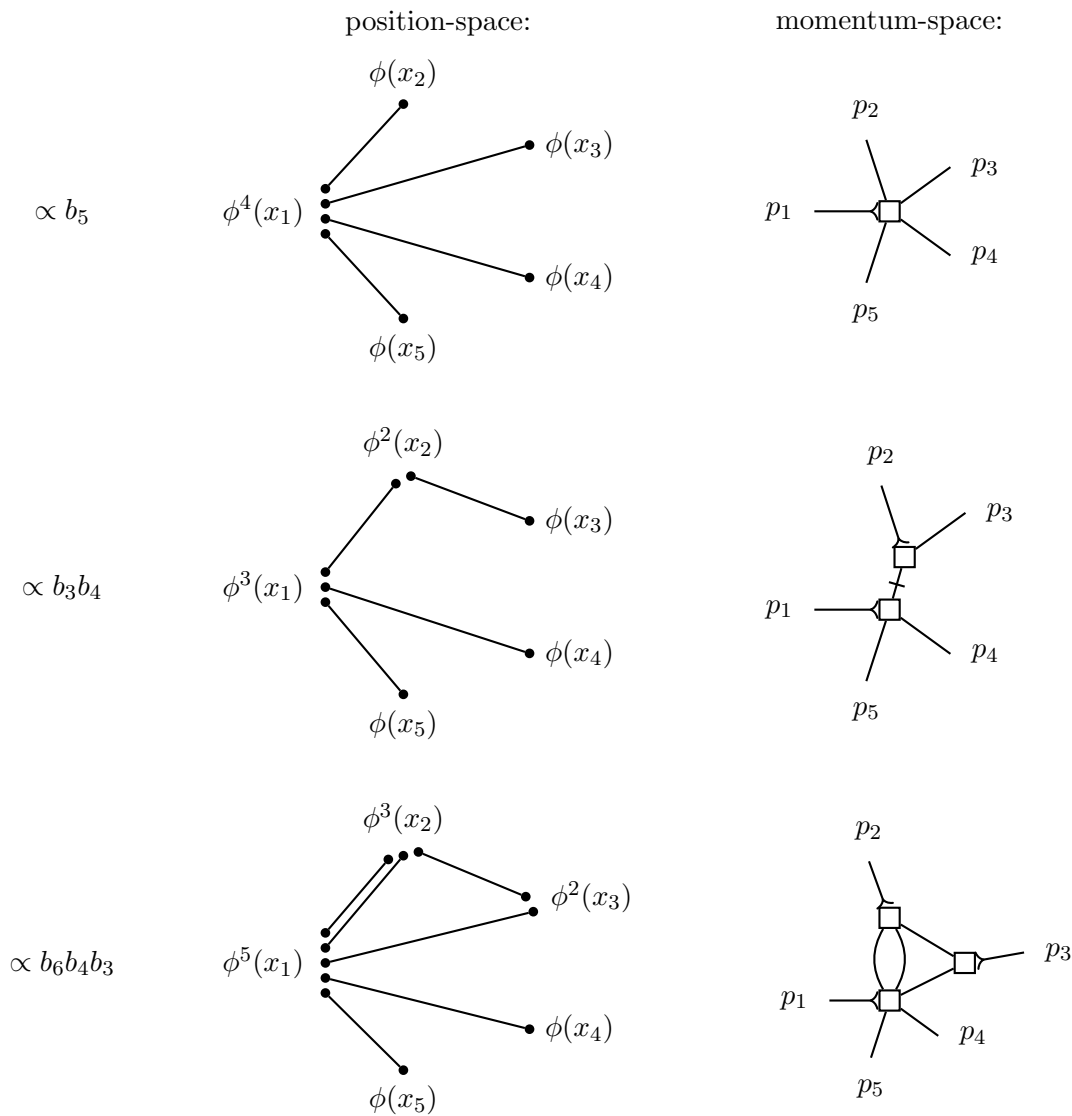

Their counterpart in position-space makes it completely obvious why the momentum-space Feynman rules in the connected perspective \cref{thm_feynmanrules} allow for maximal $n$ metavertices in a $n$-point amplitude: There are not more than $n$ points in position-space which could possibly be turned into metavertices. This argument can even be applied to loop graphs in momentum-space as shown in the bottom row in \cref{fig_pos_mom}. If a momentum-space amplitude contains loops, it involves integrals over internal momenta. All integrations in momentum-space are leftovers of the Fourier transform, their ultraviolet divergence corresponds to two of the  external  coordinates approacing each other.

\section{The 2-point function in momentum space}\label{sec:twopoint}

Using the Feynman rules of the connected perspective \cref{thm_feynmanrules}, the connected amplitude with two external momenta is supported on Feynman graphs with up to two metavertices, both of which are external. Excluding Tadpoles, the only remaining graph topology is that of $l$-loop multiedges $M^{(l)}$   where the vertices are $(l+2)$-valent metavertices $iV_{l+2} = -ib_{l+2} \cdot s$ cancelling the external propagator $\frac i s$, see \cref{fig_G2}. The two external propagators are not included. 
\begin{align}\label{G2}
G_2 (s)&=    -is  + \sum_{l=1}^\infty (-ib_{l+2}s)^2\frac{ M^{(l)}(s)}{(l+1)!}  =: -is \left(1+ G_2^{\text{fin}}(s) +  G_2^{\text{div}}(s) \right) + \mathcal O \left( \epsilon  \right).
\end{align}
Here we have separated the singular part of $G_2(s)$ where $\renop$ is defined in \cref{def_R}:
\begin{align}\label{G2div}
	G_2^{\text{div}}(s) &:=   \frac{i}{s}\renopk{G_2(s)}= -is \sum_{l=1}^\infty b_{l+2}^2\frac{M^{(l)}_{\text{div}} (s)}{(l+1)!} . 
\end{align}

\begin{figure}[htbp]
	\centering
	\begin{tikzpicture}

		\node at (-3.5,0) {$G_2(s) $};
		\node[anchor=west] at (-2.7,0) {$=$};
		
		\node (c1) at (-1.7,0) [] {};
		\node (c2) at (-1, 0) []{}; 
		\draw (c1) -- (c2);
		
		\node [anchor=west] at (-.6,0) {$+$};
		
		\node [treeVertex]  (c1) at (1,0) {};
		\node [treeVertex]  (c2) at (2.5,0) {};
		
		\draw [>-] (c1) -- + (180:.5);
		\draw [bend angle =30, bend left] (c1) to (c2);
		\draw [bend angle =30, bend right] (c1) to (c2);
		\draw [>-] (c2) -- + (0:.5);
		
		\node [anchor=west] at (3.5,0) {$+$};
		
		\node [treeVertex]  (c1) at (5,0) {};
		\node [treeVertex]  (c2) at (6.5,0) {};
		
		\draw [>-] (c1) -- + (180:.5);
		\draw [bend angle =45, bend left] (c1) to (c2);
		\draw  (c1) to (c2);
		\draw [bend angle =45, bend right ] (c1) to (c2);
		\draw [>-] (c2) -- + (0:.5);
		
		\node [anchor=west]at (7.6,0) {$+\quad \ldots$};
		
		\node [anchor = west] at (-2.7,-1){$=$};
		
		\node at (-1.5, -1) {$-is$};
		
		\node [anchor=west] at (-.6,-1) {$+$};
		
		\node at (1.8, -1) {$ (-ib_3 s)^2 \frac 1{2!} M^{(1)}(s)  $};
		
		\node [anchor=west]at (3.5,-1) {$+$};
		
		\node at (5.8, -1) {$ (-ib_4 s)^2 \frac 1{3!}M^{(2)}(s) $};
		
		\node [anchor=west]at (7.6,-1) {$+\quad \ldots$};

	\end{tikzpicture}
	\caption{The amputated connected two-point-amplitude $G_2$ in momentum-space with external momentum $s:= p^2$ in the connected perspective. Only multiedges contribute. The two external propagators are not included. Both metavertices cancel the external edges as indicated with arrowheads.}
	\label{fig_G2}
\end{figure}

\begin{example}[Massless theory]\label{ex_G2_massless}
	In all explicit examples, we consider the massless theory $s_p=p^2$ in $D=4-2 \epsilon$ dimensions. Using \cref{lem_Ml},
	\begin{align*}
		G_2(s) &= -is- \sum_{l=1}^\infty \frac{b_{l+2}^2}{(l+1)!} \frac{\left(- i  s \right) ^{l+1}}{(4\pi)^{2l}  \left( l! \right) ^2} \left(  \frac{1}{\epsilon }  -1 + (2l+1) H_l -l\gamma_E -l \ln s + l \ln (2\pi)   \right) +\mathcal O \left( \epsilon  \right).  
	\end{align*}
	The two constituents are 
	\begin{align}
		G_2^{\text{fin}} (s) &=   -\sum_{l=1}^\infty \frac{b_{l+2}^2}{(l+1)!} \frac{\left(- i  s \right) ^{l}}{  (4\pi)^{2l}\left( l! \right) ^2} \left(  (2l+1) H_l-1 -l\gamma_E + l \ln (4\pi) -l \ln s  \right) , \label{G2fin_massless} \\
		G_2^{\text{div}}(s) &=    -\sum_{l=1}^\infty \frac{b_{l+2}^2}{(l+1)!} \frac{\left(- i   s \right) ^{l}}{ (4\pi)^{2l} \left( l! \right) ^2} \frac 1 \epsilon  . \label{G2div_massless}
	\end{align}
	As announced in \cref{sec:pos_interpretation}, the divergent part $G_2^{\text{div}}$ corresponds to the coincidence of points in position space, i.e. $x=0$, as can be verified by compuing the Fourier-transform of the summands in \cref{G2div}, using e.g.\cite[pp. 155,163]{kanwal_generalized_2004}, which produces terms $\propto (i\square)^l \delta(x)$.
\end{example}

The following comments are in order:
\begin{itemize}
	\item Even if the multiedges $M^{(l)}$ are technically 1PI graphs, $G_2(s)$ is indeed  the amputated \emph{connected}, not the 1PI 2-point-function. The latter requires systematic use of 1PI counterterms and will be computed in \cref{sec_1PI_2point}.
	\item In the onshell limit $s\rightarrow 0$, $G_2(s)$ reproduces the free two-point function. The $S$-matrix is unaltered as claimed in  \cite{balduf_perturbation_2020}.
	\item One of $G^{\text{fin}}_2$ and $G^{\text{div}}_2$  can be chosen  freely by picking suitable coefficients $b_n$. If one of those functions is fixed, the other one is, too. Especially, to have $G_2 \equiv -is$ for all $s$ (i.e. a free 2-point function), one has to set all $b_j=0$ and consequently \emph{all} $n$-point-functions are free, such a theory is a free theory altogether. This is in accordance with the Jost-Schroer-Federbush-Johnson-Theorem  \cite{Jost-Schroer,federbush_uniqueness_1960,Pohlmeyer}.
	\item One could be tempted to choose $b_n \propto \sqrt{\epsilon}$ such that $  G_2^{\text{div}}(s)$ becomes a regular function, but this choice does not remove the divergences of higher $n$-point functions. It seems impossible to render the theory finite by choosing an \enquote{infinitesimal} diffeomorphism.
\end{itemize}

At least in the massless case, some higher correlation functions can in principle be computed  but the explicit results are not of immediate interest for the following sections. For illustration, we consider the 3-point-function in \cref{sec:triangles}. Here, we merely note that starting from two loop order, the 3-point and all higher correlation functions contain non-lokal divergences such as $ \frac{1}{\epsilon} \ln s_1$ in \cref{G33}. This indicates that, to remove these subdivergences, we have to extend the connected perspective in a way to also include counterterm vertices, which will be done in the following section.

\section{Divergences in the connected perspective}\label{sec_divergences}

In this chapter, we extend the formalism of the connected perspective (\cref{thm_feynmanrules}) to incorporate even counterterms, eventually giving rise to finite amplitudes. To this end, we define \enquote{meta-counterterms} $C_k$ which share the combinatorial properties of metavertices: they absorb all possible internal cancellations and appear in graphs without changing the graph topology. 

It turns out that the metacounterterms can be classified by three numbers $j,k,l$ corresponding to the graphs they arise from, namely $C_{n,k}^{(l)}$ cancels the superficial divergence of graphs with 
\begin{itemize}
	\item $n$ external legs
	\item $k \leq n$ external legs offshell (which implies precisely $k$ metavertices)
	\item $l$ loops. 
\end{itemize}
For fixed $n$, the $l$-grading also implies an overall order in momentum. To impose $k\leq n$, we definie
\begin{align*}
C_{n, k >n}^{(l)} &:=0 \qquad \forall n,k,l
\end{align*} 
since there are no graphs and consequently no divergences with $k>n$, i.e. more metavertices than external legs. All counterterms are in the minimal subtraction (MS) scheme in dimensional regularization, but remember that we demand vanishing of tadpoles. In massive theories, where tadpoles do not vanish automatically, they produce an additional contribution to the counterterms which we do not include explicitly.

In the connected perspective, by \cref{thm_feynmanrules} a metavertex must not cancel internal edges . The same restriction applies to meta-counterterms with the only difference that a meta-counterterm can cancel more than one of its adjacent edges.

\begin{theorem}\label{thm_metacounterterms}	
		In the connected perspective of a free field diffeomorphism, all amplitudes can be made finite if meta-counterterms are included according to the following rules:
	\begin{enumerate}
		\item  Proceed according to the BPHZ renormalization prescription, recursively replacing divergent subgraphs $\gamma \subset \Gamma$ by their corresponding meta-counterterm which subtracts the divergence. Finally, remove the superficial divergence.
		\item The meta-counterterm $C_{n,k}$ is inserted in place of a graph on $k$ metavertices, it cancels $k$ out of its $n$ adjacent edges simultaneously. 
		\item No internal edge of a graph must be cancelled, neither by a metavertex nor by a meta-counterterm.
		\item There are neither internal metavertices nor internal meta-counterterms.
	\end{enumerate}
	
\end{theorem}
Like in \cref{thm_feynmanrules}, Point (4) is implied by the first three and kept for clarity
\begin{proof}
	For the general procedure to render Feynman integrals finite, we can ignore the physical interpretation of the connected perspective and just apply   the well-known procedure of BPHZ  renormalization \cite{bogoliubow_ueber_1957,hepp_proof_1966,zimmermann_convergence_1969,weinberg_highenergy_1960,kreimer_hopf_1998,kreimer_combinatorics_2002}. That is, in an $l$-loop amplitude one first has to subtract all subdivergences which arise from graphs with less than $l$ loops and finally the superficial divergence, which is guaranteed to be local. 
	
	To be shown is the compatibility of this procedure with the combinatorial restrictions of the connected perspective, i.e. that there is an 1:1 correspondence between divergent subgraphs and meta-counterterms which are inserted as stated in the theorem.
	
	$\Rightarrow$: If $\gamma \subset \Gamma$ is a divergent subgraph with $n$ legs and $k$ metavertices, then all $k$ metavertices are external in $\Gamma$ and $\gamma$ is adjacent to at least $k$ external legs of $\Gamma$. At the same position a meta-counterterm $C_{n,k}$ can be inserted since at least $k$ external edges are available to be cancelled.
	
	$\Leftarrow$: A meta-counterterm $C_{n,k}$ is, amongst others, a sum over all different ways to assign the $k$ cancellations to its $n\geq k$ legs. If a leg $e$ is cancelled, then it is adjacent to a metavertex of the graph whose divergence contributes to $C_{n,k}$. Consequently, when inserted into a graph $\Gamma$, a meta-counterterm is a sum of counterterms, some of which would cancel internal edges of $\Gamma$. By the theorem, we must not include the terms which cancel internal edges. This means that only those terms in $C_{n,k}$ remain which are divergences of allowed graphs in the connected perspective. 
	
	Finally,    if $\Gamma$ is a divergent graph with $n$ external edges and $k$ metavertices and superficial degree of divergence greater than zero, then the amplitude of $\Gamma$ contains some new positive power of momenta $s^j$ which did not arise from metavertices.
	One might think that this factor invalidates the above discussion since it is in principle arbitrary which momentum is chosen as a scale variable, thus introducing an arbitrary edge cancellation. 
	
	This ist not so, because in fact $s$ can not be choosen freely amongst the external momenta of $\gamma$. It can only depend on the \emph{total} momenta entering the graph $\gamma$ at metavertices. Let $V$ be a metavertex adjacent to $j>1$ external legs $s_1, \ldots, s_j$. Then $s$ can only depend on $s_{1+\ldots + j}$, not on any of the individual momenta. But a factor of $s_{1+\ldots + j}$ does not cancel any adjacent propagator, therefore it does not influence the above discussion. Now let $V$ be a metavertex adjacent to only a single external leg $s_1$. Then it is possible to choose $s=\propto s_1$ and $s$ cancels the adjacent edge. But that edge is cancelled anyway since $V$ needs to cancel one of its external legs by \cref{thm_feynmanrules}. So again, the overall momentum scale can not introduce any new cancellations.  
\end{proof}

\begin{example} (3-point-function)\label{ex_subdivergence}
	The requirement of not cancelling internal edges automatically selects the correct parts of the meta-counterterms.  Consider the three-loop graph $\Gamma_1$ shown in \cref{fig_ren_3point_gamma}. It has a quadratic subdivergence $\gamma\subset \Gamma_1$. This subdivergence is removed by the counterterm graph $\tilde \Gamma$ where a meta-counterterm $C_{4,2}^{(2)}$ is inserted into the cograph $\frac{\Gamma_1}{\gamma}$. On the other hand, the graph $\Gamma_2$ amounts to a different orientation of $\gamma$ in the same cograph. However, $\Gamma_2$ is not present in the connected perspective since it has an internal metavertex. This restriction is automatically respected by the meta-counterterm $C_{4,2}^{(2)}$: When cancelling two edges, only those graphs contribute to $C_{4,2}^{(2)}$ where said edges are incident to two distinct metavertices, see \cref{fig_ren_4point_2loop}. If we label the edges of $C_{4,2}^{(2)}$ as $1,2,3,4$, then the graphs shown in \cref{fig_ren_4point_2loop} are $\propto s_{1+2}$ or permutations thereof, but not $\propto s_{j}$ where $j\in \left \lbrace 1,2,3,4 \right \rbrace $. The only cancellations stem from the metavertices.
\end{example}

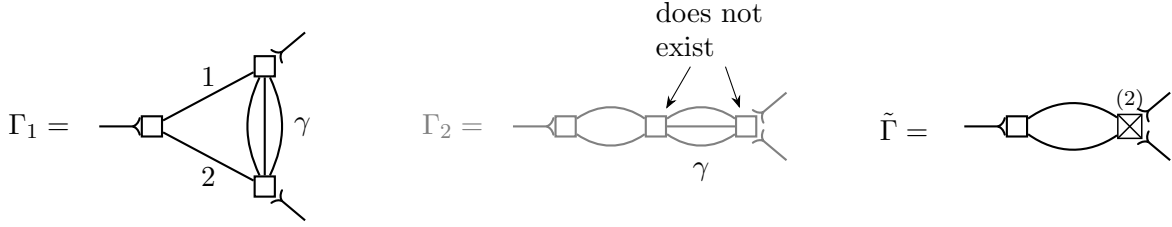
\begin{figure}[htbp] 
	\centering
	\begin{tikzpicture}

		\node at (.5,0) {$\Gamma_1  = $};
		\node [treeVertex] (c1) at (2,0){};
		\node [treeVertex] (c2) at (3.5,.8){};
		\node [treeVertex] (c3) at (3.5,-.8){};
		\draw (c1) --  (c2) node[pos=.5,above]{ 1};
		\draw (c1) -- (c3) node[pos=.5,below]{ 2};
		\draw [bend angle = 25, bend left](c2) to (c3);
		\draw (c2) -- (c3);
		\draw [bend angle = 25, bend right](c2) to (c3);
		\draw [>-](c1) -- + (180:.7);
		\draw [>-] (c2) -- + (40:.7);
		\draw [>-] (c3) -- + (-40:.7);

		\node   at (4,0) {$\gamma$};
		
		\node [gray] at (6,0) {$\Gamma_2  = $};
		\node [gray, treeVertex] (c1) at (7.5,0){};
		\node [gray, treeVertex] (c2) at (8.7,0){};
		\node [gray, treeVertex] (c3) at (9.9,0){};
		\draw [gray, bend angle = 35, bend left](c1) to (c2);
		\draw [gray, bend angle = 35, bend right](c1) to (c2);
		\draw [gray, bend angle = 35, bend left](c2) to (c3);
		\draw [gray](c2) -- (c3);
		\draw [gray, bend angle = 35, bend right](c2) to (c3);
		\draw [gray, >-](c1) -- + (180:.7);
		\draw [gray, >-] (c3) -- + (40:.7);
		\draw [gray, >-] (c3) -- + (-40:.7);

			\node [text width=1.5cm](tx) at ($(c2)+(60:1.5)$) {does not  exist};
		\draw [thin, -Stealth, bend angle = 20,shorten >=1mm ] (tx) to (c2);
		\draw [thin, -Stealth, bend angle = 20,shorten >=1mm ] (tx) to (c3);
		
		\node   at (9.3,-.6) {$\gamma$};
		
		\node at (12,0) {$\tilde \Gamma  = $};
		\node [treeVertex] (c1) at (13.5,0){};
		\node [treeCounterVertex, label={[label distance=-1mm]90:{$\scriptstyle (2)$} }] (c2) at (15,0){};
		\draw [bend angle = 35, bend left](c1) to (c2);
		\draw [bend angle = 35, bend right](c1) to (c2);
		\draw [>-](c1) -- + (180:.7);
		\draw [>-] (c2) -- + (40:.7);
		\draw [>-] (c2) -- + (-40:.7);

	\end{tikzpicture}
	\caption{Three-loop contributions to the 3-point function. They are two different ways to insert the divergent subgraph $\gamma$, only $\Gamma_1$ contributes in the connected perspective. The meta-counterterm is drawn as a crossed out metavertex. }
	\label{fig_ren_3point_gamma}
\end{figure}

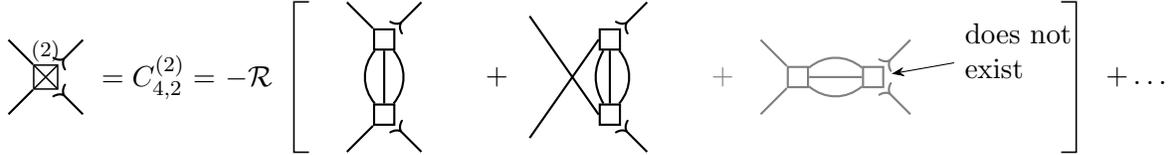
\begin{figure}[htbp] 
	\centering
	\begin{tikzpicture}
		
		\node [treeCounterVertex, label={[label distance=-1mm]90:{$\scriptstyle (2)$} }] (c) at (-1.5,0) {};
		\draw  [>-](c) --++ (45:.7);
		\draw  (c) --++ (135:.7);
		\draw  (c) --++ (225:.7);
		\draw [>-] (c) --++ (315:.7);
		
		\node at (0,0) {$=C_{4,2}^{(2)}=$};
		
		\node at (1.2,0) {$-\renop  $};
		
		\draw [semithick]  (2,-1) to [ncbar=.1](2,1);

		\node [treeVertex] (c1) at (3,-.5){};
		\node [treeVertex] (c2) at (3,.5){};
		\draw [bend angle =40, bend left] (c1) to (c2);
		\draw (c1) -- (c2);
		\draw [bend angle =40, bend right] (c1) to (c2);
		\draw (c2) -- + (135:.7);
		\draw [>-] (c2) -- + (45:.7);
		\draw [>-] (c1) -- + (-45:.7);
		\draw (c1) -- + (-135:.7);
		
		\node at (4.5,0) {$+$};
		
		\node [treeVertex] (c1) at (6,-.5){};
		\node [treeVertex] (c2) at (6,.5){};
		\draw [bend angle =30, bend left] (c1) to (c2);
		\draw (c1) -- (c2);
		\draw [bend angle =40, bend right] (c1) to (c2);
		\draw [bend angle =30 ] (c1.180) to + (125:1.6);
		\draw [>-] (c2) -- + (45:.7);
		\draw [>-] (c1) -- + (-45:.7);
		\draw [ bend angle = 30 ] (c2.180) to + (-125:1.6);
	
		\node [gray] at (7.5,0) {$+$};
		
		\node [gray, treeVertex] (c1) at (8.5,0){};
		\node [gray, treeVertex] (c2) at (9.5,0){};
		\draw [gray, bend angle =40, bend left] (c1) to (c2);
		\draw [gray] (c1) -- (c2);
		\draw [gray, bend angle =40, bend right] (c1) to (c2);
		\draw  [gray](c1) -- + (135:.7);
		\draw  [gray](c1) -- + (-135:.7);
		\draw  [gray,>-](c2) -- + (45:.7);
		\draw  [gray,>-](c2) -- + (-45:.7);
		
		\node [text width=1.5cm](tx) at ($(c2)+(10:2)$) {does not  exist};
		\draw [thin, -Stealth, bend angle = 20,shorten >=1mm ] (tx) to (c2);
		
		\draw [semithick] (12,-1) to [ncbar=-.1] (12,1);
		
		\node [] at (13,0) {$+\ldots$};

	\end{tikzpicture}
	\caption{Relevant part of   $C^{(2)}_{4,2}$. Indicated by arrows, the rightmost graph involves a metavertex which cancels two adjacent edges simultaneously, and another one which cancels no adjacent edge. By \cref{thm_feynmanrules}, such metavertices do not exist. Consequently, this graph does not contribute to $C_{4,2}^{(2)}$.}
	\label{fig_ren_4point_2loop}
\end{figure}

\begin{lemma}\label{lem_subdivergences}
	In the graphs contributing to the connected $n$-point amplitude, all possible subdivergences can be removed with meta-counterterms $C_{m,k}$ where $ k<n$. 
\end{lemma}
\begin{proof}
	Follows from \cref{thm_feynmanrules,thm_metacounterterms}: The connected $n$-point amplitude is supported on graphs $\Gamma$ with at most $n$ metavertices. A subdivergence is given by a subgraph $\gamma$  with $k$ vertices where $k \leq n$. If $k=n$, then the cograph $\frac \Gamma \gamma$ contains just a single metavertex and is a tadpole which is assumed to vanish. Hence $k<n$. The superficial divergence of a graph with $k$ metavertices is subtracted in $C_{m,k}$ where $m\geq k$ is the number of legs. If $\gamma$ is primitive, this finishes the proof, otherwise proceed inductively, using meta-counterterms $C_{m_2,k_2}$ where $k_2 < k$.
\end{proof}

\subsection{k=0 and k=1 legs offshell}
If all external legs are onshell, i.e. $k=0$, then the amplitudes of the connected perspective vanish, consequently there is no divergence and
\begin{align}\label{Cn0}
C_{n,0}^{(l)} &= 0 \qquad \forall n,l \qquad \Rightarrow \quad C_{n,0} =0.
\end{align}
If only one external leg is offshell, the amplitude is supported on graphs with a single metavertex. Such graphs are tadpoles and we assume them to vanish.  We therefore have
\begin{align}\label{Cn1}
C_{n,1}^{(l)} &= 0 \qquad \forall n,l \qquad \Rightarrow \quad C_{n,1} =0.
\end{align}
Graphically, these two identities are shown in \cref{Cn1_picture}.

\begin{figure}[htbp]
	\centering
	\begin{tikzpicture}

	\node [treeCounterVertex] (c) at (-.5,0) { };
	\draw [-|] (c) --++ (0:.5);
	\draw [-|] (c) --++ (50:.5);
	\draw [-|] (c) --++ (100:.5);
	\draw [-|] (c) --++ (-50:.5);
	\draw [-|] (c) --++ (-100:.5);
	\node at (-1,.1){$\vdots$};

	\node at (.8,0) {$=0$};

	\node [treeCounterVertex] (c) at (4,0) { };
	\draw [>-] (c) --++ (0:.5);
	\draw [-|] (c) --++ (50:.5);
	\draw [-|] (c) --++ (100:.5);
	\draw [-|] (c) --++ (-50:.5);
	\draw [-|] (c) --++ (-100:.5);
	\node at (3.5,.1){$\vdots$};
	
	\node at (5.3,0) {$=0$};

	\end{tikzpicture}
	\caption{The meta-counterterms $C_{n,k}$ for connected amplitudes vanish identically if $k=0$ or $k=1$, i.e. zero or one external leg is offshell.}
	\label{Cn1_picture}
\end{figure}

\subsection{n=2 legs}\label{sec_n2}

The two-point-function $n=2$ is supported on $l$-loop multiedge graphs $M^{(l)}(s)$. 
Since tadpoles are assumed to vanish, these graphs have no subdivergences. Correspondingly, no other meta-counterterms $C_{n,k>2}$ are necessary as asserted by \cref{lem_subdivergences}. The $l$-loop meta-counterterm for the 2-point-function is just the divergent part of $-M^{(l)}$,
\begin{align}\label{C22l}
C_{2,2}^{(l)}(s) &= -(-i b_{l+2}s)^2  \frac{M^{(l)}_{\text{div}}(p^2)}{(l+1)!} = b_{l+2}^2 s^2  \frac{M^{(l)}_{\text{div}}(p^2)}{(l+1)!}, 
\end{align}
and the all-order counterterm consequently is
\begin{align*}
C_{2,2} (s) &:= \sum_{l=1}^\infty C_{2,2}^{(l)}= - (-is)^2 \sum_{l=1}^\infty   \frac{b_{l+2}^2 }{(l+1)!}M^{(l)}_{\text{div}}(p^2).
\end{align*}

\begin{example}[Massless theory]\label{ex_C22}
	In the massless theory in $D=4-2\epsilon$ dimensions, the divergences of multiedges are given by \cref{Mldiv} and therefore
	\begin{align*}
	C_{2,2} (s) &= + \sum_{l=1}^\infty \frac{b_{l+2}^2}{(l+1)!} \frac{\left(-is\right) ^{l+1}}{ (4\pi)^{2l} \left( l! \right) ^2} \frac 1 \epsilon.
	\end{align*}
	This of course coincides with $- (-is)G_2^{\text{div}} (s) $ from \cref{G2div_massless}. Adding this counterterm, the connected two-point-function \cref{G2} is finite:
		\begin{align}\label{G2R}
	G^R_2 (p^2) &:= G_2(p^2) +C_{2,2}(p^2)=-is \left(1- G_2^{\text{fin}}(s) \right).
	\end{align}
\end{example}

\subsection{ k=2 legs offshell}
 
A meta-counterterm with $k=2$ of its legs offshell represents the superficial divergence of a graph on 2 metavertices i.e. a multiedge.

\begin{figure}[htbp]
	\centering
	\begin{tikzpicture}

	\node [treeCounterVertex, label={[label distance=-2mm]5:{$\scriptstyle (l)$}}  ] (c) at (-.5,0) { };
	\draw  (c) --++ (0:.5);
	\draw [>-] (c) --++ (-120:.5);
	\draw [>-] (c) --++ (120:.5);

	\node at (.5,0) {$=$};
	
	\node at (1.2,0) {$-\renop  $};
	
	\draw [semithick]  (2,-.8) to [ncbar=.1](2,.8);
	 
	\node [treeVertex] (c1) at (3,-.4){};
	\node [treeVertex, label={[label distance=1mm]180:{$\scriptstyle (l)$}}  ] (c2) at ($(c1)+ (60:.8)$){};
	\draw [bend angle =50, bend left] (c1) to (c2);
	\draw [bend angle =30, bend left] (c1) to (c2);
	\draw [bend angle =50, bend right] (c1) to (c2);
	\draw [>-](c1) -- + (-120:.5);
	\draw (c2) -- + (0:.5);
	\draw [>-](c2) -- + (120:.5);
	
	\node at (4.5,0) {$+$};
	
	\node [treeVertex, label={[label distance=1mm]0:{$\scriptstyle (l)$}}] (c1) at (5.5,.4){};
	\node [treeVertex] (c2) at ($(c1) + (-60:.8)$){};
	\draw [bend angle =50, bend left] (c1) to (c2);
	\draw [bend angle =30, bend left] (c1) to (c2);
	\draw [bend angle =50, bend right] (c1) to (c2);
	
	\draw [>-](c1) -- + (120:.5);
	\draw (c2) -- + (0:.5);
	\draw [>-](c2) -- + (-120:.5);

	\draw [semithick] (7,-.8) to [ncbar=-.1] (7,.8);

	\end{tikzpicture}
	\caption{Meta-counterterm $C_{3,2}^{(l)}$  according to \cref{C32l}. For the indicated orientation of cancelled edges, only two graphs contribute. }
	\label{C32_picture}
\end{figure}
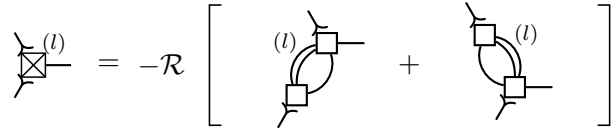

For graphs on $k=2$ metavertices, but with $n>2$ external edges, there are multiple orientations of the 2-vertex-multiedge. With $n=3$ external edges, one of the metavertices is adjacent to one external edge and the other one to the remaining two, see \cref{C32_picture}, and there are three ways to choose which two edges are offshell. The $l$-loop meta-counterterm reads
\begin{align}\label{C32l}
C_{3,2}^{(l)} &=\frac{b_{l+2} b_{l+3}}{(l+1)!} \left( s_1 (s_2+s_3) M^{(l)}_{\text{div}}(s_1) + s_2 (s_1+s_3) M^{(l)}_{\text{div}}(s_2)+s_3 (s_1+s_2) M^{(l)}_{\text{div}}(s_3) \right)  .
\end{align}

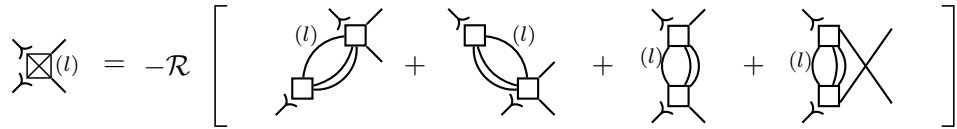
\begin{figure}[htbp]
	\centering
	\begin{tikzpicture}

	\node [treeCounterVertex, label={[label distance=-1mm]0:{$\scriptstyle (l)$}}  ] (c) at (-.5,0) {};
	\draw   (c) --++ (45:.5);
	\draw   (c) --++ (-45:.5);
	\draw  [>-](c) --++ (135:.5);
	\draw  [>-](c) --++ (-135:.5);

	\node at (.5,0) {$=$};
	
	\node at (1.2,0) {$-\renop  $};
	
	\draw [semithick]  (2,-.8) to [ncbar=.1](2,.8);

	\node [treeVertex] (c1) at (3,-.3){};
	\node [treeVertex, label={[label distance=2mm]180:{$\scriptstyle (l)$}}] (c2) at ($(c1) + (45:1)$){};
	\draw [bend angle =40, bend left] (c1) to (c2);
	\draw [bend angle =30, bend right] (c1) to (c2);
	\draw [bend angle =50, bend right] (c1) to (c2);
	\draw (c2) -- + (45:.5);
	\draw (c2) -- + (-45:.5);
	\draw [>-] (c2) -- + (135:.5);
	\draw [>-] (c1) -- + (-135:.5);
	
	\node at (4.5,0) {$+$};
	
		\node [treeVertex] (c1) at (6,-.3){};
	\node [treeVertex, label={[label distance=2mm]00:{$\scriptstyle (l)$}}] (c2) at ($(c1) + (135:1)$){};
	\draw [bend angle =40, bend right] (c1) to (c2);
	\draw [bend angle =30, bend left] (c1) to (c2);
	\draw [bend angle =50, bend left] (c1) to (c2);
	\draw (c1) -- + (45:.5);
	\draw (c1) -- + (-45:.5);
	\draw [>-] (c1) -- + (-135:.5);
	\draw [>-] (c2) -- + (135:.5);
	
		\node at (7,0) {$+$};
	
	\node [treeVertex] (c1) at (8,-.4){};
	\node [treeVertex, label={[label distance=-1mm]190:{$\scriptstyle (l)$}}  ] (c2) at ($(c1)+ (90:.8)$){};
	\draw [bend angle =40, bend left] (c1) to (c2);
	\draw [bend angle =30, bend right] (c1) to (c2);
	\draw [bend angle =50, bend right] (c1) to (c2);
	\draw [>-](c1) -- + (-135:.5);
	\draw (c1) -- + (-45:.5);
	\draw [>-](c2) -- + (135:.5);
	\draw (c2) -- + (45:.5);
	
		\node at (9,0) {$+$};
	
	\node [treeVertex] (c1) at (10,-.4){};
	\node [treeVertex, label={[label distance=-1mm]220:{$\scriptstyle (l)$}}  ] (c2) at ($(c1)+ (90:.8)$){};
	\draw [bend angle =40, bend left] (c1) to (c2);
	\draw [bend angle =20, bend right] (c1) to (c2);
	\draw [bend angle =40, bend right] (c1) to (c2);
	\draw [>-](c1) -- + (-135:.5);
	\draw [bend angle =40 ](c1.330) to + (55:1.2);
	\draw [>-](c2) -- + (135:.5);
	\draw [bend angle =40 ] (c2.30) to + (-55:1.2);

	\draw [semithick] (11.5,-.8) to [ncbar=-.1] (11.5,.8);

	\end{tikzpicture}
	\caption{Meta-counterterm $C_{4,2}^{(l)}$  according to \cref{C42l} for one of six orientations. }
	\label{C42_picture}
\end{figure}

With $n=4$ external edges and $k=2$ metavertices, two different configurations of multiedges are possible: Either each metavertex is adjacent to two external edges or one of them to three and one to only one external edge. In the former case, the multiedge depends on a sum  offshell variable $s_{i+j}$. There are six possibilities to choose two out of four edges offshell, each of them contributes four graphs as shown in \cref{C42_picture}; the sum can be written as
\begin{align}\label{C42l}
C_{4,2}^{(l)} &= \left \langle 4 \right \rangle \cdot  b_{l+2} b_{l+4} s_1 (s_2+s_3+s_4) \frac {M^{(l)}_{\text{div}}(s_1)}{(l+1)!}  + \left \langle 3 \right \rangle \cdot b_{l+3}^2 (s_1+s_2)(s_3+s_4) \frac { M^{(l)}_{\text{div}}(s_{1+2})}{(l+1)!}.
\end{align}
As expected, $C_{4,2}^{(l)}$ again cancels only two out of its four external edges.

Computing the higher valent meta-counterterms, this pattern continues:
\begin{align}\label{C52}
C_{5,2}^{(l)} &= \left \langle 5 \right \rangle \cdot  b_{l+2} b_{l+5} s_1 (s_2+s_3+s_4+s_5) \frac{M^{(l)}_{\text{div}}(s_1)}{(l+1)!}  + \left \langle 10 \right \rangle\cdot  b_{l+3} b_{l+4} (s_1+s_2) (s_3+s_4+s_5) \frac{M^{(l)}_{\text{div}} (s_{1+2})}{(l+1)!}  \\
C_{6,2}^{(l)} &= \left \langle 6 \right \rangle \cdot  b_{l+2} b_{l+6} s_1 (s_2+\ldots + s_6) \frac{M^{(l)}_{\text{div}}(s_1)}{(l+1)!} + \left \langle 15 \right \rangle \cdot  b_{l+3} b_{l+5} (s_1+s_2) (s_3+\ldots + s_6) \frac{M^{(l)}_{\text{div}}(s_{1+2})}{(l+1)!}  \nonumber  \\
&\qquad + \left \langle 10 \right \rangle \cdot b_{l+4}^2 (s_1+s_2+s_3)(s_4+s_5+s_6) \frac{M^{(l)}_{\text{div}}(s_{1+2+3})}{(l+1)!}. \nonumber
\end{align}
The meta-counterterm $C_{n,2}^{(l)}$ is a symmetric sum of terms, each of which corresponds to a way to partition $n$ into two nonempty disjoint sets. These partitions are elegantly given by Bell polynomials
\begin{align}\label{Bn2}
	B_{n,2} (x_1, x_2,\ldots) &= \sum_{j=1}^{n-1} \frac 12 \binom n j x_j x_{n-j}.
\end{align}
	Compare e.g. \cref{C52} to the Bell polynomials 
	\begin{align*}
		B_{5,2} \left( b_{l+2},b_{l+3},b_{l+4},\ldots  \right) &= 5 b_{l+2} b_{l+5} + 10 b_{l+3} b_{l+4} \\
		B_{6,2} \left( b_{l+2},b_{l+3},b_{l+4},\ldots  \right) &= 6 b_{l+2} b_{l+6} + 15 b_{l+3} b_{l+5}+ 10 b_{l+4}^2
	\end{align*}

\begin{lemma}\label{lem_Cn2} 
	The $l$-loop meta-counterterm with $n$ edges, two of which are cancelled, is
	\begin{align*}
		C_{n,2}^{(l)} &= \frac 12 \sum_{j=1}^{n-1}  \left \langle K_j \right \rangle  b_{l+1+j}b_{l+n+1-j}   (s_1+\ldots + s_j) (s_{j+1}+ \ldots+ s_n) \frac{M^{(l)}_{\text{\normalfont div}}(s_{1+\ldots + j})}{(l+1)!}
	\end{align*}
where  $K_j = \binom n j $.
\end{lemma}
\begin{proof}
	$C_{n,2}^{(l)}$ is given by the divergences of multiedge graphs $M^{(l)}$. A multiedge graph amounts to a partition of the $n$ external edges into precisely two nonempty disjoint sets, each of which contains the edges connected to one of the two vertices in $M^{(l)}$. We sum over symmetric permutations, therefore it is sufficient to store the cardinality $j$ of one of the two sets and fix this set to be $\{s_1, \ldots , s_j\}$. The other set  contains all remaining variables, which are to be connected to the second vertex. For one fixed permutation of the external edges, we have to sum over all ways to cancel one edge of each of the sets, this produces a factor $(s_1 + \ldots + s_j)\cdot (s_{j+1} + \ldots + s_n)$. 
	
	The momentum flowing through the multiedge is the sum of either set, we choose $s_{1+\ldots + j}$. The multiedge comes with a symmetry factor and we have to sum over $j$. The vertex-factors $(-i)^2$ produce one minus sign which is cancelled because the meta-counterterm is minus the divergence of the graph.
	
	The number of possible permutations is given by the Bell polynomial $B_{n,k}$ with values \cref{Bn2}. Note that the factor $\frac 1 2$ therein already accounts for the possibility to exchange both sets.

\end{proof}

\begin{example}[One loop]\label{ex_Cn2_1loop}
	Assume that   $M^{(1)}_{\text{div}}$ is independent of momenta, this is true for example in $D=4-2\epsilon$ dimensions for quadratic propagators. Then the explicit prefactors in \cref{lem_Cn2} constitute the only momentum-dependence. The product $(s_1+ \ldots + s_j)(s_{j+1} + \ldots + s_n)$ contains $j \cdot (n-j)$ summands. There are $K_j = \binom n j$ such terms and the sum is symmetric. The elementary symmetric quadratic polynomial $E_{2}(s_1, \ldots, s_n)$ has $\frac{n(n-1)}{2}$ factors, therefore
	\begin{align*}
		\left \langle K_j \right \rangle (s_1 + \ldots + s_j)(s_{j+1} + \ldots + s_n) &= \frac{\binom n j j (n-j)}{\frac{n(n-1)}{2}}E_2 (s_1 , \ldots, s_n) = 2 \binom{n-2}{j-1}E_2 (s_1, \ldots, s_n)
	\end{align*}
	and
	\begin{align*}
		C_{n,2}^{(1)} &=  E_2 (s_1, \ldots, s_n) M^{(1)}_{\text{div}}\sum_{j=1}^{n-1} \binom{n-2}{j-1} b_{j+2} b_{n-j+2}.
	\end{align*}
\end{example}

\subsection{k=3 legs offshell}

Meta-counterterms with $k=3$ represent the superficial divergence of triangle graphs where the sides are multiedges. These multiedges constitute subdivergences which have to be subtracted with suitable $C_{2,k}^{(l)}$ meta-counterterms. The latter are known from \cref{lem_Cn2} and thanks to \cref{lem_subdivergences} they are sufficent to eliminate all possible subdivergences. The removal of a subdivergence for $k=3$ has already been illustrated in \cref{ex_subdivergence}. 

To clarify the procedure, we compute the counterterm of the massless connected  2-loop 3-point  amplitude from \cref{sec:triangles}. In the connected perspective, three different topologies contribute. We will ignore in the following the multiedge graph $M^{(2)}$, which produces $C_{3,2}^{(2)}$ as indicated in \cref{C32_picture} and was computed already in \cref{C32l}. The two remaining graphs are shown in \cref{fig_3point_2loop}.

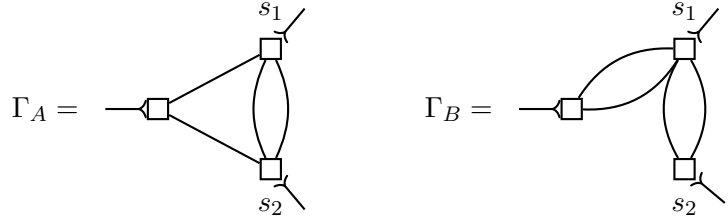
\begin{figure}[htbp] 
	\centering
	\begin{tikzpicture}

	\node at (-.5,0) {$\Gamma_A  = $};
	\node [treeVertex] (c1) at (1,0){};
	\node [treeVertex] (c2) at (2.5,.8){};
	\node [treeVertex] (c3) at (2.5,-.8){};
	\draw (c1) --  (c2);
	\draw (c1) -- (c3);
	\draw [bend angle = 25, bend left](c2) to (c3);
	
	\draw [bend angle = 25, bend right](c2) to (c3);
	\draw [>-](c1) -- + (180:.7);
	\draw [>-](c2) -- + (50:.7);
	\draw [>-](c3) -- + (-50:.7);
	
	\node at ($ (c2)+ (0,.5) $) {$s_1$};
	\node at ($ (c3)+ (0,-.5) $) {$s_2$};

	\node at (5,0) {$\Gamma_B  = $};
	\node [treeVertex] (c1) at (6.5,0){};
	\node [treeVertex] (c2) at (8,.8){};
	\node [treeVertex] (c3) at (8,-.8){};
	\draw [bend angle = 30, bend left](c2) to (c1);
	\draw [bend angle = 30, bend right](c2) to (c1);
	
	\draw [bend angle = 30, bend left](c2) to (c3);
	\draw [bend angle = 30, bend right](c2) to (c3);
	\draw [>-](c1) -- + (180:.7);
	\draw [>-](c2) -- + (50:.7);
	\draw [>-](c3) -- + (-40:.7);
	\node at ($ (c2)+ (0,.5) $) {$s_1$};
	\node at ($ (c3)+ (0,-.5) $) {$s_2$};

	\end{tikzpicture}
	\caption{The two topologies of  two-loop graphs contributing to the connected three-vertex correlation function where all three legs are cancelled. Each graph has three different permutations $s_1\rightarrow s_2 \rightarrow s_3$, they are not indicated. }
	\label{fig_3point_2loop}
\end{figure}

	If $s=p^2$ and $D=4-2\epsilon$, the graphs $\Gamma_A,\Gamma_B$ are given by \cref{G34_2loop,ex_multiedges}, where we sum over all orientations:
\begin{align}\label{3p_2l_GammaAB}
		\Gamma_A 		&=  i s_1 s_2 s_3 b_3 b_4^2 \frac{1}{4 (4\pi)^4}\left(  \frac{3}{\epsilon^2} +  \frac{1}{\epsilon } \left(  15 -6\gamma_E  + 6 \ln (4 \pi)  -2 \left( \ln s_1 + \ln s_2 + \ln s_3\right) \right)    \right)   + \text{finite terms}  \\
		\Gamma_B &= \left \langle 3 \right \rangle \cdot i b_5 b_3^2 s_1 s_2 s_3 \cdot \frac{ M^{(1)}(s_2)}{2!} \frac{ M^{(1)}(s_3)}{2!}\nonumber  \\
		&=  i b_5 b_3^2 s_1 s_2 s_3 \frac{1}{4(4\pi)^4}\left( \frac 3 {\epsilon^2} + \frac 1 \epsilon \left( 12-6\gamma_E +6\ln(4\pi) -2\left( \ln s_1 +\ln s_2+\ln s_3\right) \right)  \right)+ \text{finite terms}.\nonumber
\end{align}

\begin{figure*}[htbp]
	\begin{tikzpicture}

	\node [anchor=east] at (0,0) {$\tilde \Gamma_1 =$};
	\node [treeVertex] (c) at (1,0){ };
	\node [treeCounterVertex, label={[label distance=-1mm]0:{$\scriptstyle (1)$}} ] (c2) at (2.5,0){ };
	
	\draw [bend angle=35,bend left] (c) to (c2);
	\draw [bend angle=35,bend right] (c) to (c2);
	
	\draw [>-](c) -- + (180:.7);
	\draw [>-](c2) -- + (60:.7);
	\draw [>-](c2) -- + (-60:.7);
	
	\node [anchor=east] at (7,0) {$\tilde \Gamma_2 = C_{3,3}^{(2)} =$};
	
	\node [treeCounterVertex, label={[label distance=-1mm]0:{$\scriptstyle (2)$}} ] (c) at (8,0){};
	\draw [>-](c) -- + (180:.7);
	\draw [>-](c) -- + (60:.7);
	\draw [>-](c) -- + (-60:.7);

	\end{tikzpicture}
	\caption{Counterterm graphs to absorb the divergences of the graphs shown in \cref{fig_3point_2loop}.   $\tilde \Gamma_1$ absorbs the subdivergences of both $\Gamma_A$ and $\Gamma_B$ while $\tilde \Gamma_2$ cancels the superficial divergences of $\Gamma_A+ \Gamma_B$. }
	\label{fig_3point_2loop_counterterms}
\end{figure*}

The topologies $\Gamma_A$ and $\Gamma_B$ involve divergent subgraphs. To eliminate the subdivergences, we need the counterterm graph $\tilde \Gamma_1$ from  \cref{fig_3point_2loop_counterterms}. It contains the meta-counterterm of the 1-loop multiedge, $C_{4,2}^{(1)}$ from \cref{C42l}, which is to be inserted into another 1-loop multiedge where it must not cancel internal edges.  Let $s_1,s_2$ be the external edges as indicated in \cref{fig_3point_2loop}, then the appropriate amplitude of the meta-counterterm is
\begin{align*}
C_{4,2}^{(1)}\Big|_{s_3=s_4=0} &=  2 b_3 b_5 s_1 s_2  \frac{M^{(1)}_{\text{div}}}{2!}  + 2 b_4^2 s_1s_2 \frac{M^{(1)}_{\text{div}}}{2!}.
\end{align*}
Summing over all three orientations, the counterterm graph has the amplitude
\begin{align*}
\tilde \Gamma_1 &=   -  \frac 12 i   \left( b_3^2 b_5 + b_3 b_4^2\right) s_1 s_2s_3  M^{(1)}_{\text{div}} \left( M^{(1)}_{\text{fin}}(s_1)+M^{(1)}_{\text{fin}}(s_2)+M^{(1)}_{\text{fin}}(s_3)\right)\\
&\qquad  - \frac 32 i \left( b_3^2 b_5 + b_3 b_4^2\right)  s_1 s_2 s_3   \left(M^{(1)}_{\text{div}}\right)^2 +\text{finite terms} .
\end{align*}
	If $s=p^2$ and $D=4-2\epsilon$, this is
\begin{align}\label{3p_2l_tildeGamma}
	\tilde \Gamma_1 &=  - 2 i   \left( b_3^2 b_5+ b_3 b_4^2\right) s_1 s_2s_3 \frac{1}{4(4\pi)^4}\frac{1}{\epsilon } \left( 6-3\gamma_E + 3 \ln (4\pi) -\ln s_1-\ln s_2-\ln s_3\right) \\
	&\qquad  - 6 i  \left( b_3^2 b_5 + b_3 b_4^2\right)  s_1 s_2 s_3   \frac{1}{4(4\pi)^4}\frac{1}{\epsilon^2}  +\text{finite terms} . \nonumber 
\end{align}
The remaining graph $\tilde \Gamma_2$ from \cref{fig_3point_2loop_counterterms} absorbs the overall divergence, 
\begin{align}\label{C33_2loop}
	C_{3,3}^{(2)} &= -\renop \left( \Gamma_A + \Gamma_B + \tilde \Gamma_1 \right) .
\end{align}
In the massless case, using \cref{3p_2l_GammaAB,3p_2l_tildeGamma},
	\begin{align}\label{C33_2loop_massless}
	C_{3,3}^{(2)}	&= i s_1 s_2 s_3  \frac{3}{4 (4\pi)^4} \left( \left( b_3 b_4^2+b_3^2 b_5\right) \frac{1}{\epsilon^2} -  b_3 b_4^2   \frac 1 \epsilon\right).
\end{align}

\section{Exponential diffeomorphisms}\label{sec:rec}

So far, the set of parameters $\left \lbrace a_j \right \rbrace _{j\in \mathbbm N}$ in \cref{def_diffeomorphism} or equivalently $\left \lbrace b_j \right \rbrace  _{j\in \mathbbm N}$ in \cref{def_diffeomorphism_inverse} has been arbitrary. In the present section,   we choose  $u\in \mathbbm N$ fixed and arbitrary and require 
\begin{align}\label{rec_bn}
	b_n &= \begin{cases}
		1 & n=2\\
		\lambda^{n-2} & \exists k\in \mathbbm N_0: uk = n-2 \\
		0 & \text{else}
	\end{cases}.
\end{align}
We call this class of theories \emph{exponential diffeomorphisms}.

Note that the choice \cref{rec_bn} implies that not only the $n$-valent metavertex \cref{Vn}, but actually all summands of the connected $n$-point function are proportional to $\lambda^{n-2}$. More precisely, the connected tree-level $n$-point function has the form
\begin{align}\label{rec_Gntl_lambda}
	G_n^{\text{tl}} &= -i \lambda^{n-2}  \left( E_1(s_1, \ldots, s_n) + R_n\right)
\end{align}
where $E_1(s_1, \ldots, s_n) = s_1 + \ldots + s_n$ is the elementary symmetric polynomial of order one and $R_n$ is a rational function symmetric in $\{s_1, \ldots,  s_{1+2} , \ldots, s_{1+2+3}, \ldots   \}$ of overall order one in $s$.

\subsection{Inverse field diffeomorphism}

The hyperbolic function of order $u$ of the $r^{\text{th}}$ kind is defined as \cite{ungar_generalized_1982}
\begin{align*}
	H_{u,r}(x) &:= \sum_{k=0}^\infty \frac{x^{uk+r}}{\Gamma\left( uk+1+r \right)  }.
\end{align*}

\begin{lemma}\label{lem_rec_positionspace}
	If the connected $n$-point functions $ib_n$ of a field $\rho(x)$ fulfil  \cref{rec_bn} for a fixed $u\in \mathbbm N$, then $\rho$ is related to a free field $\phi(x)$ by
	\begin{align*}
		\lambda \rho(x) 	&=   H_{u,1}(\lambda \phi(x))= \lambda \phi \cdot \begin{cases}{}_1 F_1 (1;2|\lambda \phi (x)), &u=1\\
			{}_0F_{u-1} \left(  \left \lbrace   \right \rbrace ;  \left \lbrace \frac 2 u, \frac 3 u, \ldots, \frac{u-1}{u}, \frac{u+1}{u} \right \rbrace   \Big| \left( \frac{\lambda \phi(x)}{u} \right) ^u \right)  & u\geq 2.	
		\end{cases}
	\end{align*}
\end{lemma}
\begin{proof}
	 If a propagator-cancelling field has metavertex amplitudes $b_n$ then it is related to a free field $\phi$ by the diffeomorphism \cref{def_diffeomorphism_inverse}
	\begin{align*}
		\rho(x) &= \sum_{n=1}^\infty \frac{b_{n+1}}{n!}\phi^n(x)=  \phi \sum_{k=0}^\infty \frac{\left( \lambda \phi \right) ^{ku}}{(uk+1)!}=  \phi \sum_{k=0}^\infty \frac{\left( \lambda \phi \right) ^{ku}}{\Gamma(uk+2)}.
	\end{align*}
	For $u=1$, the series is
	\begin{align*}
		\sum_{k=0}^\infty \frac{(\lambda \phi)^k}{\Gamma(k+2)} &=  \sum_{k=0}^\infty \frac{\Gamma(s)}{\Gamma(k+2)}\frac{\Gamma(k+1)}{\Gamma(1)}\frac{(\lambda \phi)^k}{k!} = {}_1F_1 (1; 2|\lambda \phi).
	\end{align*}
	In the general case, use Gauss' product formula for Gamma functions \cite[Sec. 26]{gauss_disquisitiones_1813} twice:
	\begin{align*}
		\sum_{k=0}^\infty \frac{1}{\Gamma(uk+2)}\left( \lambda \phi \right) ^{ku}	&=  \sum_{k=0}^\infty \frac{(2\pi)^{\frac {u-1} 2 }}{ \Gamma\left( \frac{2+uk}{u} \right) \cdots \Gamma \left( \frac{u+1+uk}{u} \right) u^{uk + \frac 3 2} }  \left( \lambda \phi \right) ^{uk} \\
		&=  \sum_{k=0}^\infty \frac{(2\pi)^{\frac {u-1} 2 }}{ \Gamma\left( \frac 2 u + k  \right) \cdots \Gamma \left( 1+k \right)  \Gamma \left( \frac{u+1}{u}+k \right) u^{\frac 3 2} }    \frac{\Gamma \left( \frac 2 u  \right) \cdots \Gamma \left( \frac{u+1}{u} \right) u^{\frac 3 2}}{(2 \pi)^{\frac{u-1}{2}} \Gamma(2)}  \frac{ \left( \lambda \phi \right) ^{uk} }{u^{uk}}.
	\end{align*}
\end{proof}

We note in passing that the diffeomorphisms given by \cref{rec_bn} fulfil
\begin{align*}
	\frac{\text d^u }{\text d \phi^u} \rho(x) &= \lambda^u \cdot  \rho(x)
\end{align*}
and are therefore sums of terms $\rho \propto e^{q_i \lambda \phi}$ where $q_i$ are $u^{\text{th}}$ roots of unity. This fact motivates the name \emph{exponential diffeomorphisms}.

\begin{example}\label{ex_rec_positionspace}
	For small $u$, the hypergeometric functions evaluate to
	\begin{align*}
		u=1 :\quad \rho  &= \lambda^{-1} \left( e^{\lambda \phi}-1 \right) \\
		u=2 : \quad \rho  & = (2\lambda)^{-1} \left( e^{\lambda\phi} - e^{-\lambda \phi} \right) = \lambda^{-1} \sinh \left( \lambda \phi  \right)  \\
		u=3 : \quad \rho &= (3\lambda)^{-1} \left( e^{\lambda \phi} + (-1)^{\frac 1 3} e^{-(-1)^{\frac 13} \lambda \phi} + (-1)^{\frac 23} e^{(-1)^{\frac 23} \lambda \phi} \right) \\
		 &= (3\lambda)^{-1}e^{-\frac{\lambda \phi }{2}} \left(  e^{\frac 3 2 \lambda \phi }  + 2 \sin \left( \frac 16 (3\sqrt 3 \lambda \phi -\pi) \right)   \right) \\
		u=4 : \quad \rho &= (2\lambda)^{-1} \left( \sin (\lambda \phi) + \sinh(\lambda \phi) \right) .
	\end{align*}
\end{example}

\subsection{Field diffeomorphism}
Using \cref{anbn}, the diffeomorphism coefficients $a_n$ can be computed in principle, but there seems to be no easy explicit formula. One obtains
\begin{align*}
	u=1: \quad a_n &= \frac{(-1)^n \lambda^n}{n+1} \\
	u>1: \quad a_n &= \begin{cases}
		\frac{(-1)^k \lambda^n}{(n+1)!} \cdot \alpha_k ,\quad &n = k\cdot u\\
		0 &\text{else}
	\end{cases}
\end{align*}
where the sequences $\left \lbrace \alpha_k \right \rbrace _{k\in \mathbbm N}$ have recently been interpreted in terms of Whitney numbers\cite{engbers_restricted_2017}.
\begin{example}
	The case $u=1$ amounts to $\alpha_n = n!$. Other examples  of $\left \lbrace \alpha_k \right \rbrace $ are 
	\begin{align*}
		u&=2: \quad \left \lbrace 1,9,225,11025,893025,\ldots   \right \rbrace\quad\text{\cite[A001818]{oeis}}  \\
		u&= 3 : \quad \left \lbrace 1,34,5446,2405116, 2261938588 ,\ldots   \right \rbrace \quad  \text{\cite[A292750]{oeis}} \\
		u &= 4 : \quad \left \lbrace 1,125,124125,477257625,\ldots   \right \rbrace  \\
		u &= 5: \quad \left \lbrace 1,461,2846481,95573412836,\ldots  \right \rbrace  
	\end{align*}
\end{example}

For $u=1,2$, the function $\phi(\rho)$ can be obtained by inverting the function $\rho(\phi)$ from \cref{ex_rec_positionspace}:
\begin{align}\label{rec_inverse}
	u=1 : \quad \phi &= \lambda^{-1} \ln (1+\lambda \rho)\\
	u=2 : \quad \phi &= \lambda^{-1} \operatorname{asinh} (\lambda \rho) = \lambda^{-1} \ln (\sqrt{1+(\lambda \rho)^2}+ \lambda \rho ). \nonumber 
\end{align} 
It is instructive to write down the Lagrangian density for these diffeomorphisms in the case $s=p^2-m^2$, namely
\begin{align}
u=1: \qquad 	\mathcal L  &= -\frac 12 \partial_\mu \phi \partial^\mu \phi -\frac 12 m^2 \phi^2  = -\frac 12 \frac{1}{(1+\lambda \rho)^2}\partial_\mu \rho \partial^\mu \rho -\frac 12 \frac{m^2}{\lambda^2}\ln^2 (1+\lambda \rho) \label{rec_L_u1}\\
	u=2: \qquad \mathcal L &= -\frac 12 \frac{1}{1+(\lambda \rho)^2} \partial_\mu \rho \partial^\mu \rho -\frac 12 \frac{m^2}{\lambda^2}\sinh^2 (\lambda \phi).\nonumber
\end{align}
In both cases, the kinetic term is multiplied by the inverse squared of the field. The mass term on the other hand becomes an unorthodox transcendental interaction term. Setting $m=0$ and defining a field $\varrho := 1+\lambda \rho$,   \cref{rec_L_u1} takes the form  
\begin{align}\label{rec_u1_L}
 u=1: \qquad 	\mathcal L &= -\frac{1}{2\lambda^2} \cdot  \frac{1}{\varrho^2}\partial_\mu \varrho \partial^\mu \varrho
\end{align} 
which vagely reminds of the Einstein-Hilbert-Lagrangian.
Using \cref{vn_massless} and $s_j = p_j^2$, the vertex Feynman rules of the Lagrangian \cref{rec_u1_L}, in terms of the field $\rho = \lambda^{-1}(\varrho -1)$, are
\begin{align*} 
	iv_n &= i (-1)^n \lambda^{n-2} \frac  {(n-1)!}2\cdot \left( p_1^2 + p_2^2 + \ldots + p_n^2 \right)  .
\end{align*}

\subsection{Two-point-function in position space}

\begin{lemma}\label{rec_lem_G2}
	If the field $\rho(x)$ fulfils \cref{rec_bn} for a fixed $u\in \mathbbm N$, then its  connected full two-point function in position space  is
	\begin{align*}
		G(z) &= \frac 1 \lambda H_{u,1}(\lambda^2 G_F(z)) 
	\end{align*}
	where $G_F(z)$ is the Greens function of the original field differential operator, i.e. $\hat s_z G_F(z) = i\delta(z)$   and $H_{u,1}$ can be expressed in $_0F_{u-1}$ like in \cref{lem_rec_positionspace}.
\end{lemma}
\begin{proof}
	The non-tadpole part of the 2-point function is \cref{posspace_2point}, using \cref{rec_bn}, it reads
	\begin{align}\label{rec_pos_G2_series}
		\left \langle \rho(x) \rho(y) \right \rangle &= \sum_{t=1}^\infty \frac{\left( \lambda^{t-1} \right) ^2 \delta_{t-1 = uk}}{t!} G_F^t(x-y) = G_F(x-y) \sum_{k=0}^\infty  \frac{\left( \lambda^{2} G_F(x-y)\right)^{uk} }{\Gamma(uk+2)}  . 
	\end{align}
	This is up to a different argument the same series as in the proof of \cref{lem_rec_positionspace}. 
\end{proof}
\begin{example}[Massless field]
	Consider the  standard, massless theory $s=p^2$ with propagator \cref{propagators}. Like in \cref{ex_rec_positionspace}, one obtains
	\begin{align}\label{rec_u1_G2_pos}
		u=1: \quad G(z) &= \lambda^{-2} \left( \exp \left( \frac{i\Gamma(1-\epsilon) \lambda^2}{(z^2)^{1-\epsilon} 4 \pi^{2-\epsilon}}   \right) -1 \right)  = \lambda^{-2} \left( e^{ i\frac{\lambda^2}{z^2 (2 \pi)^2} }   -1 \right) +   \mathcal O (\epsilon) \\
		u=2: \quad G(z) &= \lambda^{-2} \sinh \left(  \frac{ \Gamma(1-\epsilon)\lambda^2}{(z^2)^{1-\epsilon} 4 \pi^{2-\epsilon}} \right)  .\nonumber 
	\end{align}
	In stark contrast to the free propagator $G_F$ (\cref{propagators}) or the perturbative 2-point function of any renormalizable theory, these functions have an essential singularity at $z=0$. But apart from that, they are finite in the limit $\epsilon \rightarrow 0$. Especially, the case $u=1$ amounts to  the exponential superpropagator \cref{superpropagator}. 
\end{example}

\subsection{Higher correlation functions in position space}\label{sec:pos_higher}

By \cref{thm_npoint_position}, the $n$-point function in position space for a theory fulfilling \cref{rec_bn} with $u=1$ is 
\begin{align*}
	\left \langle \rho(x_1) \cdots \rho(x_n)\right \rangle = \sum_{\stackrel{l_1, \ldots, l_k\in \mathbbm N_0}{t_j \geq 1 \; \forall j}} \frac{\lambda^{t_1-1} \cdots \lambda^{t_n-1}}{l_1! \cdots l_k!}G_F^{l_1}(x_1-x_2)G_F^{l_2}(x_1-x_3)\cdots G_F^{l_k}(x_{n-1}-x_n).
\end{align*}
Here, $l_j\in \mathbbm N_0$ represent the number of edges between a pair of points and $t_j\in \mathbbm N$ the number of edges at one point $j\in \left \lbrace 1,\ldots, n \right \rbrace  $. Hence each $l_i$ contributes to precisely two $t_j$ and   $t_1 + \ldots + t_n = 2 (l_1  + \ldots + l_k)$ and 
\begin{align}\label{rec_pos_npoint}
	\left \langle \rho(x_1) \cdots \rho(x_n)\right \rangle = \frac{1}{\lambda^n}\sum_{\stackrel{l_1, \ldots, l_k\in \mathbbm N_0}{t_j \geq 1 \; \forall j}} \frac{\lambda^{2 l_1} \cdots \lambda^{2 l_k}}{l_1! \cdots l_k!}G_F^{l_1}(x_1-x_2)G_F^{l_2}(x_1-x_3)\cdots G_F^{l_k}(x_{n-1}-x_n).
\end{align}
 Recognizing the 2-point function  
\begin{align*}
	G (z) &:= \sum_{l=0}^\infty \frac{\lambda^{2l}}{l!}G_F^{l}(z),
\end{align*}
we see that \cref{rec_pos_npoint} is the sum of all ways to connect the $n$ points by those 2-point functions. This gives us yet another interpretation of the recursion relations \cref{rec_bn}: They are the unique choice of parameters $b_n$ for which the position-space correlation functions factor into products of position-space superpropagators.

\subsection{Two-point function in momentum space}
Assuming recursion relations \cref{rec_bn}, the connected amputated 2-point function \cref{G2} for $s_p=p^2$ in $D=4-2\epsilon$ dimensions specialises to
\begin{align}
	G_2^{\text{fin}} (s) &=   \sum_{k=1}^\infty   \frac{\left(- i \pi^2 \lambda^2 s \right) ^{uk}}{ \Gamma(uk+2)\Gamma^2 (uk+1)} \left(  (2uk+1) H_{uk}-1 -uk\gamma_ {\text E}  + u k \ln (2\pi) -uk \ln  s   \right)  \label{rec_G2fin} \\
	G_2^{\text{div}}(s) &=  \sum_{k=1}^\infty   \frac{\left(- i  \lambda^2 s \right) ^{uk}}{ (4\pi)^{2uk} \Gamma(uk+2) \Gamma^2 (uk+1)}. \label{rec_G2div}
\end{align}

Restricting further to $u=1$, the finite part 
\begin{align}\label{rec_u1_G2fin}
	G_2^{\text{fin}} (s) &=   \sum_{k=1}^\infty   \frac{\left(- i \pi^2 \lambda^2 s \right) ^{ k}}{ \Gamma( k+2)\Gamma^2 ( k+1)} \left(  (2 k+1) H_{ k}-1 - k\gamma_ {\text E}  +   k \ln (2\pi) - k \ln  s   \right)  
\end{align}
almost coincides with the exponential superpropagator \cref{superpropagator_mom}. The difference of our formula \cref{rec_u1_G2fin} to the various historic results is a finite constant $\delta_k$ for every order $s^k$. In a renormalizable theory, such differences are expected when using different renormalization prescriptions, compare e.g.  \cite{celmaster_renormalizationprescription_1979}. In the present, non-renormalizable case, the historic results assumed different non-standard conditions to fix a unique Fourier-transform, our own result is in MS and assumes the vanishing of tadpoles. \Cref{rec_u1_G2fin,rec_u1_G2_pos} show that the exponential superpropagator can indeed be obtained within rigorous perturbation theory both in position space and in momentum space.

\begin{lemma}\label{rec_lem_G2div}
	In $D=4-2\epsilon$ dimensions with $s_p=p^2$, if a theory fulfils \cref{rec_bn} then the divergent part of its connected amputated 2-point function is
	\begin{align*}
		G_2^{\text{\normalfont div}} (s) &=  \begin{cases}
			{}_0 F_2 \left( \left \lbrace ; 1,2,\big| -i\pi^2 \lambda^2 s \right \rbrace   \right)-1 & u=1 \\
			{}_0F_{3u-1} \left( \left \lbrace   \right \rbrace ; \frac 1 u , \frac 1 u , \frac 2 u, \frac 2 u , \frac 2 u , \ldots, \frac{u-1}{u},\frac{u-1}{u},\frac{u-1}{u}, 1, 1, \frac{u+1}{u}\Big| \left( \frac{ - i  \lambda^2 s  }{ (4\pi)^2 u^3} \right)^u \right)  -1& u>1.
		\end{cases} 
	\end{align*}
\end{lemma}
\begin{proof}
	Analogous to \cref{lem_rec_positionspace}.  Use
	\begin{align*}
		\Gamma \left( uk+2 \right) \Gamma^2 \left( uk+1 \right)  &= \frac{\Gamma \left( \frac 2 u + k  \right) \cdots \Gamma \left( \frac{u-1}u \right)  \Gamma \left( \frac{u+1}{k} \right)  \Gamma^2 \left( \frac 1 u + k  \right) \cdots \Gamma^2 \left( \frac{u-1}{u}+k \right)  }{\Gamma \left(  \frac 2 u  \right) \cdots \Gamma \left( \frac{u-1}{u} \right) \Gamma \left( \frac{u+1}{u} \right)  \Gamma^2 \left( \frac 1 u  \right) \cdots \Gamma ^2\left( \frac{u-1}{u} \right)  }\Gamma ^3\left( k+1 \right)  u^{3uk }
	\end{align*}
to rewrite \cref{rec_G2div} as a hypergeometric function.
\end{proof}

\subsection{One loop divergences}\label{sec_rec_1loop}

Consider massless one-loop graphs, not including external propagators. In $D=4-2\epsilon$ dimensions, the only divergent one-loop graph is the one-loop multiedge $M^{(l)}$ which contributes to all $n$-point functions, see  \cref{fig_rec_1L}. Similar to \cref{rec_Gntl_lambda}, $G_n^{(1)} \propto \lambda^n$.

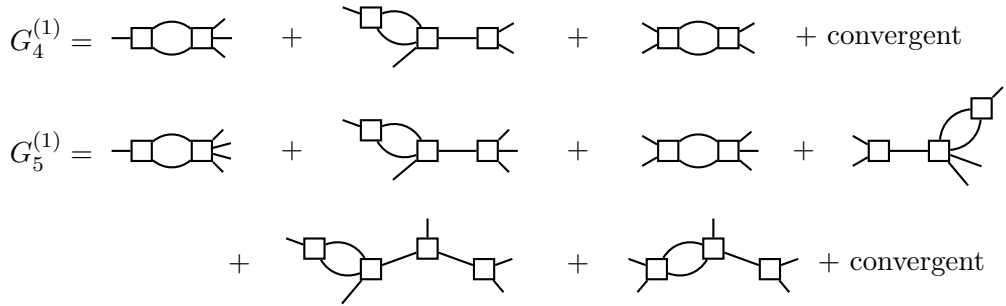
\begin{figure}[htbp] 
	\centering
	\begin{tikzpicture}

		\node at (0,-1.5) {$G_4^{(1)}=$};
		
		\node [treeVertex] (c1) at (2,-1.5){};
		\node [treeVertex] (c2) at ($(c1) + (180:.8)$){};
		\draw [bend angle =40, bend left] (c1) to (c2);
		\draw [bend angle =40, bend right] (c1) to (c2);
		\draw (c2) -- + (180:.4);
		\draw (c1) -- + (0:.4);
		\draw (c1) -- + ( 40:.4);
		\draw (c1) -- + (-40:.4);
		
		\node at (3.2,-1.5) {$+$};
		
		\node [treeVertex] (c1) at (5,-1.5){};
		\node [treeVertex] (c2) at ($(c1) + (160:.8)$){};
		\node [treeVertex] (c3) at ($(c1) + (0:.8)$){};
		\draw [bend angle =40, bend left] (c1) to (c2);
		\draw [bend angle =40, bend right] (c1) to (c2);
		\draw (c2) -- + (160:.4);
		\draw (c1) -- (c3);
		\draw (c1) -- + (220:.6);
		\draw (c3) -- + ( -30:.4);
		\draw (c3) -- + (30:.4);
		
		\node at (7,-1.5) {$+$};
		
		\node [treeVertex] (c1) at (9,-1.5){};
		\node [treeVertex] (c2) at ($(c1) + (180:.8)$){};
		\draw [bend angle =40, bend left] (c1) to (c2);
		\draw [bend angle =40, bend right] (c1) to (c2);
		\draw (c2) -- + (210:.4);
		\draw (c2) -- + (150:.4);
		\draw (c1) -- + ( 30:.4);
		\draw (c1) -- + (-30:.4);
		
		\node at (11,-1.5){+ convergent};
		
		\node at (0,-3) {$G_5^{(1)}= $};
		
		\node [treeVertex] (c1) at (2,-3){};
		\node [treeVertex] (c2) at ($(c1) + (180:.8)$){};
		\draw [bend angle =40, bend left] (c1) to (c2);
		\draw [bend angle =40, bend right] (c1) to (c2);
		\draw (c2) -- + (180:.4);
		\draw (c1) -- + (45:.4);
		\draw (c1) -- + ( 15:.4);
		\draw (c1) -- + (-15:.4);
		\draw (c1) -- + (-45:.4);

		\node at (3.2,-3) {$+$};
		
		\node [treeVertex] (c1) at (5,-3){};
		\node [treeVertex] (c2) at ($(c1) + (160:.8)$){};
		\node [treeVertex] (c3) at ($(c1) + (0:.8)$){};
		\draw [bend angle =40, bend left] (c1) to (c2);
		\draw [bend angle =40, bend right] (c1) to (c2);
		\draw (c2) -- + (160:.4);
		\draw (c1) -- (c3);
		\draw (c1) -- + ( 220:.6);
		\draw (c3) -- + ( 45:.4);
		\draw (c3) -- + (-45:.4);
		\draw (c3) -- + (0:.4);
		
		\node at (7,-3) {$+$};
		
		\node [treeVertex] (c1) at (9,-3){};
		\node [treeVertex] (c2) at ($(c1) + (180:.8)$){};
		\draw [bend angle =40, bend left] (c1) to (c2);
		\draw [bend angle =40, bend right] (c1) to (c2);
		\draw (c2) -- + (210:.4);
		\draw (c2) -- + (150:.4);
		\draw (c1) -- + ( 40:.4);
		\draw (c1) -- + ( 0:.4);
		\draw (c1) -- + (-40:.4);
		
		\node at (10,-3) {$+$};
		
		\node [treeVertex] (c1) at (11,-3){};
		\node [treeVertex] (c3) at ($(c1) + (0:.8)$){};
		\node [treeVertex] (c2) at ($(c3) + (45:.8)$){};
		\draw [bend angle =40, bend left] (c3) to (c2);
		\draw [bend angle =40, bend right] (c3) to (c2);
		\draw (c2) -- + (45:.4);
		\draw (c1) -- (c3);
		\draw (c1) -- + ( 150:.4);
		\draw (c1) -- + ( 210:.4);
		\draw (c3) -- + (-20:.6);
		\draw (c3) -- + (-50:.6);

		\node at (2.5,-4.5) {$+$};
		
		\node [treeVertex] (c1) at (5,-4.3){};
		\node [treeVertex] (c2) at ($(c1) + (200:.8)$){};
		\node [treeVertex] (c3) at ($(c1) + (-20:.8)$){};
		\node [treeVertex] (c4) at ($(c2) + (160:.8)$){};
		\draw [bend angle =40, bend left] (c2) to (c4);
		\draw [bend angle =40, bend right] (c2) to (c4);
		\draw (c1) -- + (90:.4);
		\draw (c1) -- (c3);
		\draw (c3) -- + ( -50:.4);
		\draw (c3) -- + (20:.4);
		\draw (c1) -- (c2);
		\draw (c2) -- + (230:.6);
		\draw (c4) -- + (160:.4);
		
		\node at (7,-4.5) {$+$};
		
		\node [treeVertex] (c1) at (8.8,-4.3){};
		\node [treeVertex] (c2) at ($(c1) + (200:.8)$){};
		\node [treeVertex] (c3) at ($(c1) + (-20:.8)$){};
		\draw [bend angle =40, bend left] (c1) to (c2);
		\draw [bend angle =40, bend right] (c1) to (c2);
		\draw (c2) -- + (160:.4);
		\draw (c2) -- + (230:.4);
		\draw (c1) -- + (90:.4);
		\draw (c1) -- (c3);
		\draw (c3) -- + ( 20:.4);
		\draw (c3) -- + (-50:.4);
		
		\node at (11.3,-4.5){+ convergent};
		
	\end{tikzpicture}
	\caption{Divergent contribution  to the $4-$ and 5-point amplitudes. Permutations of external edges are not indicated. }
	\label{fig_rec_1L}
\end{figure}

\begin{example}[5-point function]\label{ex_rec_1loop_An}
	Assuming  \cref{rec_bn} for $u=1$, the graphs  in \cref{fig_rec_1L} evaluate to
	\begin{align*}
		G^{(1)}_4 &= -\lambda^4 \left( \langle 4 \rangle  s_1 (s_2+s_3+s_3) M^{(1)}(s_1)+\langle 12 \rangle  s_1 s_2 \frac{1}{s_{1+2}} (s_3+s_4) M^{(1)}(s_1)+  \right.  \\
		&\qquad \left.  +  \langle 3 \rangle(s_1 + s_2) (s_3+s_4) M^{(1)} (s_{1+2})  \right)  + \text{convergent graphs} \\
		G^{(1)}_5 &= -\lambda^5 \left( \langle 5 \rangle s_1 (s_2+s_3+s_4+s_5) M^{(1)} (s_1)   + \langle 20 \rangle  s_1 s_2 \frac{1}{s_{1+2}}(s_3+s_4+s_5) M^{(1)} (s_1) \right.\\
		&\qquad  + \langle 10 \rangle (s_1+s_2)(s_3+s_4+s_5) M^{(1)}(s_{1+2})  + \langle 15 \rangle (s_1+s_2) \frac{1}{s_{1+2} }s_3 s_4 s_5 M^{(1)}(s_3)\\
		&\qquad \left. + \langle 60 \rangle s_1 s_2 \frac{1}{s_{1+2}}s_3 \frac{1}{s_{4+5}}(s_4+s_5) M^{(1)}(s_1)  +\langle 30 \rangle (s_1+s_2) s_3 \frac{1}{s_{4+5}} (s_4+s_5) M^{(1)}(s_{1+2})  \right)+ \text{c}.   
	\end{align*}
\end{example}

We are interested only in the leading terms of $G_n^{(1)}$, which do not involve uncancelled propagators, because all other terms are products of  $G_j^{(1)}$ with $j<n$. The divergence of the leading terms is cancelled by the meta-counterterm $C_{n,2}^{(1)}$.

\begin{lemma}\label{lem_Cn2_rec_1loop}
	In $D=4-2\epsilon$ dimensions and for $s_p=p^2$ and assuming the relations \cref{rec_bn} for $u=1$, the one-loop meta-counterterm with $n$ edges cancels precisely two of its edges and reads
	\begin{align*}
		C_n^{(1)} \equiv C_{n,2}^{(1)} &= -\lambda^n \frac{2^{n-3}}{(4\pi)^2}  \frac 1 \epsilon E_2 \left( p_1^2, \ldots, p_n^2 \right).
	\end{align*}
	 $E_2(s_1, s_2, \ldots, s_n)$ is the elementary symmetric polynomial of order two in $n$ variables.
\end{lemma}
\begin{proof}
	Follows from \cref{lem_Cn2} since at one-loop level, the only divergent graph is the one-loop-multiedge, or, equivalently, $C_{n,k>2}^{(1)}=0$. The multiedge has a symmetry factor $\frac 12$ and, as computed in \cref{ex_multiedges}, its divergent part is momentum-independent,  $M^{(1)}_{\text{div}} (s_{P_j}) = -(4\pi)^{-2} \epsilon^{-1}$. Therefore  we   insert \cref{rec_bn} into \cref{ex_Cn2_1loop}: 
	\begin{align*}
		C_{n,2}^{(1)} &=  -E_2 (s_1, \ldots, s_n) \frac 12 (4\pi)^{-2} \epsilon^{-1}\sum_{j=1}^{n-1} \binom{n-2}{j-1} \lambda^j \lambda^{n-j} =  -E_2 (s_1, \ldots, s_n) \frac 1 {2(4\pi)2\epsilon }  (1+1)^{n-2}.
	\end{align*}
\end{proof}

\subsection{Higher correlation functions in momentum space}

We have seen in \cref{sec:pos_higher} that if \cref{rec_bn} is fulfilled with $u=1$, then the $n$-point-amplitudes factor into products of all-order 2-point-functions in position space. A similar statement holds in momentum space.

\begin{theorem}[Feynman rules of exponential diffeomorphism]\label{thm_rec_fey}
	Assume the coefficients $b_n$ fulfil \cref{rec_bn} for $u=1$. Then the connected $(n>2)$-point amplitude $G_n$ is given by the following connected graphs:
	\begin{enumerate}
		\item The vertices are metavertices $V_n = -i\lambda^{n-2}(s_1 + \ldots + s_n)$ for any valence $n\geq 3$.
		\item Edges between metavertices are given by the function $G_2(s)$.
		\item $G_2(s)$ is the full, amputated 2-point-function, the sum of all multiedges \cref{rec_u1_G2fin}.
		\item In $G_{n>2}$, there is at most one  edge directly between any two metavertices.
		\item Every Metavertex cancels precisely one of the $n$ external edges.
		\item Every external edge which is not cancelled by a metavertex is dressed by $G_2(s)$.
		
	\end{enumerate}
\end{theorem}
\begin{proof}
	Consequence of \cref{thm_feynmanrules} by identification of the connected 2-point amplitude with a full propagator. The only point to be shown is that the recursion relations \cref{rec_bn} are sufficient to identify each internal multiedge with the corresponding term in $G_2$. 
	
	To see this, let there be two metavertices $V_{n_1}, V_{n_2}$ and between them $j$ propagators, this object is $\propto \lambda^{n_1+n_2-4}$. There are $n_1-j$ resp. $n_2-j$ edges left adjacent to the metavertices which are not between them. Now, the contribution to $G_2$ which has $j$ internal edges is the $(j-1)$-loop multiedge and comes with prefactor $\lambda^{2j-2}$ due to its two vertices $V_{j+1} \propto \lambda^{j-1}$. The two vertices cancel the two outer edges of the multiedge, hence these propagators do not contribute to $G_2$. To make up for the missing external edges, one needs to connect to $G_2$ two metavertices $V_{n_1-j+1}$ and $V_{n_2-j+1}$ respectively. These produce factors $\lambda^{n_1-j-1}$ and $\lambda^{n_2-j-1}$. In total, the $j$-edge term of $C_2$ together with the two vertices is $\propto \lambda^{2j-2 + n_1-j-1+n_2-j-1} = \lambda^{n_1+n_2-4}$. This is the same factor as if we did not insert $C_2$ and instead used $j$ internal edges. 
	
	Phrased differently, connecting $C_2$ to a vertex $V_n$ amounts to splitting $V_n$ into two parts $V_{n_1} \cdot V_{n-n_1+2}$ and cancelling the intermediate propagator. But in both cases, the amplitude is $\lambda^{n-2}$ by \cref{rec_bn}.
	
\end{proof}

Observe that to compute a $n$-point function  by \cref{thm_rec_fey}, assuming $G_2$ is known, only finitely many integrals remain  to be solved. The graph with highest loop number is the completely connected graph on $n$ vertices, it has $\frac{n(n-1)}{2}$ internal edges and hence $\frac 12 (n^2-3n+2)$ loops.

\begin{example}[Lowest amplitudes]
	Following \cref{thm_rec_fey}, the   lowest $n$-point amplitudes are shown in \cref{fig_rec_npoint}. If one knows $G_2$, one can compute $G_3$ with only one integration and $G_4$ with three integrations. Since all metavertices are external, the graphs constructed from \cref{thm_rec_fey} all carry symmetry factor 1. 
\end{example}

\begin{figure}[htbp] 
	\centering
	\begin{tikzpicture}
		
		\node [anchor=west]at (-1,.5) {$ G_2 = $};
		
		\node [fullVertex] (c1) at (1.2,.5){};
		\draw (c1) -- + (180:.5);
		\draw (c1) -- + (0:.5);

		\node   at (6,.5) {$ G_3= \ \left \langle 3 \right \rangle  $};
		
		\node [treeVertex] (c1) at (7.6,.5){};
		\node [fullVertex] (c2) at ($(c1) + (140:.5)$){};
		\node [fullVertex] (c3) at ($(c1) + (220:.5)$){};
		\draw (c1) -- (c2);
		\draw (c1) -- (c3);
		\draw  (c2) -- + (140:.5);
		\draw [>-](c1) -- + (0:.5);
		\draw  (c3) -- + (220:.5);
		
		\node at (8.5,.5) {$+$};
		
		\node [treeVertex] (c1) at (10,.5){};
		\node [fullVertex] (c2) at ($(c1) + (140:.5)$){};
		\node [fullVertex] (c3) at ($(c1) + (220:.5)$){};
		\node [treeVertex] (c4) at ($(c1) + (140:1)$){};
		\node [treeVertex] (c5) at ($(c1) + (220:1)$){};
		\node [fullVertex] (c6) at ($(c1) + (180:.76)$){};
		
		\draw (c1) -- (c2);
		\draw (c1) -- (c3);
		\draw (c4) -- (c2);
		\draw (c5) -- (c3);
		\draw (c4) -- (c6);
		\draw (c5) -- (c6);
		\draw [>-](c4) -- + (140:.5);
		\draw [>-](c1) -- + (0:.5);
		\draw [>-](c5) -- + (220:.5);

		\node [anchor=west] at (-1,-2) {$G_4 = \ \left \langle 4 \right \rangle  $};
		\node [treeVertex] (c1) at (1.5,-2){};
		\node [fullVertex] (c2) at ($(c1) + (135:.5)$){};
		\node [fullVertex] (c3) at ($(c1) + (225:.5)$){};
		\node [fullVertex] (c4) at ($(c1) + (315:.5)$){};
		
		\draw (c1) -- (c2);
		\draw (c1) -- (c3);
		\draw (c1) -- (c4);
		
		\draw [>-](c1) -- + (45:.8);
		\draw (c2) -- + (135:.5);
		\draw (c3) -- + (225:.5);
		\draw (c4) -- + (315:.5);
		
		\node at (3,-2){$+ \ \left \langle 12 \right \rangle  $};
		
		\node [treeVertex] (c1) at (4.5,-2){};
		\node [fullVertex] (c2) at ($(c1) + (0:.5)$){};
		\node [treeVertex] (c3) at ($(c1) + (0:1)$){};
		\node [fullVertex] (c4) at ($(c1) + (140:.5)$){};
		\node [fullVertex] (c5) at ($(c3) + (315:.5)$){};
		
		\draw (c1) -- (c4);
		\draw (c1) -- (c2);
		\draw (c2) -- (c3);
		\draw (c3) -- (c5);
		
		\draw [>-](c1) -- + (220:.8);
		\draw [>-](c3) -- + (40:.8);
		\draw (c4) -- + (140:.5);
		\draw (c5) -- + (315:.5);
		
		\node at (7,-2){$+ \ \left \langle 12 \right \rangle  $};
		
		\node [treeVertex] (c1) at (8.5,-2){};
		\node [treeVertex] (c2) at ($(c1) + (30:1)$){};
		\node [treeVertex] (c3) at ($(c1) + (-30:1)$){};
		\node [fullVertex] (c4) at ($(c1) + (140:.5)$){};
		\node [fullVertex] (c5) at ($(c1) !.5! (c2)$){};
		\node [fullVertex] (c6) at ($(c1) !.5! (c3)$){};
		\node [fullVertex] (c7) at ($(c3) !.5! (c2)$){};

		\draw (c1) -- (c4);
		\draw (c1) -- (c3);
		\draw (c1) -- (c2);
		\draw (c2) -- (c3);
		
		\draw [>-](c1) -- + (220:.8);
		\draw [>-](c2) -- + (40:.8);
		\draw (c4) -- + (140:.5);
		\draw [>-](c3) -- + (315:.8);

		\node at (.8,-4.5) {$+ \ \left \langle 3 \right \rangle  $};
		
		\node [treeVertex] (c1) at (1.8,-3.9){};
		\node [treeVertex] (c2) at ($(c1) + (1.2,0)$){};
		\node [treeVertex] (c3) at ($(c1) + (0,-1.2)$){};
		\node [treeVertex] (c4) at ($(c1) + (1.2,-1.2)$){};
		
		\node [fullVertex] (c12) at ($(c1)!.5!(c2)$){};
		\node [fullVertex] (c13) at ($(c1)!.5!(c3)$){};
		\node [fullVertex] (c34) at ($(c3)!.5!(c4)$){};
		\node [fullVertex] (c24) at ($(c2)!.5!(c4)$){};
		
		\draw (c1) -- (c12);
		\draw (c12) -- (c2);
		\draw (c1) -- (c13);
		\draw (c13) -- (c3);
		\draw (c3) -- (c34);
		\draw (c34) -- (c4);
		\draw (c4) -- (c24);
		\draw (c24) -- (c2);
		
		\draw [>-](c2) -- + (45:.8);
		\draw [>-](c1) -- + (135:.8);
		\draw [>-](c3) -- + (225:.8);
		\draw [>-](c4) -- + (315:.8);
		
		\node at (4.8,-4.5) {$+ \ \left \langle 6 \right \rangle  $};
		
		\node [treeVertex] (c1) at (5.8,-3.9){};
		\node [treeVertex] (c2) at ($(c1) + (1.2,0)$){};
		\node [treeVertex] (c3) at ($(c1) + (0,-1.2)$){};
		\node [treeVertex] (c4) at ($(c1) + (1.2,-1.2)$){};
		
		\node [fullVertex] (c12) at ($(c1)!.5!(c2)$){};
		\node [fullVertex] (c13) at ($(c1)!.5!(c3)$){};
		\node [fullVertex] (c34) at ($(c3)!.5!(c4)$){};
		\node [fullVertex] (c24) at ($(c2)!.5!(c4)$){};
		\node [fullVertex] (c23) at ($(c2)!.5!(c3)$){};
		
		\draw (c1) -- (c12);
		\draw (c12) -- (c2);
		\draw (c1) -- (c13);
		\draw (c13) -- (c3);
		\draw (c3) -- (c34);
		\draw (c34) -- (c4);
		\draw (c4) -- (c24);
		\draw (c24) -- (c2);
		\draw (c2) -- (c23);
		\draw (c23) -- (c3);

		\draw [>-](c2) -- + (45:.8);
		\draw [>-](c1) -- + (135:.8);
		\draw [>-](c3) -- + (225:.8);
		\draw [>-](c4) -- + (315:.8);
		
		\node at (8,-4.5) {$+$};
		
		\node [treeVertex] (c1) at (9,-3.9){};
		\node [treeVertex] (c2) at ($(c1) + (1.2,0)$){};
		\node [treeVertex] (c3) at ($(c1) + (0,-1.2)$){};
		\node [treeVertex] (c4) at ($(c1) + (1.2,-1.2)$){};
		
		\node [fullVertex] (c12) at ($(c1)!.5!(c2)$){};
		\node [fullVertex] (c13) at ($(c1)!.5!(c3)$){};
		\node [fullVertex] (c34) at ($(c3)!.5!(c4)$){};
		\node [fullVertex] (c24) at ($(c2)!.5!(c4)$){};
		\node [fullVertex] (c23) at ($(c2)!.3!(c3)$){};
		\node [fullVertex] (c14) at ($(c1)!.3!(c4)$){};
		
		\draw (c1) -- (c12);
		\draw (c12) -- (c2);
		\draw (c1) -- (c13);
		\draw (c13) -- (c3);
		\draw (c3) -- (c34);
		\draw (c34) -- (c4);
		\draw (c4) -- (c24);
		\draw (c24) -- (c2);
		\draw (c2) -- (c23);
		\draw (c23) -- (c3);
		\draw (c1) -- (c14);
		\draw (c14) -- (c4);

		\draw [>-](c2) -- + (45:.8);
		\draw [>-](c1) -- + (135:.8);
		\draw [>-](c3) -- + (225:.8);
		\draw [>-](c4) -- + (315:.8);

	\end{tikzpicture}
	\caption{Connected amplitudes of a theory fulfilling \cref{rec_bn} with $u=1$. A prefactor $\left \langle j \right \rangle $ indicates the presence of in total $j$ graphs of the corresponding topology where external edges are permuted.  }
	\label{fig_rec_npoint}
\end{figure}
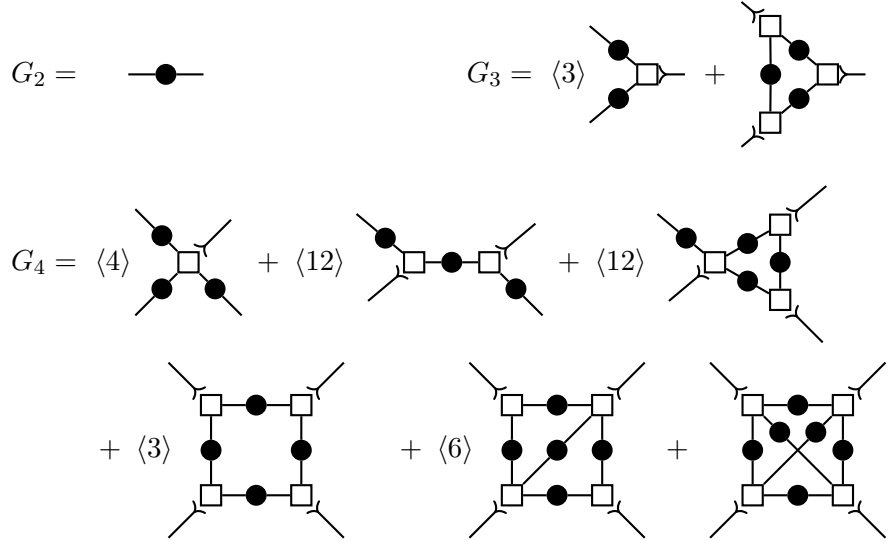

\section{1PI counterterms}\label{sec_1PI}
Knowing the meta-counterterms $C_{n,k}^{(l)}$, which cancel the divergences of connected amplitudes, we can reconstruct the amputated 1PI counterterms which we call $c_{n,k}^{(l)}$. The indices $n,l,k$ have the same meaning as for the meta-counterterms in \cref{sec_divergences}. Eventually, the sum $c_n^{(l)} := \sum_{k=0}^n c_{n,k}^{(l)}$ represents the $l$-loop counterterm to the 1PI $n$-point function, i.e. this is the object which normally is called counterterm in a local quantum field theory. 

\subsection{1PI counterterms with l=1 loops}
The 1PI one-loop graph of the 2-point-function coincides with the connected 2-point function \cref{C22l}, therefore 
\begin{align}\label{c2_1loop}
C_2^{(1)}= C_{2,2}^{(1)} = b_3^2 s^2 \frac{M^{(1)}_{\text{div}}}{2} &= c_2^{(1)} .
\end{align}

For the 3-point function, the connected 3-point divergence is the product of the 1PI 3-point divergence $c_3^{(1)}$ and three adjacent connected 2-point divergences. To one-loop order, this product can contain only  one divergent term  in total, either $c_3^{(1)}$ or one of the propagator corrections, consequently
\begin{align}\label{c3_1loop_ansatz}
C_3^{(1)} &= c_3^{(1)} + \langle 3 \rangle iv_3 \frac i {s_1} c_2^{(1)} (s_1).
\end{align}

First consider the case where all external legs are onshell, i.e. $C_{3,0}^{(1)}$. Then the meta-counterterms $C_3^{(1)}$ and $C_2^{(1)}$ vanish due to \cref{Cn0} and \cref{c3_1loop_ansatz} simplifies to
\begin{align}\label{c30_1loop}
0 &= c_{3,0}^{(1)} + 0.
\end{align}
Now let one of the legs be offshell. The connected counterterm $C^{(1)}_{3,1}$ vanishes due to \cref{Cn1} but one of the terms $C_2^{(1)}$ in \cref{c3_1loop_ansatz} survives, so
\begin{align}\label{c31_1loop_ansatz}
C_{3,1}^{(1)} &= 0 = c_{3,1}^{(1)} + \langle 3 \rangle  b_3(-i  s_1) \frac{i}{s_1} c_{2,2}^{(1)}(s_1).
\end{align}
This implies that 
\begin{align}\label{c31_1loop}
c_{3,1}^{(1)} &= - \left \langle 3 \right \rangle b_3 c_{2,2}^{(1)} (s_1) =  -b_3^3  \left( s_1^2 + s_2^2 + s_3^2\right) \frac { M^{(1)}_{\text{div}}}2.
\end{align}
For the higher $n$-point functions these arguments become increasingly cumbersome; they are much more transparent if carried out in a graphical manner, see \cref{c3_1loop_picture}.
\begin{figure}[htbp]
	\centering
	\begin{tikzpicture} 
		
		\node at (-2,-1){\cref{c2_1loop}:};
		
		\node[treeCounterVertex, label={[label distance=-2mm]5:{$\scriptstyle (1)$}}] (c) at (0,-1){};
		\draw [>-](c) --++(0:.4);
		\draw [>-](c) --++(180:.4);
		
		\node at (1,-1) {$=$};
		
		\node[counterVertex , label={[label distance=-2mm]5:{$\scriptstyle (1)$}} ] (c) at (2,-1){};
		\draw [>-](c) --++(0:.4);
		\draw [>-](c) --++(180:.4);
		
		\node [anchor=west] at (3,-1) {$=b_3^2 s^2 \frac{M^{(1)}_{\text{div}}}{2}$};

		\node at (-2,-3){\cref{c31_1loop_ansatz}:};
		
		\node [treeCounterVertex, label={[label distance=-2mm]5:{$\scriptstyle (1)$}}  ] (c) at (0,-3) {};
		\draw  (c) --++ (0:.4);
		\draw [-|] (c) --++ (-120:.4);
		\draw [-|] (c) --++ (120:.4);
		
		\node at (0,-3.7){$\underbrace{\hspace{1cm}}_{=0}$};
		
		\node at (1,-3) {$=$};
		
		\node [counterVertex, label={[label distance=-2mm]5:{$\scriptstyle (1)$}}  ] (c) at (2,-3) {};
		\draw  (c) --++ (0:.4);
		\draw [-|] (c) --++ (-120:.4);
		\draw [-|] (c) --++ (120:.4);
		
		\node at (3.5,-3) {$+$};
		
		\node [treeVertex] (c1) at (4.5,-3){};
		\node [counterVertex, label={[label distance=-2 mm]60:{$\scriptstyle (1)$}}] (c2) at ($(c1)+ (0:.7)$){};
		\draw   (c1) -- (c2);
		
		\draw [-|] (c1) -- + (120:.5);
		\draw [-|] (c1) -- + (-120:.5);
		\draw (c2) -- + (0:.5);
		
		\node at (6.5,-3) {$+$};
		
		\node [treeVertex] (c1) at (8,-2.6){};
		\node [counterVertex, label={[label distance=-2mm]120:{$\scriptstyle (1)$}}  ] (c2) at ($(c1)+ (-120:.7)$){};
		\draw   (c1) -- (c2);
		\draw [-|](c1) -- + (120:.5);
		\draw (c1) -- + (0:.5);
		\draw [-|](c2) -- + (-120:.5);
		
		\node [rotate=60] at (8.1,-3.5){$\underbrace{\hspace{.8cm}}_{=0}$};	
		
		\node at (9,-3) {$+$};
		
		\node [treeVertex] (c1) at (10.5,-3.4){};
		\node [counterVertex, label={[label distance=-2mm]30:{$\scriptstyle (1)$}} ] (c2) at ($(c1) + (120:.7)$){};
		\draw   (c1) -- (c2);
		\draw [-|] (c1) -- + (-120:.5);
		\draw (c1) -- + (0:.5);
		\draw [-|] (c2) -- + (120:.5);
		
		\node [rotate=-60] at (9.7,-3){$\underbrace{\hspace{.8cm}}_{=0}$};

		\node at (-2,-5){\cref{c31_1loop}:};
		
		\node [counterVertex, label={[label distance=-2mm]5:{$\scriptstyle (1)$}}  ] (c) at (0,-5) {};
		\draw  (c) --++ (0:.4);
		\draw [-|] (c) --++ (-120:.4);
		\draw [-|] (c) --++ (120:.4);

		\node at (1,-5) {$=-$};

		\node [treeVertex] (c1) at (2,-5){};
		\node [counterVertex, label={[label distance=-2 mm]60:{$\scriptstyle (1)$}}] (c2) at ($(c1)+ (0:.7)$){};
		\draw  [>-<] (c1) -- (c2);
		
		\draw [-|] (c1) -- + (120:.5);
		\draw [-|] (c1) -- + (-120:.5);
		\draw [>-] (c2) -- + (0:.5);
		
		\node [anchor=west] at (3.8,-5) {$=-b_3 \cdot b_3^2 s_1^2 \frac{M^{(1)}_{\text{div}}}{2}$};
		
	\end{tikzpicture}
	\caption{Graphical representation for the computation of $c_{3,1}^{(1)}$. The perpendicular line indicates an external edge which must not be cancelled by the adjacent vertex.  }
	\label{c3_1loop_picture}
\end{figure}
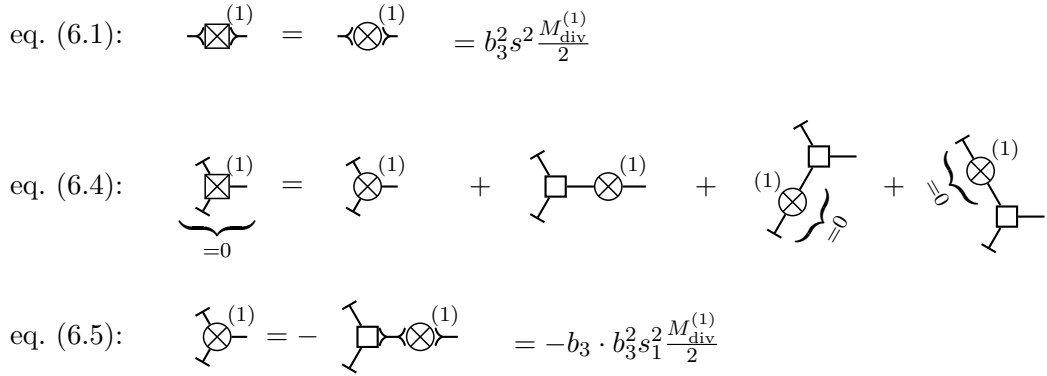

For $c_{3,2}^{(1)}$, the meta-counterterm $C_{3,2}^{(1)}$ does not vanish, see \cref{C32l}. The construction of $c_{3,2}^{(1)}$ is shown in \cref{c32_1loop_picture}, it yields 
\begin{align}\label{c32_1loop}
	c_{3,2}^{(1)} &= C_{3,2}^{(1)} - \left \langle 3 \right \rangle (-ib_3 s_1) \frac i {s_2} c_{2,2}^{(1)}(s_2) - \left \langle 3 \right \rangle (-ib_3 s_2) \frac i {s_1} c_{2,2}^{(1)}(s_1) \nonumber \\
	&= \left \langle 6 \right \rangle   b_{3} \left( b_{4}-b_3^2\right)  s_1 s_2 \frac{M^{(1)}_{\text{div}}}{2} = b_3 \left( b_{4}-b_3^2\right)    \frac{M^{(1)}_{\text{div}}}{2} \cdot 2 E_2 \left( s_1,s_2,s_3 \right) .
\end{align}
Note that we assumed that $M^{(1)}_{\text{div}}$ be independent of momenta only to simplify notation, the result would be $\propto b_3 (b_4-b_3^2)$ regardless.

\begin{figure}[htbp]
	\centering
	\begin{tikzpicture} 
		
		\node at (-2,0) {$-\renop  $};
		
		\draw [semithick]  (-1.3,-.8) to [ncbar=.1](-1.3,.8);
		
		\node [treeVertex] (c1) at (-.7,-.4){};
		\node [treeVertex  ] (c2) at ($(c1)+ (60:.8)$){};
		\draw [bend angle =40, bend left] (c1) to (c2);
		\draw [bend angle =40, bend right] (c1) to (c2);
		\draw [>-](c1) -- + (-120:.5);
		\draw (c2) -- + (0:.5);
		\draw [>-](c2) -- + (120:.5);
		
		\node at (.5,0) {$+$};
		
		\node [treeVertex] (c1) at (1.3,.4){};
		\node [treeVertex] (c2) at ($(c1) + (-60:.8)$){};
		\draw [bend angle =40, bend left] (c1) to (c2);
		\draw [bend angle =40, bend right] (c1) to (c2);
		
		\draw [>-](c1) -- + (120:.5);
		\draw (c2) -- + (0:.5);
		\draw [>-](c2) -- + (-120:.5);

		\draw [semithick] (2.5,-.8) to [ncbar=-.1] (2.5,.8);
		
		\node at (3,0) {$=$};
		
		\node [treeCounterVertex, label={[label distance=-2mm]5:{$\scriptstyle (1)$}}  ] (c) at (4,0) {};
		\draw [>-] (c) --++ (0:.5);
		\draw [>-] (c) --++ (120:.5);
		\draw [-|] (c) --++ (-120:.4);

		\node at (5,0) {$=$};
		
		\node [counterVertex, label={[label distance=-2mm]5:{$\scriptstyle (1)$}}  ] (c) at (6,0) {};
		\draw [>-] (c) --++ (0:.5);
		\draw [>-] (c) --++ (120:.5);
		\draw [-|] (c) --++ (-120:.4);
		
		\node at (7.5,0) {$+$};
		
		\node [treeVertex] (c1) at (8.5,0){};
		\node [counterVertex, label={[label distance=-2 mm]60:{$\scriptstyle (1)$}}] (c2) at ($(c1)+ (0:.7)$){};
		\draw   [-<](c1) -- (c2);
		
		\draw [-|] (c1) -- + (-120:.5);
		\draw [>-] (c1) -- + (120:.5);
		\draw [>-] (c2) -- + (0:.5);

		\node at (10,0) {$+$};
		
		\node [treeVertex] (c1) at (11.5,-.4){};
		\node [counterVertex, label={[label distance=-2mm]30:{$\scriptstyle (1)$}} ] (c2) at ($(c1) + (120:.7)$){};
		\draw   [-<](c1) -- (c2);
		\draw [-|] (c1) -- + (-120:.5);
		\draw [>-] (c1) -- + (0:.5);
		\draw [>-] (c2) -- + (120:.5);

	\end{tikzpicture}
	\caption{Graphical notation for the computation of $c_{3,2}^{(1)}$. The meta-counterterm has been taken from \cref{C32_picture}.}
	\label{c32_1loop_picture}
\end{figure}
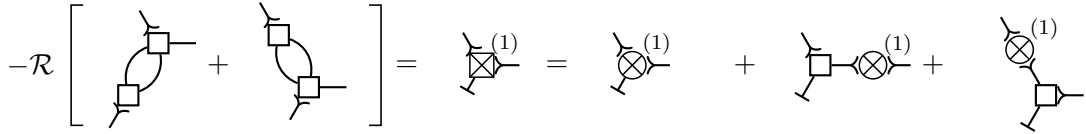

Finally, $C_{3,3}^{(1)}=0$ and there is no divergent connected graph that cancels three external edges at one loop, therefore
\begin{align}\label{c33_1loop}
	c_{3,3}^{(1)}=0.
\end{align}

\begin{figure}[htbp]
	
	\begin{tikzpicture}

	\node [treeCounterVertex, label={[label distance=-2mm]5:{$\scriptstyle (1)$}}  ] (c) at (-1,-3) {};
	\draw [-|]  (c) --++ (0:.4);
	\draw [-|] (c) --++ (-90:.4);
	\draw [-|] (c) --++ (90:.4);
	\draw [-|] (c) --++ (180:.4);

	\node at (.5,-3) {$= 0 = $};
	
	\node [counterVertex, label={[label distance=-2mm]5:{$\scriptstyle (1)$}}  ] (c) at (2,-3) {};
	\draw [-|] (c) --++ (0:.4);
	\draw [-|] (c) --++ (-90:.4);
	\draw [-|] (c) --++ (90:.4);
	\draw [-|] (c) --++ (180:.4);

	\node at (3.2,-3) {$+\ \langle 6 \rangle$};
	
	\node [treeVertex] (c1) at (5,-3.4){};
	\node [counterVertex, label={[label distance=-2mm]20:{$\scriptstyle (1)$}}  ] (c2) at ($(c1)+ (135:1)$){};
	\draw  [>-<<] (c1) -- (c2);
	\draw [-|] (c1) --++ (0:.4);
	\draw [-|] (c1) --++ (-90:.4);
	\draw [-|] (c2) --++ (90:.4);
	\draw [-|] (c2) --++ (180:.4);
	
	\node at (6,-3) {$+\ \langle 3 \rangle$};
	
	\node [treeVertex] (c1) at (8,-3.4){};
	\node [counterVertex, label={[label distance=-2mm]30:{$\scriptstyle (1)$}} ] (c2) at ($(c1) + (135:.7)$){};
	\node [treeVertex] (c3) at ($(c1) + (135:1.4)$){};
	\draw  [>-<] (c1) -- (c2);
	\draw  [>-<] (c2) -- (c3);
	\draw [-|] (c1) --++ (0:.4);
	\draw [-|] (c1) --++ (-90:.4);
	\draw [-|] (c3) --++ (90:.4);
	\draw [-|] (c3) --++ (180:.4);

	\end{tikzpicture}
	\caption{Construction of the onshell 4-point meta-counterterm from metavertices 1PI counterterms.}
	\label{c4_1loop_picture}
\end{figure}
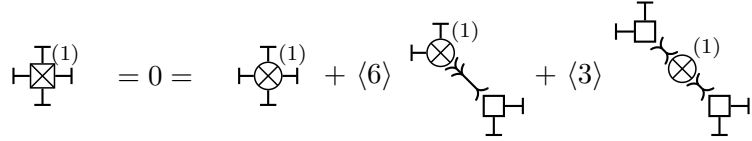

\Cref{c4_1loop_picture} indicates the contributions to the onshell 4-point meta-counterterm $C_4^{(1)}$. Using \cref{c31_1loop}, the last two summands partially cancel and what remains is 
\begin{align*}
c^{(1)}_{4,0} &= -\langle 3 \rangle c_{3,1}^{(1)} \frac{i}{s_{1+2}} (-ib_3 s_{1+2}) = -b_3^4 \left( s_{1+2}^2 + s_{1+3}^2 + s_{1+4}^2 \right) \frac  { M^{(1)}_{\text{div}}}2 \frac{1}{\epsilon}.
\end{align*}
With one external leg offshell, the meta-counterterm \cref{Cn1_picture} still vanishes and we obtain \cref{c41_1loop_picture}. 

\begin{figure}[htbp]
	
	\begin{tikzpicture}

\node [treeCounterVertex, label={[label distance=-2mm]5:{$\scriptstyle (1)$}}  ] (c) at (-1,-1) {};
\draw [>-]  (c) --++ (0:.4);
\draw [-|] (c) --++ (-90:.4);
\draw [-|] (c) --++ (90:.4);
\draw [-|] (c) --++ (180:.4);

	\node at (1,-1) {$=0=$};
	
	\node [counterVertex, label={[label distance=-2mm]5:{$\scriptstyle (1)$}}  ] (c) at (2,-1) {};
	\draw [>-]  (c) --++ (0:.4);
	\draw  [-|](c) --++ (-90:.4);
	\draw  [-|](c) --++ (90:.4);
	\draw [-|] (c) --++ (180:.4);
	
	\node at (3.5,-1) {$+\langle 3 \rangle$};
	
	\node [treeVertex] (c1) at (5.5,-1.4){};
	\node [counterVertex, label={[label distance=-2mm]30:{$\scriptstyle (1)$}} ] (c2) at ($(c1) + (135:.7)$){};
	\node [treeVertex] (c3) at ($(c1) + (135:1.4)$){};
	\draw [-<]  (c1) -- (c2);
	\draw  [>-<]  (c2) -- (c3);
	\draw  [>-](c1) --++ (0:.4);
	\draw  [-|] (c1) --++ (-90:.4);
	\draw  [-|](c3) --++ (90:.4);
	\draw [-|] (c3) --++ (180:.4);
	
	\node at (6.5,-1) {$+\ \langle 3 \rangle$};
	
	\node [treeVertex] (c1) at (8,-1.1){};
	\node [counterVertex, label={[label distance=-2mm]20:{$\scriptstyle (1)$}}  ] (c2) at ($(c1)+ (135:.7)$){};
	\draw  [-<<] (c1) -- (c2);
	\draw  [>-](c1) --++ (0:.4);
	\draw  [-|](c1) --++ (-90:.4);
	\draw [-|] (c2) --++ (90:.4);
	\draw [-|] (c2) --++ (180:.4);
	
	\node at (9,-1) {$+\ \langle 3 \rangle$};
	
	\node  [counterVertex, label={[label distance=-2mm]20:{$\scriptstyle (1)$}}  ] (c1) at (10.5,-1.1){};
	\node [treeVertex](c2) at ($(c1)+ (135:.7)$){};
	\draw   [>-<] (c1) -- (c2);
	\draw [>-] (c1) --++ (0:.4);
	\draw  [-|](c1) --++ (-90:.4);
	\draw  [-|](c2) --++ (90:.4);
	\draw  [-|](c2) --++ (180:.4);

	\node [treeCounterVertex, label={[label distance=-2mm]5:{$\scriptstyle (1)$}}  ] (c) at (-1,-3) {};
	\draw [>>-]  (c) --++ (0:.5);
	\draw [-|] (c) --++ (-90:.4);
	\draw [-|] (c) --++ (90:.4);
	\draw [-|] (c) --++ (180:.4);

	\node at (1,-3) {$=0=$};
	
	\node [counterVertex, label={[label distance=-2mm]5:{$\scriptstyle (1)$}}  ] (c) at (2,-3) {};
	\draw [>>-]  (c) --++ (0:.5);
	\draw  [-|](c) --++ (-90:.4);
	\draw  [-|](c) --++ (90:.4);
	\draw [-|] (c) --++ (180:.4);
	
	\node at (3.5,-3) {$+ \ \langle 4 \rangle $};
	
	\node [treeVertex] (c1) at (4.5,-3){};
	\node [counterVertex, label={[label distance=-2 mm]60:{$\scriptstyle (1)$}}] (c2) at ($(c1)+ (0:.7)$){};
	\draw   [>-<](c1) -- (c2);
	
	\draw  [>-](c2) --++ (0:.4);
	\draw [-|] (c1) --++ (-90:.4);
	\draw [-|] (c1) --++ (90:.4);
	\draw  [-|](c1) --++ (180:.4);

	\node at (6.5,-3) {$+\ \langle 3 \rangle$};
	
	\node  [counterVertex, label={[label distance=-2mm]20:{$\scriptstyle (1)$}}  ] (c1) at (8,-3.2){};
	\node [treeVertex](c2) at ($(c1)+ (135:.7)$){};
	\draw   [-<] (c1) -- (c2);
	\draw [>>-] (c1) --++ (0:.5);
	\draw  [-|](c1) --++ (-90:.4);
	\draw  [-|](c2) --++ (90:.4);
	\draw  [-|](c2) --++ (180:.4);

	\end{tikzpicture}
	\caption{4-valent meta-counterterm with one external leg offshell, which can be cancelled once or twice, indicated by arrows. Edges with perpendicular lines must not be cancelled.}
	\label{c41_1loop_picture}
\end{figure}
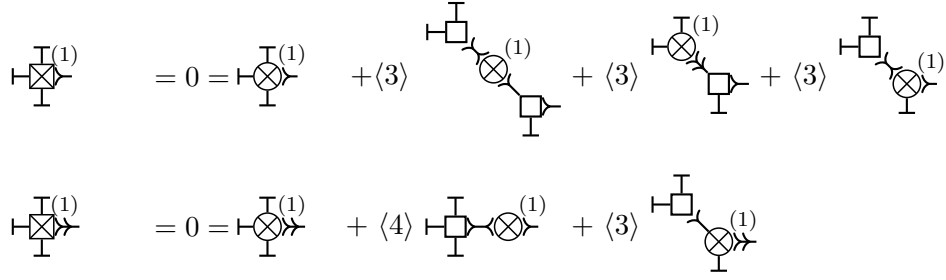

\Cref{c31_1loop} implies that the second and third graph cancel and therefore
\begin{align*} 
c_{4,1}^{(1)} &= \left \langle 4 \right \rangle \left(  b_4-3b_3^2\right)  b_3^2 \frac 12 M^{(1)}_{\text{div}}s_1^2    - \left \langle 4 \right \rangle  2 b_3^2 \left( b_4-b_3^2 \right) \frac 12 M^{(1)}_{\text{div}} \left( s_{1+2} + s_{1+3} + s_{1+4}  \right) s_1.
\end{align*}
The meta-counterterms $C_n^{(l)}$ are always symmetric polynomials in the external momenta. Conversely, the 1PI counterterms   can also contain inner offshell variables $s_{i+j+\ldots}$.  This is completely analogous to the 1PI vertices $iv_n$ \cref{vn} which, unlike the metavertices $iV_n$, must contain such internal offshell variables in order to fulfil \cref{ordinary_recursive}.

By power-counting, the one-loop 1PI counterterm $c_{n}^{(1)}$ is proportional to two powers of offshell variables, hence, five different dependencies are possible: Square of an external offshell variable $s_j^2$, square of an internal one $s_{i+j+\ldots}^2$, two different external ones $s_i s_j$, two different internal ones $s_{i+j+\ldots} s_{k+l+\ldots}$ or a mixture of both types $s_j s_{k+l+\ldots}$. In every case, symmetric sums are understood. 

\begin{lemma}\label{lem_c1_1}
	The summand in $c_n^{(1)}$, which is proportional to a square of an external offshell variable, is
	\begin{align*}
	  +(n-1)! a_{n-2} b_3^2 \frac{M^{(1)}_{\text{\normalfont div}}}{2} \left( s_1^2 + \ldots + s_n^2 \right),
	\end{align*}
	where $a_n$ is given by \cref{anbn}.
\end{lemma}
\begin{proof} 
	This can be proved by rigorously constructing $C_{n,1}^{(1)}$ from trees involving precisely one 1PI counterterm $c_j^{(1)}$ of valence $j\leq n$ and arbitrary many tree sums $b_k$. However, an inductive proof is shorter. The 3-valent term \cref{c31_1loop} has the required form.
	
	Assume the statement holds for any valence $j<n$. Consider the divergence of the $n$-valent connected amplitude, where one external edge $e$ is double-cancelled. It is given by the following contributions, shown in \cref{cn1_1loop_picture}:
	\begin{enumerate}
		\item A sum over all trees where a counterterm $c_{j,1}^{(1)}$ double-cancels $e$. There is a suitable number of tree sums $b_k$ connected to $c_{j,1}^{(1)}$. 
		\item A sum over all trees where a counterterm $c_2^{(1)}$ is inserted into $e$, thus double-cancelling it.
		\item A single vertex $iv_n$ connected to a counterterm $c_2^{(1)}$ in the edge $e$. To produce an overall factor $s_e^2$, the vertex $iv_n$ must cancel $e$, therefore only the part $\propto s_e$ is relevant. 
		\item The counterterm $c_{n,1}^{(1)}$ double-cancelling $e$.
	\end{enumerate}
	The first two cases cancel each other because of the induction hypothesis. On the other hand, the overall sum is zero since it is $C_{n,1}^{(1)}=0$ by \cref{Cn1}. Therefore, 
	\begin{align*}
		0 &=  iv_n\Big|_{\text{summand $\propto s_e$}} \frac{i}{s_e} c_{2}^{(1)} (s_e) + c_{n,1}^{(1)}\Big|_{\text{summand $\propto s_e^2$}} .
	\end{align*}
Using \cref{vn_leading} and \cref{c2_1loop} and summing over all $n$ orientations produces the statement.

\end{proof}

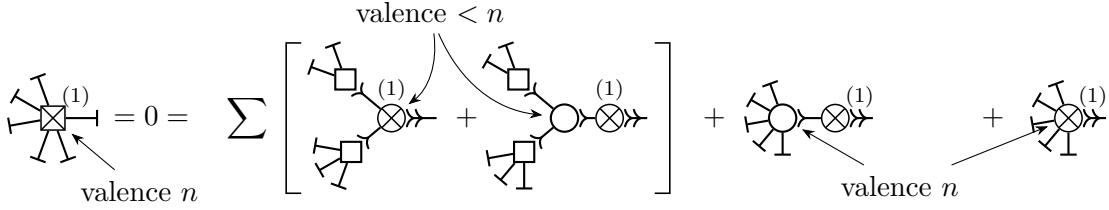
\begin{figure}[htbp]
	
	\begin{tikzpicture}
		
		\node [treeCounterVertex, label={[label distance=-2mm]5:{$\scriptstyle (1)$}}] (c) at (-1,0){};
		
		\draw [-|]  (c) --++ (0:.6);
		\draw  [-|](c) --++ (-70:.6);
		\draw  [-|](c) --++ (-110:.6);
		\draw [-|] (c) --++ (190:.6);
		\draw [-|] (c) --++ (110:.6);
		\draw [-|] (c) --++ (150:.6);
		
		\node (tx) at ($(c)+(-40:1.5)$) {valence $n$};
		\draw [thin, -Stealth, bend angle = 20,shorten >=1mm ] (tx) to (c);

		\node [anchor=east] at (2,0) {$= 0  =\quad \displaystyle \sum  $};

		\draw [semithick]  (2.3,-1) to [ncbar=.1](2.3,1);

		\node [counterVertex, label={[label distance=-1mm]90:{$\scriptstyle (1)$}}  ] (c) at (3.5,0) {};
		
		\node [treeVertex] (b1) at ($(c)+(140:.8)$) {};
		\node [treeVertex] (b2) at ($(c)+(220:.7)$) {};

		\draw [>>-]  (c) --++ (0:.6);
		\draw [-<] (c) -- (b1);
		\draw [-<] (c) -- (b2);

		\draw [-|] (b1) --++ (110:.5);
		\draw [-|] (b1) --++ (160:.5);

		\draw [-|] (b2) --++ (170:.5);
		\draw [-|] (b2) --++ (210:.5);
		\draw [-|] (b2) --++ (250:.5);
		
		\node (tx) at ($(c)+(70:1.5)$) {valence $<n$};
		\draw [thin, -Stealth, bend left,bend angle = 20,shorten >=1mm ] (tx) to  (c);
		
		\node at (4.5,0) {$+$};
		
		\node [ diffVertex ] (c) at (5.8,0) {};
		\node [counterVertex, label={[label distance=-1mm]90:{$\scriptstyle (1)$}}] (b) at ($(c)+(0:.6)$) {};
		\node [treeVertex] (b1) at ($(c)+(140:.8)$) {};
		\node [treeVertex] (b2) at ($(c)+(220:.7)$) {};
		
		\draw [>-]  (c) --(b);
		\draw [>>-]  (b) --++ (0:.5);
		\draw [-<] (c) -- (b1);
		\draw [-<] (c) -- (b2);
		
		\draw [-|] (b1) --++ (110:.5);
		\draw [-|] (b1) --++ (150:.5);
		\draw [-|] (b2) --++ (190:.5);
		\draw [-|] (b2) --++ (230:.5);
		\draw [-|] (b2) --++ (270:.5);
		
		 \draw [thin, -Stealth, bend right, bend angle = 20,shorten >=1mm ] (tx) to  (c);
		
		\draw [semithick] (7,-1) to [ncbar=-.1] (7,1);

		\node [anchor=west] at (7.5,0){$+ $};
		
		\node [ diffVertex] (v) at (8.7,0) {};
		\node [counterVertex, label={[label distance=-2mm]10:{$\scriptstyle (1)$}} ] (c) at ($(v)+(0:.7)$) {};
		
		\draw [>-] (v) --  (c);
		\draw [>>-]  (c) --++ (0:.5);
		\draw [-|] (v) --++ (110:.5);
		\draw [-|] (v) --++ (150:.5);
		\draw [-|] (v) --++ (190:.5);
		\draw [-|] (v) --++ (230:.5);
		\draw [-|] (v) --++ (270:.5);
		
		\node (tx) at ($(v)+(330:1.8)$) {valence $n$};
		\draw [thin, -Stealth, bend angle = 20,shorten >=1mm ] (tx) to  (v);
		
		\node at (11.5,0){$+$};
		
		\node [counterVertex,label={[label distance=-2mm]5:{$\scriptstyle (1)$}} ] (c) at (12.5,0) {};
		
		\draw [>>-]  (c) --++ (0:.5);
		\draw [-|] (c) --++ (110:.5);
		\draw [-|] (c) --++ (150:.5);
		\draw [-|] (c) --++ (190:.5);
		\draw [-|] (c) --++ (230:.5);
		\draw [-|] (c) --++ (270:.5);
		
		\draw [thin, -Stealth, bend angle = 20,shorten >=1mm ] (tx) to  (c);

	\end{tikzpicture}
	\caption{ Proof of \cref{lem_c1_1}. The bracket vanishes by induction hypothesis. }
	\label{cn1_1loop_picture}
\end{figure}

\begin{lemma}\label{lem_c1_0}
	In the 1PI counterterm $c^{(1)}_n$, the contributions proportional to $s_p^2$, the square of the offshell variable of some partition $p$ of the external momenta, is
	\begin{align*}
	 b_3^2\frac {M^{(1)}_{\text{\normalfont div}}}{2} \frac 12 \sum_{k=2}^{n-2} \sum_{p\in Q^{(n,k)}} k! a_{k-1}   a_{n-k-1} (n-k)!  s^2_{p}
	\end{align*}
	where $Q^{(n,k)}$ denotes the set of all possibilities to choose $k$ out of $n$ external legs.
\end{lemma}
\begin{proof}
	Analogous to the proof of \cref{lem_c1_1}. 
	Setting all external legs onshell, the sum of all connected graph divergences vanishes by \cref{Cn0}, $C_{n,0}^{(1)}=0$ . 
	
	The tree point counterterm is $c^{(1)}_{3,0}=0$ since there are no internal edges at all. Assume now that for all $j<n$, the counterterms $c^{(1)}_{j,0}$ are chosen such that no factor $s_e^2$ remains in the sum of all connected trees with $j$ external legs. Following \cref{cn0_1loop_picture}, at $n$ external legs, there are four structures in the sum of connected trees, each of which involves precisely one counterterm $c^{(1)}$: 
	\begin{enumerate}
		\item Trees $T$, where at least one vertex is connected through an edge $e$ which is not triple-cancelled. Then, triple-cancellation can only occur in the original tree $T$, but it has valence $<n$ and is, when summed over all trees, free of triple-cancellation by induction hypothesis. All trees with more than two vertices fall in this category by powercounting.
		\item A single   counterterm, where one vertex is added through a new edge $e$ which is triple-cancelled. There can be only one such edge $e$ by power-counting and it represents a partition of the external legs into precisely two sets, one on each side.
		\item Two single vertices cancelling the edge $e$ and a counterterm $c_2^{(1)}$ in that same edge $e$.
		\item A new counterterm $c^{(1)}_{n,0}$.
	\end{enumerate}
	
	Since (1) vanishes by induction hypothesis, the last three contributions have to cancel. We know from \cref{lem_c1_1} that cases (2) and (3) add up to zero, only that in the present construction, case (3) has two interchangeable vertices and hence an overall factor $\frac 12$. Consequently, $(3) = -\frac 12 \cdot (2)$ and $(2)+(3) = \frac 12 (2)$.

	The counterterm $c^{(1)}_{n,0}$ must be chosen to be the sum of all ways to partition the $n$ external legs into two sets where the first set has $k-1$ elements and amounts to a $k$-valent counterterm $c_{k,1}^{(1)}$ and the second set has $n-k+1$ elements connected to a vertex $i v_{n-k+2}$. The set of all such partitions is denoted $Q^{(n,k-1)}$. For the vertex $iv_{n-k+2}$,  the  only relevant summand ist the one cancelling the edge $e$ from \cref{vn_leading}, for $c_{k,1}^{(1)}$ the only relevant one is the part double-cancelling $e$ from \cref{lem_c1_1}, because all other contributions do not produce a triple-cancelled edge $e$.
	\begin{align*}
	c_{n,0}^{(1)} &= - \frac 12 \cdot \sum_{k=3}^{n-1} c^{(1)}_{k,1} \frac{i}{s_e} iv_{n-k+2}\\
	&= -\frac 12 \sum_{k=3}^{n-1} \sum_{p\in Q^{(n,k-1)}} (k-1)! a_{k-2} b_3^2 \frac {M^{(1)}_{\text{div}}}2 \cdot s_e^2 \cdot a_{n-k} (n-k+1)!  
	\end{align*}
	
	Note that this proof is similar to the derivation of $b_n$ from $iv_j$ in \cite{balduf_perturbation_2020}, but backwards.
\end{proof}

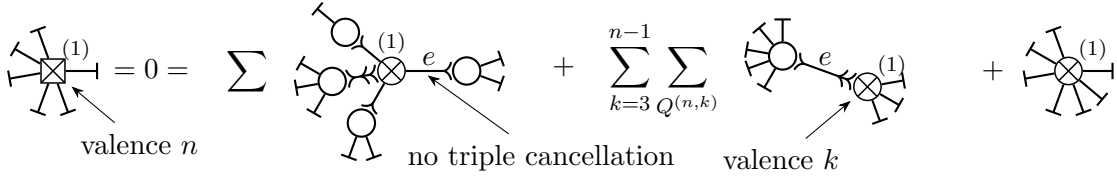
\begin{figure}[htbp]
	
	\begin{tikzpicture}
	
	\node [treeCounterVertex, label={[label distance=-2mm]5:{$\scriptstyle (1)$}}] (c) at (-1,0){};
	
	\draw [-|]  (c) --++ (0:.6);
	\draw  [-|](c) --++ (-70:.6);
	\draw  [-|](c) --++ (-110:.6);
	\draw [-|] (c) --++ (190:.6);
	\draw [-|] (c) --++ (110:.6);
	\draw [-|] (c) --++ (150:.6);
	
	\node (tx) at ($(c)+(-40:1.5)$) {valence $n$};
	\draw [thin, -Stealth, bend angle = 20,shorten >=1mm ] (tx) to (c);

	\node [anchor=east] at (2,0) {$= 0  =\quad \displaystyle \sum  $};

	\node [counterVertex, label={[label distance=-1mm]90:{$\scriptstyle (1)$}}  ] (c) at (3.5,0) {};
	\node [diffVertex] (b) at ($(c)+(0:1)$) {};

	\node [diffVertex] (b1) at ($(c)+(140:.8)$) {};
	\node [diffVertex] (b2) at ($(c)+(190:.8)$) {};
	\node [diffVertex] (b3) at ($(c)+(240:.8)$) {};
	
	\draw [-<]  (c) --  node[label={[label distance = -2mm]above:$e$}]{}(b);
	\draw [-<] (c) -- (b1);
	\draw [>>-<] (c) -- (b2);
	\draw [-<] (c) -- (b3);
	
	\draw  [-|](b) --++ (20:.5);
	\draw  [-|](b) --++ (-20:.5);
	\draw [-|] (b1) --++ (140:.5);
	\draw  [-|](b3) --++ (-70:.5);
	\draw  [-|](b3) --++ (-110:.5);
	\draw [-|] (b2) --++ (190:.5);
	\draw [-|] (b2) --++ (150:.5);
	\draw [-|] (b2) --++ (240:.5);
	
	\node (tx) at ($(c)+(-30:2.3)$) {no triple cancellation};
	
	\draw [thin, -Stealth, bend angle = 20,shorten >=1mm ] (tx) to  ($(c)+(0:.4)$);

	\node [anchor=west] at (5.5,0){$+\quad \displaystyle \sum_{k=3}^{n-1} \sum_{Q^{(n,k)}}$};
	
	\node [ diffVertex] (v) at (8.7,.2) {};
	\node [counterVertex, label={[label distance=-2mm]5:{$\scriptstyle (1)$}} ] (c) at ($(v)+(-20:1.2)$) {};
	
	\draw [>-<<] (v) -- node[label={[label distance = -2mm]above:$e$}]{} (c);
	
	\draw  [-|](c) --++ (10:.5);
	\draw  [-|](c) --++ (-30:.5);
	\draw  [-|](c) --++ (-70:.5);
	
	\draw  [-|](v) --++ (90:.5);
	\draw  [-|](v) --++ (130:.5);
	\draw [-|] (v) --++ (200:.5);
	\draw [-|] (v) --++ (170:.5);
	\draw [-|] (v) --++ (240:.5);
	
	\node (tx) at ($(c)+(220:1.5)$) {valence $k$};
	\draw [thin, -Stealth, bend angle = 20,shorten >=1mm ] (tx) to  (c);
	
		\node at (11.5,0){$+$};
	
	\node [counterVertex,label={[label distance=-2mm]5:{$\scriptstyle (1)$}} ] (c) at (12.5,0) {};
	
	\draw [-|]  (c) --++ (0:.6);
	\draw [-|]  (c) --++ (-40:.6);
	
	\draw  [-|](c) --++ (-70:.6);
	\draw  [-|](c) --++ (-110:.6);
	\draw [-|] (c) --++ (190:.6);
	\draw [-|] (c) --++ (110:.6);
	\draw [-|] (c) --++ (150:.6);

	\end{tikzpicture}
	\caption{ Proof of \cref{lem_c1_0}. The first sum involves a triple cancellation in the subtree left of $e$. This subtree has less than $n$ external legs and vanishes by induction hypothesis. }
	\label{cn0_1loop_picture}
\end{figure}

\begin{lemma}\label{lem_c1_2}
	If $b_n = \lambda^{n-2}$ then $c^{(1)}$ does not contain any summands which are proportional to $s_e \cdot s_f$, where $e\neq f$ can be external or internal offshell variables.
\end{lemma}
\begin{proof}
	The Lemma concerns in principle three different contributions:
	\begin{enumerate}
		\item Summands proportional to two distinct external edges,
		\item Summands proportional to one internal edge and one external edge,
		\item Summands proportional to two distinct internal edges.
	\end{enumerate}
	Assume case (1) is shown to vanish and consider case (2). Let $e$ be the internal edge in question, there is a factor $\propto s_e^1$ in the counterterm, that means, $e$ is double-cancelled. There are three graphs which can give rise to such contributions, each of them consist of the edge $e$ and one vertex on each of its ends: First, a vertex $c_{j,1}^{(n)}$ which double-cancels $e$, connected to some $b_k$ which does not cancel $e$. Second, a $b_j$ which cancels  $e$, then a 2-point-counterterm $c_{2,2}^{(1)}$ in $e$ and another $b_k$ which does not cancel $e$. And third, some $b_j$ which cancels $e$, connected to  $c_{n-n+2, 2}^{(1)}$ which cancels $e$ once. The first two of these add up to zero due to \cref{cn1_1loop_picture}. The third one is proportional to $c_{n-j+2,2}^{(1)}$ which is assumed to vanish. Hence, case (2) does not contribute.
	
	Similarly, case (3) does not occur if one assumes that case (1) vanishes by the repetition of the above argument for two distinct internal edges.

	It remains to show that case (1) vanishes. Use induction. If we assume that $c_{j,2}^{(1)}=0$ for all $2<j<n$ then at $n$ legs there are only three topologies which cancel two distinct external edges:
	\begin{enumerate}
		\item The counterterm $c_{n,2}^{(1)}$,
		\item A treesum $-ib_n$ with one counterterm $c_{2,2}^{(1)}$,
		\item Two treesums $-ib_j, -ib_{n-j+2}$ connected by a counterterm $c_{2,2}^{(1)}$.
	\end{enumerate}
	The latter two topologies amount to a sum over all partitions of the $n$ external edges into two nonempty distinct sets. Let $j$ be the cardinality of one of these sets, then the two treesums contribute $b_{j+1}\cdot b_{n-j+1}$, which is also valid for $j=1$ or $j=n-1$. The counterterm $c_{2,2}^{(1)}$ is known from \cref{c2_1loop}. Finally, cancelling one external edge at each of the tree sums produces a factor $(s_1 + \ldots  + s_j) \cdot (s_{j+1} + \ldots + s_n)$, where we have assumed that $\{s_1, \ldots, s_j\}$ are the $j$ edges connected to the first tree sum. In summary, the latter two topologies produce terms of the form
	\begin{align*}
		\left \langle K_j \right \rangle b_{j+1} b_{n-j+1} b_3^2 (s_1+\ldots + s_j) (s_{j+1} + \ldots + s_n) \frac{M^{(1)}_{\text{div}}}{2}
	\end{align*}
	where $\left \langle K_j \right \rangle $ denotes the number of possible ways to choose $j$ out of $n$ momenta. We have to sum over $j$ from $1$ to $n-1$ and include an overall factor $\frac 12 $ to account for the interchangeability of the two vertices. 
	
	But this is precisely the same construction which defines the one-loop meta-counterterm $C^{(1)}_{n,2}$ in \cref{lem_Cn2}. The summands match one to one, the only difference is the prefactors, we obtain
	\begin{align*}
		c_{n,2}^{(1)} &= C_{n,2}^{(1)}- \frac 12 \sum_{j=1}^{n-1}  \left \langle K_j \right \rangle b_{j+1} b_{n-j+1} b_3^2 (s_1+\ldots + s_j) (s_{j+1} + \ldots + s_n) \frac{M^{(1)}_{\text{div}}}{2}\\
		&= \frac 12 \sum_{j=1}^{n-1}  \left \langle K_j \right \rangle \left(b_{j+2}b_{n-j+2} - b_{j+1} b_{n-j+1} b_3^2\right) (s_1+\ldots + s_j) (s_{j+1} + \ldots + s_n) \frac{M^{(1)}_{\text{div}}}{2}.
	\end{align*}
	For every $j$, the prefactor vanishes if $b_{j+2} = b_{j+1} b_3$. This is \cref{rec_bn}.
\end{proof}

\begin{theorem}\label{thm_c1}
	If $b_n = \lambda^{n-2}$ then the one-loop counterterm $c^{(1)}_n$ has the same structure as the 1PI vertex $v_n$ \cref{vn_general},
	\begin{align*}
	c_n^{(1)} &=  v_n \Big|_{s_e \rightarrow \frac{M^{(1)}_{\text{\normalfont div}}}{2} \lambda^2 s_e^2}.
	\end{align*}
\end{theorem}
\begin{proof}
	Add the contributions of \cref{lem_c1_1,lem_c1_0} to obtain
	\begin{align*}
		c_n^{(1)} &= 	 (n-1)! a_{n-2} b_3^2 \frac{M^{(1)}_{\text{\normalfont div}}}{2} \left( s_1^2 + \ldots + s_n^2 \right) + b_3^2\frac {M^{(1)}_{\text{  div}}}{2} \frac 12  \sum_{k=2}^{n-2} \sum_{p\in Q^{(n,k)}}     a_{n-k-1}  a_{k-1} (n-k)! k! s^2_{p}\\
		&= \frac 12 \sum_{k=1}^{n-2} \sum_{p\in Q^{(n,k)}}    a_{n-k-1}a_{k-1}  (n-k)! k! s_p^2 b_3^2\frac {M^{(1)}_{\text{  div}}}{2} .
	\end{align*}
	This equals the vertex \cref{vn_general} up to the stated replacement. 
	Given  $b_n = \lambda^{n-2}$,  \cref{lem_c1_2} asserts the absence of other terms.
\end{proof}

\Cref{thm_c1} asserts that if we set
\begin{align}\label{sren}
	is^R := is - (is)^2 b_3^2\frac {M^{(1)}_{\text{  div}}}{2},
\end{align}
then
\begin{align*}
	i v^R_n &:=  \frac 12 \sum_{k=1}^{n-1} a_{n-k-1} a_{k-1} (n-k)!k! \sum_{p\in Q^{(n,k)}} is^R_p= iv_n +  c_n^{(1)}
\end{align*}
is a \enquote{renormalized} vertex in the sense that using $i v^R_n$ in place of $i v_n$, all one-loop divergences are removed from the theory. The \enquote{renormalization} \cref{sren} is a divergent non-linear rescaling of a quantity, much like the rescaling $g^R = g + \mathcal O (g^2)$ in conventional  renormalization, only that it is not a rescaling of a coupling parameter, but, in a certain sense, a non-linear rescaling of spacetime.

\subsection{1PI counterterms with l=2 loops}

 The meta-counterterm of the connected 2-point-function is $C^{(2)}_{2,2} = -\frac 16 b_4^2 s^2 M^{(2)}_{\text{div}}(s)$ from \cref{C22l}. On the other hand, it can be built from 1PI counterterms as shown in \cref{c2_2loop_picture}. After working out  the cancellations due to \cref{c31_1loop}, what remains is
 \begin{align}\label{c2_2loop} 
 	c_ {2,2}^{(2)}(s) &= C_{2,2}^{(2)}  - c_{2,2}^{(1)} \frac i s   c_{2,2}^{(1)} =    b_4^2 s^2 \frac{M^{(2)}_{\text{div}}(s)}{6} - ib_3^4 s^3  \left( \frac{M^{(1)}_{\text{div}} (s) }2 \right) ^2.
 \end{align}

\begin{figure*}[ht]
	\begin{tikzpicture}
	
	\node [treeCounterVertex, label={[label distance=-2mm]5:{$\scriptstyle (2)$}}  ] (c) at (0,0) {};
	\draw  [>>-] (c) --++ (0:.5);
	\draw   [>-](c) --++ (180:.5);

	\node at (1,0) {$=$};
	
	\node [counterVertex, label={[label distance=-2mm]5:{$\scriptstyle (2)$}}  ] (c) at (2,0) {};
	\draw  [>>-] (c) --++ (0:.5);
	\draw   [>-](c) --++ (180:.5);
	
	\node at (3,0) {$+$  };
	
	\node [counterVertex, label={[label distance=-2mm]5:{$\scriptstyle (1)$}}  ] (c1) at (4,0) {};
	\node [counterVertex, label={[label distance=-2mm]5:{$\scriptstyle (1)$}}  ] (c2) at  ($(c1)+ (0:.7)$) {};
	\draw  [>-<](c1) -- (c2);
	\draw  [>-] (c2) --++ (0:.5);
	\draw  [>-] (c1) --++ (180:.5);
	
	\node at (6,0) {$+ \ \left \langle 2 \right \rangle $};
	
	\node [counterVertex, label={[label distance=-2mm]5:{$\scriptstyle (1)$}}  ] (c1) at (8,0) {};
	\node [treeVertex  ] (c2) at  ($(c1)- (0:1)$) {};
	\draw [-|-,bend angle=35,bend left] (c1) to (c2);
	\draw [-|-,bend angle=35,bend right] (c1) to (c2);
	
	\draw  [>-](c2) -- + (180:.5);
	\draw  [>>-](c1) -- + (00:.5);

	\node at (9.5,0) {$+ \ \left \langle 2 \right \rangle $};
	
	\node [treeVertex ] (c1) at (11.5,0) {};
	\node [counterVertex, label={[label distance=-2mm]5:{$\scriptstyle (1)$}}  ] (c0) at ($(c1)+ (0:.7)$){};
	\node [treeVertex  ] (c2) at  ($(c1)- (0:1)$) {};
	\draw [-|-,bend angle=35,bend left] (c1) to (c2);
	\draw [-|-,bend angle=35,bend right] (c1) to (c2);
	
	\draw  [>-<](c0) -- (c1);
	\draw  [>-] (c2) -- + (180:.5);
	\draw  [>-](c0) -- + (00:.5);
	
	\draw [decorate,decoration={brace,amplitude=10pt} ] 	(13,-.8) -- (6,-.8);
	
	\node at (8.5,-1.8) {$=\  \left \langle 2 \right \rangle $};
	
	\node [treeCounterVertex, label={[label distance=-1mm]90:{$\scriptstyle (1)$}}  ] (c1) at (10.7,-1.8) {};
	\node [treeVertex  ] (c2) at  ($(c1)- (0:1)$) {};
	\draw [-|-,bend angle=35,bend left] (c1) to (c2);
	\draw [-|-,bend angle=35,bend right] (c1) to (c2);
	
	\draw  [>-](c2) -- + (180:.5);
	\draw  [>>-](c1) -- + (00:.5);
	
	\node  (tx) at (12.5,-2){$C_{3,1}^{(1)}=0$};
	\draw [thin, -Stealth, shorten >=2mm , bend left] (tx.190) to (c1.270);
	
	\end{tikzpicture}
	\caption{Contributions to the two-loop divergence of the connected 2-point-function. Counterterms in the internal edges are not considered since they give rise to tadpole graphs. The last two graphs together vanish thanks to \cref{c31_1loop_ansatz}}
	\label{c2_2loop_picture}
\end{figure*}

\begin{example}[Massless theory]
	Using \cref{ex_multiedges}, in the massless theory the counterterm reads
	\begin{align*} 
	c_{2,2}^{(2)} &=    b_4^2 s^2  \frac 16\frac{ is }{4 (4\pi)^4}  \frac 1 \epsilon- ib_3^4 s^3   \left(  -\frac{1}{2(4\pi)^2}  \frac 1 \epsilon \right) ^2=  \frac{ is^3 }{24 (4\pi)^4} \left(b_4^2  \frac 1 \epsilon- 6 b_3^4  \frac 1  {\epsilon^2} \right) .
	\end{align*}
\end{example}

For the 3-point-function at two loops, if all legs are onshell, every graph vanishes and 
\begin{align}
	c_{3,0}^{(2)} &= 0 \label{c30_2loop}.
\end{align}
If one leg is offshell, we obtain an analogue of \cref{c31_1loop_ansatz} which is shown in \cref{c31_2loop_picture}. Two of the four graphs add to zero thanks to \cref{c31_1loop_ansatz}.
\begin{align} 
	C_{3,1}^{(2)}=0 &= c_{3,1}^{(2)} 
	+ \left \langle 3 \right \rangle c_{3,1}^{(1)} \frac{i}{s_1} c_{2,2}^{(1)}  (s_1) 
	+\left \langle 3 \right \rangle  (-ib_3s_1)\frac{i}{s_1} c_{2,2}^{(1)} (s_1)\frac{i}{s_1} c_{2,2}^{(1)}(s_1) 
	+\left \langle 3 \right \rangle  (-ib_3s_1)\frac{i}{s_1} c_{2,2}^{(2)} (s_1) \nonumber  \\
	\Rightarrow &\quad c_{3,1}^{(2)}  = -\left \langle 3 \right \rangle    b_3  c_{2,2}^{(2)} (s_1)  . \label{c31_2loop}
\end{align}
We find that $c_{3,1}^{(2)}$ behaves in just the same way with respect to $c_{2,2}^{(2)}$ as $c_{3,1}^{(1)}$ does with respect to $c_{2,2}^{(1)}$ in \cref{c31_1loop}.

\begin{figure*}[htbp]
	\begin{tikzpicture}
\node [treeCounterVertex, label={[label distance=-2mm]5:{$\scriptstyle (2)$}}  ] (c) at (-.5,0) {};
\draw  [>>>-](c) --++ (0:.6);
\draw  [-|](c) --++ (-120:.4);
\draw  [-|](c) --++ (120:.4);

\node at (1,0) {$=0=$};

\node [counterVertex, label={[label distance=-2mm]5:{$\scriptstyle (2)$}}  ] (c) at (2,0) {};
\draw  [>>>-](c) --++ (0:.6);
\draw  [-|](c) --++ (-120:.4);
\draw  [-|](c) --++ (120:.4);

\node at (3.5,0) {$+   $};

\node [counterVertex, label={[label distance=-2 mm]80:{$\scriptstyle (1)$}}] (c1) at (4.5,0){};
\node [counterVertex, label={[label distance=-2 mm]80:{$\scriptstyle (1)$}}] (c2) at ($(c1)+ (0:.8)$){};
\draw  [>>-<] (c1) -- (c2);

\draw [-|](c1) -- + (120:.5);
\draw [-|](c1) -- + (-120:.5);
\draw [>-](c2) -- + (0:.5);

\node at (6.5,0) {$+  $};

\node [treeVertex ] (c1) at (7.5,0){};
\node [counterVertex, label={[label distance=-2 mm]90:{$\scriptstyle (1)$}}] (c2) at ($(c1)+ (0:.7)$){};
\node [counterVertex, label={[label distance=-2 mm]90:{$\scriptstyle (1)$}}] (c3) at ($(c2)+ (0:.8)$){};
\draw  [>-<] (c1) -- (c2);
\draw  [>-<] (c2) -- (c3);

\draw [-|](c1) -- + (120:.5);
\draw [-|](c1) -- + (-120:.5);
\draw [>-](c3) -- + (0:.5);

\node at (10,0) {$+ $};

\node [treeVertex ] (c1) at (11,0){};
\node [counterVertex, label={[label distance=-2 mm]60:{$\scriptstyle (2)$}}] (c2) at ($(c1)+ (0:.7)$){};
\draw  [>-<] (c1) -- (c2);

\draw [-|](c1) -- + (120:.5);
\draw [-|](c1) -- + (-120:.5);
\draw [>>-](c2) -- + (0:.5);

\draw [decorate,decoration={brace,amplitude=10pt} ] 	(9,-.8) -- (4,-.8);

\node at (6,-1.8){$=$};
\node [treeCounterVertex , label={[label distance=-2 mm]60:{$\scriptstyle (1)$}}] (c1) at (7,-1.8){};
\node [counterVertex, label={[label distance=-2 mm]60:{$\scriptstyle (1)$}}] (c2) at ($(c1)+ (0:.8)$){};
\draw  [>>-<] (c1) -- (c2);

\draw [-|](c1) -- + (120:.5);
\draw [-|](c1) -- + (-120:.5);
\draw [>-](c2) -- + (0:.5);

\node  (tx) at (10,-2){$C_{3,1}^{(1)}=0$};
\draw [thin, -Stealth, shorten >=2mm , bend left] (tx.190) to (c1.270);

	\end{tikzpicture}
\caption{Contributions to the two-loop meta-counterterterm $C_{3,1}^{(2)}$ in terms of 1PI counterterms. }
\label{c31_2loop_picture}
\end{figure*}

With two legs offshell, the vanishing of $C^{(1)}_{3,1}$ and $C^{(1)}_{4,1}$ eliminates all one-loop graphs with inserted counterterms similar to the mechanism in \cref{c2_2loop_picture}. What remains are the tree graphs shown in \cref{c32_2loop_picture}. Inserting $C_{3,2}^{(l)}$ from \cref{C32l}, one finds
\begin{align}\label{c32_2loop}
	c_{3,2}^{(2)} &= C_{3,2}^{(2)} -\left \langle 6 \right \rangle  C_{3,2}^{(1)} \big|_{s_3=0} \frac i {s_1} c_{2,2}^{(1)}(s_1) - \left \langle 6 \right \rangle   (-ib_3 s_2) \frac i {s_1} c_{2,2}^{(2)}(s_1) \nonumber \\
	&=  \left \langle 3 \right \rangle \left( b_5 b_4-b_3 b_4^2 \right)  s_1 (s_2+s_3) \frac{M^{(2)}_{\text{div}}(s_1)}{6} - \left \langle 3 \right \rangle   b_3^3 b_4 2 i s_1^2 (s_2+s_3)     \left( \frac{M^{(1)}_{\text{div}}}{2}\right)^2 .
\end{align}
Compared to the one-loop case \cref{c32_1loop}, there is a structurally different contribution which stems from the braced graphs in \cref{c32_2loop_picture}.

\begin{figure*}[htbp]
	\begin{tikzpicture}
		\node [treeCounterVertex, label={[label distance=-2mm]5:{$\scriptstyle (2)$}}  ] (c) at (-.5,0) {};
		\draw  [>>-](c) --++ (0:.6);
		\draw  [>-](c) --++ (-120:.5);
		\draw  [-|](c) --++ (120:.4);
		
		\node at (1,0) {$=C_{3,2}^{(2)}=$};
		
		\node [counterVertex, label={[label distance=-2mm]5:{$\scriptstyle (2)$}}  ] (c) at (2.5,0) {};
		\draw  [>>-](c) --++ (0:.6);
		\draw  [>-](c) --++ (-120:.5);
		\draw  [-|](c) --++ (120:.4);
		
		\node at (3.5,0) {$+   $};
		
		\node [counterVertex, label={[label distance=-2 mm]80:{$\scriptstyle (1)$}}] (c1) at (4.5,0){};
		\node [counterVertex, label={[label distance=-2 mm]80:{$\scriptstyle (1)$}}] (c2) at ($(c1)+ (0:.8)$){};
		\draw  [>-<] (c1) -- (c2);
		
		\draw [-|](c1) -- + (120:.4);
		\draw [>-](c1) -- + (-120:.5);
		\draw [>-](c2) -- + (0:.5);
		
		\node at (6.5,0) {$+  $};
		
		\node [treeVertex ] (c1) at (7.5,0){};
		\node [counterVertex, label={[label distance=-2 mm]90:{$\scriptstyle (1)$}}] (c2) at ($(c1)+ (0:.7)$){};
		\node [counterVertex, label={[label distance=-2 mm]90:{$\scriptstyle (1)$}}] (c3) at ($(c2)+ (0:.8)$){};
		\draw  [-<] (c1) -- (c2);
		\draw  [>-<] (c2) -- (c3);
		
		\draw [-|](c1) -- + (120:.4);
		\draw [>-](c1) -- + (-120:.5);
		\draw [>-](c3) -- + (0:.5);
		
		\node at (10,0) {$+ $};
		
		\node [treeVertex ] (c1) at (11,0){};
		\node [counterVertex, label={[label distance=-2 mm]60:{$\scriptstyle (1)$}}] (c2) at ($(c1)+ (0:.7)$){};
		\node [counterVertex, label={[label distance=-1 mm]0:{$\scriptstyle (1)$}}] (c3) at ($(c1)+ (-120:.7)$){};
		\draw  [>-<] (c1) -- (c2);
		\draw  [-<] (c1) -- (c3);
		
		\draw [-|](c1) -- + (120:.5);
		\draw [>-](c3) -- + (-120:.5);
		\draw [>-](c2) -- + (0:.5);

		\node at (13,0) {$+ $};
		
		\node [treeVertex ] (c1) at (14,0){};
		\node [counterVertex, label={[label distance=-2 mm]60:{$\scriptstyle (2)$}}] (c2) at ($(c1)+ (0:.7)$){};
		\draw  [-<] (c1) -- (c2);
		
		\draw [-|](c1) -- + (120:.5);
		\draw [>-](c1) -- + (-120:.5);
		\draw [>>-](c2) -- + (0:.5);

		\draw [decorate,decoration={brace,amplitude=10pt} ] 	(12,-1) -- (4,-1);
		
		\node at (7,-2){$=$};
		\node [treeCounterVertex , label={[label distance=-2 mm]60:{$\scriptstyle (1)$}}] (c1) at (8,-2){};
		\node [counterVertex, label={[label distance=-2 mm]60:{$\scriptstyle (1)$}}] (c2) at ($(c1)+ (0:.8)$){};
		\draw  [>-<] (c1) -- (c2);
		
		\draw [-|](c1) -- + (120:.5);
		\draw [>-](c1) -- + (-120:.5);
		\draw [>-](c2) -- + (0:.5);

	\end{tikzpicture}
	\caption{Contributions to the two-loop meta-counterterterm $C_{3,2}^{(2)}$ in terms of 1PI counterterms. All contributions of one-loop-graphs cancel since $C_{n,1}^{(1)}=0$ \cref{Cn1}. The brace indicates \cref{c32_1loop_picture}.}
	\label{c32_2loop_picture}
\end{figure*}

Finally, allowing all three legs to be offshell, $C_{3,3}^{(2)}$ is the meta-counterterm computed in \cref{C33_2loop}. This allows to reconstruct the 1PI-counterterm where we use \cref{c2_1loop,c32_1loop}
\begin{align}\label{c33_2loop}
	c_{3,3}^{(2)} &= C_{3,3}^{(2)} - \left \langle 3 \right \rangle c_{2,2}^{(1)}(s_1) \frac{i}{s_1} (-i b_3 s_2) \frac i {s_3} c_{2,2}^{(1)} (s_3) - \left \langle 3 \right \rangle c_{2,2}^{(1)}(s_1) \frac i {s_1} c_{3,2}^{(1)}\big|_{s_1=0} \nonumber \\
	&= C_{3,3}^{(2)} -  3  i b_3^3 \left(2b_4 - b_3^2 \right)   s_1  s_2 s_3    \left( \frac{M^{(1)}_{\text{div}}}{2}\right)^2  .
\end{align}

\begin{example}[Massless theory]\label{ex_c33_2loop_massless}
	For the massless theory, $C_{3,3}^{(2)}$ was computed in \cref{C33_2loop_massless} and one has, inserting \cref{ex_multiedges}
	\begin{align*}
		c_{3,3}^{(2)}			&= i s_1 s_2 s_3  \frac{3}{4 (4\pi)^4} \left( \left( b_3 b_4^2+b_3^2 b_5  +b_3^5 -2b_3^3b_4  \right) \frac{1}{\epsilon^2} -  b_3 b_4^2   \frac 1 \epsilon\right).
	\end{align*}
\end{example}

\begin{example}[Exponential diffeomorphism]\label{ex_cn_2loop_rec}
	Assuming a massless theory and \cref{rec_bn}, $b_n = \lambda^{n-2}$,  the two-loop counterterms \cref{c2_2loop,c31_2loop,c32_2loop,c33_2loop} simplify to 
	\begin{align*}
		c_{2,2}^{(2)} &=    \frac{ is^3 }{24 (4\pi)^4} \lambda^4  \left(  \frac 1 \epsilon-  6\frac 1  {\epsilon^2} \right) \\
		c_{3,1}^{(2)}  &=      \frac{ i\left( s_1^3+s_2^2 + s_3^2\right)  }{24 (4\pi)^4} \lambda^5  \left(    6\frac 1  {\epsilon^2}-\frac 1 \epsilon \right) \\
		c_{3,2}^{(2)} &= -   \frac{2 i  \left(s_1^2 s_2 + s_1^2 s_3 + s_2^2 s_1 +s_2^2 s_3 + s_3^2 s_1 + s_3^2 s_2 \right)     }{4 (4\pi)^4   }  \lambda^5 \frac{1}{\epsilon^2}     \\
		c_{3,3}^{(2)}
		&=   \frac{3is_1s_2s_3 }{4 (4\pi)^4}  \lambda^5\left( \frac{1}{\epsilon^2} -   \frac 1 \epsilon\right).
	\end{align*}
	This shows that, contrary to the one-loop-counterterms in \cref{thm_c1}, the two-loop-counterterms can not be generated by a simple rescaling like \cref{sren} of the momenta in the bare vertices. This was to be expected because already in the connected perspective, the dunce's cap graph $\Gamma_A$ in \cref{fig_3point_2loop} clearly does not represent a propagator correction. Without this graph $c_{3,3}^{(2)}=0$ but still $c_{3,2}^{(2)}\neq 0$.
\end{example}

\subsection{1PI two-point-function}\label{sec_1PI_2point}

As opposed to $G_2(s)$ from \cref{G2}, the 1PI Function $G_2^{\text{1PI}}(s)$ contains divergences logarithmic in the momentum. Removing them requires the systematic use of 1PI counterterms. We have seen in \cref{c2_2loop_picture} that at 2 loops, $c_2^{(2)}$ can be computed from $C_2$ alone and does not involve meta-counterterms of higher valence. This is true to all loop orders. 

\begin{lemma}\label{lem_1PI_2point}
	Let 
	\begin{align*}
	\qquad C(t) := \sum_{l=1}^\infty C_2^{(l)} t^l = \sum_{l=1}^\infty  b_{l+2}^2 s^2  \frac{M^{(l)}_{\text{\normalfont div}}(s)}{(l+1)!}t^l
	\end{align*}
be the ordinary generating functions of  meta-counterterms of the 2-point-function, then
\begin{align*}
	c^{(l)} &=  -is[t^l]   \frac{C(t)}{C(t)-is}
\end{align*}
	
\end{lemma}
\begin{proof}
	The stated form of $C_2^{(l)}$ is \cref{C22l}. 	Use induction to prove the lemma. At one loop, $c^{(1)}_2 = C^{(1)}_2$ by \cref{c2_1loop}.  Consider loop order $l$ and assume all $c^{(j)}$ for $j<l$ are local. 
	
	Then the connected $l$-loop 2-point-function  is given by the sum over all chains of $k$ 1PI-graphs $\Gamma^{(j)}$ where $1 \leq k \leq l$ and $\Gamma^{(j)}$ is supposed to be the sum of all $j$-loop 1PI graphs. Following the usual BPHZ renormalization procedure, for each $\Gamma^{(j)}$ with $j<l$ the graph can be replaced by its corresponding 1PI counterterm $c^{(j)}$. Thereby, we remove all divergences except one: In the case that no $c^{(j)}$ are inserted at all but only the graphs $\Gamma^{(j)}$ themselves, we know by \cref{thm_feynmanrules} that these graphs sum up to $M^{(l)}$. Therefore, the summand which arises if all the $c^{(j)}$ are inserted  has no graph to cancel. It remains as an additional divergent term which needs to be absorbed by the only undetermined quantity, the counterterm $c^{(l)}$. 
	
	If the chain consists of $k$ factors, the so-produced divergence amounts to all possible  partitions of $l$ loops into $k$ counterterms, each connected by a propagator $\frac i s$, in total
	\begin{align*}
		\frac{k!}{n!} \left( \frac i s \right) ^{k-1} B_{n,k} \left( 1!c^{(1)}, 2! c^{(2)}, 3! c^{(3)}, \ldots  \right) . 
	\end{align*} 
	The factorials are present because we do not distinguish between the factors, e.g. $c^{(1)}\frac is c^{(1)} \frac is c^{(1)}$ appears only once, not $3!$ times. Finally, we have to sum over all $k$ where $k=1$ represents the new, undetermined counterterm. The result is minus the overall divergence of the connected $l$-loop function, but we know this to be $C^{(l)}_2$ where no subdivergences occur. So
	\begin{align*}
		C^{(l)}_2 &=  \sum_{k=1}^n	\frac{k!}{n!} \left( \frac i s \right) ^{k-1} B_{n,k} \left( 1!c^{(1)}, 2! c^{(2)}, 3! c^{(3)}, \ldots  \right) . 
	\end{align*} 
	This is   Faa di Brunos formula for ordinary generating functions. Define 
	\begin{align*}
		F(t) &:= \sum_{k=1}^\infty \left( \frac i s \right) ^{k-1} t^k = -is\frac{t}{-is-t} ,		\qquad c(t) := \sum_{l=1}^\infty c_2^{(l)} t^l,
	\end{align*}
then $C^{(l)}$ is the $l$-th coefficient of a power series $	C(t) =  F(c(t))$. Inversion gives the lemma.
\end{proof}

\begin{example}[Massless theory]
	In the massless theory, by \cref{lem_Ml}, 
	\begin{align*}
		\qquad C(t)  = +\frac 1 \epsilon \sum_{l=1}^\infty  b_{l+2}^2   \frac{(-is)^{l+1}}{(4\pi)^{2l} (l!)^2(l+1)!}t^l.
	\end{align*}
	As a peculiar example, let $b_n=0$ for all $n>3$, then $C(t) = -\frac{s^2 b_3^2}{2(4\pi)^2 \epsilon}t$ and the $l$-loop counterterm is
	\begin{align*}
		c^{(l)} &= -is \left( \frac{-is b_3^2}{2(4\pi)^2\epsilon} \right) ^l.
	\end{align*}
If one uses analytic continuation of the sum, then the all-order counterterm is regular in the limit $\epsilon \rightarrow 0$ and cancels the 1PI tree-level 2-valent vertex $-is$:
\begin{align*}
	c_{2} &= \sum_{l=1}^\infty c_2^{(l)} = is \frac{is b_3^2}{2(4\pi)^2\epsilon + isb_3^2}  \rightarrow is.
\end{align*} 
Apart from amusement, taking the limit $\epsilon \rightarrow 0$ at this point is illicit since one then fails to reproduce \cref{G2div_massless}.
This curious example indicates that regularity or divergence of low-order perturbative quantities does not necessarily imply the same for their all-order  counterparts. A similar behaviour is suspected e.g. for the wavefuction renormalization factor $Z$, which diverges in perturbation theory but is assumed to fulfil $0<\abs Z <1$ on physical grounds, see \cite[Sec. 8]{lutz} and references therein.
\end{example}

\section{Ward-Slavnov-Taylor-Identities}\label{sec:st}

In \cref{sec_1PI}, we computed several 1PI counterterms and found that many of them are related. See e.g. \cref{c31_1loop,c31_2loop} or the fact that the all order counterterm $c_2$ of the 2-point-function does not require knowledge of vertex-counterterms of higher valence. All these effects can be summarized by the following Slavnov-Taylor-like identities. 

\begin{theorem}\label{thm_st}
	Let $c_n^{(l)} = \sum_{k=0}^n c_{n,k}^{(l)}$ be the $l$-loop $n$-valent 1PI counterterms and let
	\begin{align*}
		\Gamma_2 &:= -i s - \sum_{l=1}^\infty c_2^{(l)}(s) ,\qquad 	\Gamma_{n\geq 3} := iv_n + \sum_{l=1}^\infty c_n^{(l)}, 
	\end{align*}
	and assume that tadpoles vanish, then for $n\geq 2$
	\begin{align*}
		(1)& \qquad \left[ \Gamma_n \frac{1}{\Gamma_2}(-is) \right] _{\textnormal{only $s$ offshell}} = iv_n \Big| _{\textnormal{only $s$ offshell}} ,\\  
		(2) &\qquad 	\sum_{\Gamma_j \star \Gamma_k = \Gamma_n}\left[ \Gamma_j \frac{1}{\Gamma_2} \Gamma_k \right] _{\textnormal{onshell}} = -\Gamma_n \Big|_{\textnormal{onshell}} ,
	\end{align*}
where the product $\star$ implies $j+k=n+2$ and a sum over all orientations of the graphs.
\end{theorem}
\begin{proof}
	First note that
	\begin{align*}
		\frac{1}{\Gamma_2} = \frac{i}{s-i c_2} = \frac i s \sum_{r=0}^\infty \left( \frac i s c_2 \right) ^r
	\end{align*}
	is the non-amputated chain of all 1PI 2-point counterterms. Consequently, $\frac{1}{\Gamma_2}(-is)$ is the same chain where the outermost propagator is removed.
	
	For any $n\geq 3$, the connected $n$-point correlation function vanishes if not more than one external edge is offshell due to \cref{thm_feynmanrules}. Consequently, its divergent part vanishes and $C_{n,0}=C_{n,1}=0$, see \cref{Cn0,Cn1}. It suffices to consider connected graphs where all internal edges are cancelled since the remaining graphs are products of the former type. 
	
	First prove (1). 
	Use induction on $n$. For $n=2$, (1) becomes $(-is) = i v_2$ which is true. For $n=3$, since $c_2$ vanishes onshell, the connected graph where only $s$ is offshell is $\Gamma_3 \frac{1}{\Gamma_2}(-is)$ where $\Gamma_2$ is the counterterm of edge $s$. But $C_{3,1}=0$ and hence only the regular term survives of this sum, which is $iv_3$ as claimed in (1). 
	Now assume (1) holds for $j<n$. Then, in the sum of all connected graphs, all divergent contributions cancel where $s$ is adjacent to a $j$-valent counterterm, either directly or via a string of propagator counterterms. The only non-vanishing terms are those where a $n$-valent counterterm is involved. But again, the sum over all divergent terms has to vanish and the only remaining term is $iv_n$.  This proves (1). 
	
	For (2), the case $n=2$ reads $\Gamma_2|_{\text{onshell}} = -\Gamma_2|_{\text{onshell}}$ which is true since $\Gamma_2|_{\text{onshell}}=0$ by \cref{lem_1PI_2point}. The same holds for $n=3$ since, by \cref{Cn0}, $\Gamma_3|_{\text{onshell}}=0$.
	
	Assume (2) holds for $j,k<n$. The onshell connected amplitude can only be proportional to powers of internal momenta $s_e$. If there is only one such internal momentum, corresponding to one internal edge $e$, then all terms proportional to $s_e$ arise from $\Gamma_j \frac{1}{\Gamma_2(e)}\Gamma_k$. Since these terms are not present in the end result, we know $\Gamma_n$ must absorb them. If there is more than one edge, pick one and call it $e$. Then, there are two subtrees $T_j, T_k$, each of which has valence $<n$ and only one external edge offshell, namely $e$. But by \cref{Cn1}, such trees do not contain divergent terms.  In fact, as a consequence of (1), such trees do not even contain powers of internal momenta since they are made from tree-level vertices $iv_k$ and such trees evaluate to $b_j$ by \cref{thm_bn}. Therefore, the only relevant contribution stems from trees with exactly one internal multi-cancelled edge, which proves (2).	
\end{proof}
The compatibility of \cref{thm_st} with locality is expressed by the fact that such identifications between different $n$-point-functions represent Hopf ideals in the core Hopf algebra \cite{kreimer_recursive_2009,prinz_gauge_2019}. Note also that \cref{ordinary_recursive} is the tree-level version of statement (2). Technically, only statement (1) of \cref{thm_st} requires the vanishing of tadpoles, i.e. $C_{n,1}=0$, whereas (2) holds regardless. 

Statement (1) can be rewritten in the form
\begin{align*}
	\Gamma_n \Big|_{\text{only $s$ offshell}} &= \left[ iv_n \frac i s \Gamma_2(s) \right] _{\text{only $s$ offshell}},
\end{align*}
which implies $c_n^{(l)} \Big|_{\text{only $s$ offshell}}  = \left[ v_n s^{-1} c_2(s) \right] _{\text{only $s$ offshell}} $. This means, \cref{lem_c1_1} holds to all orders in perturbation theory. Inserting (1) into (2) produces
\begin{align*}
	-\Gamma_n  \Big|_{\text{onshell}}  &= \sum_{\Gamma_j \star \Gamma_k=\Gamma_n} \left[ i v_n \frac i s \Gamma_k   \right] _{\text{onshell}}  = \sum_{\Gamma_j \star \Gamma_k = \Gamma_n} \left[ i v_n \frac i s \Gamma_2(s) \frac i s iv_k \right] _{\text{onshell}} 
\end{align*}
and thereby also \cref{lem_c1_0} holds to all orders. It is \cref{lem_c1_2} which fails at higher than one-loop order: To all orders, the parts of the counterterms, which have the same momentum dependence as the vertices $iv_n$, can be obtained by replacing $-is \rightarrow \Gamma_2(s)$. But starting from two-loop order, there are additional kinematic form factors in the counterterms which are not obtained in this way. These terms can not be constructed inductively from the 2-point-counterterms.

\section{Analogy to gauge theory amplitudes}\label{sec:mhv}
In this last section we argue that there are surprising  similarities between a scalar field diffeomorphism and gauge theories. 

The most striking one is the presence of the Slavnov-Taylor-like identities \cref{thm_st}. 
In  quantum gauge theories, the various divergent amplitudes and especially their local counterterms are related to each other by the Ward identity \cite{ward_identity_1950} in QED respectively  Slavnonv-Taylor identities \cite{taylor_ward_1971,slavnov_ward_1972,thooft_renormalization_1971} in QCD. These identities guarantee gauge invariance for the quantized theory and imply that all divergent correlation functions can be rendered finite by a redefinition of \emph{the same} coupling constant, see for example \cite{gracey_selfconsistency_2019}. Relations similar to \cref{thm_st} were also proposed for quantum Einstein gravity \cite{gravity,kreimer_not_2009}.  

In our case, the invariance of the $S$-matrix under global field diffeomorphisms takes the role of local gauge invariance. Conceptually, the identities \cref{thm_st} play the same  role as in gauge theories: They guarantee that  the counterterms of loop-level amplitudes are compatible with the invariance, which in our case means that they vanish in the onshell limit $s_e\rightarrow 0$.  As a by-product, they significantly reduce the number of independent counterterms. For the $n$-valent counterterm, only the summand $c_{n,n}$ is truly independent, all $c_{n,k}$ for $k<n$ are determined by lower valence counterterms $c_{k}$. For example, in \cref{c2_2loop,c31_2loop,c32_2loop,c33_2loop}, the only independent contribution to the 3-valent counterterm is $c_{3,3}$, whereas $c_{3,0}, c_{3,1}$ and $c_{3,2}$ are determined by $c_2$. In the connected perspective, the corresponding statement is \cref{lem_subdivergences}.

The analogy between a scalar field diffeomorphism and QCD becomes visible in the case  $s_p = p^2$ and $a_1 = \frac{-g}2, a_{n>1}=0$ of the diffeomorphism \cref{def_diffeomorphism}.  Then, using $\partial_\mu \rho^2 = 2 \rho \partial_\mu \rho$, the Lagrangian of $\rho$ is
\begin{align}\label{Lrho_pt}
	\mathcal L_\rho  &= \frac 12 \left(- \partial_\mu \rho + g \;\rho \partial_\mu \rho   \right) \left( -\partial^\mu \rho + g \; \rho \partial^\mu \rho  \right)  
\end{align}
which is reminiscent of the Yang-Mills-Lagrangian of QCD,
\begin{align}\label{Lqcd}
	\mathcal L_{\text{QCD}} &= -\frac 14 \left(  \partial_\mu A^a_\nu -\partial_\nu A^a_\mu + g\; f^{abc} A^b_\mu A^c_\nu \right) \left(  \partial^\mu A^{a\nu} -\partial^\nu A^{a \mu} + g \; f^{abc} A^{b \mu} A^{c \nu} \right) .
\end{align}
The scalar field is not a Lorentz vector, so there  necessarily are differences in the tensor structure between \cref{Lrho_pt} and \cref{Lqcd}, and the latter also needs gauge-fixing. But still \cref{Lrho_pt} is perhaps the closest one can get to the Yang-Mills-Lagrangian by only using scalar fields.
 
In QCD, the maximum helicity violating (MHV) amplitudes are those where precisely two out of $n$ external onshell gluons have a different helicity than the rest. To leading order in $N$ of the gauge group $SU(N)$, their matrix element is given by the Parke-Taylor-Formula \cite{parke_amplitude_1986}
\begin{align}\label{MHV}
	\abs{M_{\text{MHV}}(1^-, 2^-, 3^+,\ldots )}^2 &= \frac{g^{2n-4}}{2^{2n-4}  }  \frac{N^{n-2} (N^2-1) }n   (p_1 \cdot p_2)^4 \sum_P \frac{1}{(p_1 \cdot p_2) (p_2 \cdot p_3) \cdots (p_n\cdot p_1)}  
\end{align}
where $P$ ranges over all permutations of $1, \ldots, n$. 
The validity of \cref{MHV} is a consequence of Berends-Giele relations \cite{berends_recursive_1988} in QCD. The BG relations are a recursive construction of $n$-gluon currents $\hat J^x_\xi(1,2,\ldots, n)$, which are the connected $(n+1)$-point gluon functions where precisely one leg is offshell. 
This is almost verbatim the definition of $b_{n+1}$ in \cref{def_bn}. Consequently, the proof of \cref{thm_bn} in \cite{KY17} is strikingly similar to the derivation of BG relations in \cite{berends_recursive_1988}. In this sense, the connected  amplitudes $iV_n \sim i b_n$ in \cref{Vn} are a scalar analogue of the Parke-Taylor-amplitudes in QCD. 
In the scalar case \cref{Lrho_pt}, the tree-level matrix element with one external edge offshell is the square of \cref{Vn},
\begin{align}\label{Vn_pt}
	\abs{V_n}^2 =\abs{b_n}^2 \sum_{i=1}^n s_i^2=  g^{2n-4} \left((2n-1)!!\right)^2  \sum_{i=1}^n \left( p_i \cdot p_i \right) ^2.
\end{align}
In terms of kinematics, the tree sums $b_n$ are simpler than the gluon currents $\hat J^x_\xi$ since there is no remainder at all of the internal edges. Interestingly, the $b_n$ resp. $\abs{V_n}^2$ can be computed easily even if \cref{Lrho_pt} contains more than quadratic terms in each factor.  This might be interesting for more complicated gauge theories such as quantum gravity.

\begin{example}[Exponential diffeomorphism]\label{ex_st_rec}
	In \cref{sec:rec}, we studied in detail the choice $a_n = \frac{(-1)^n g^n}{n+1}$ and, in \cref{rec_u1_L}, found it reminiscent of Einstein gravity .  It coincides at leading order with the choice $a_1 = \frac{-g}{2}$ in \cref{Lrho_pt}  and the Lagrangian can be written in the form
	\begin{align*}
		\mathcal L_\rho &= \frac 12 \left( -\partial_\mu \rho + g \; \rho \partial_\mu \rho - g^2 \; \rho^2 \partial_\mu \rho + g^3 \; \rho^3 \partial_\mu \rho \mp \ldots \right) ^2. 
	\end{align*}
	Since $b_n = g^{n-2}$, it has the leading order matrix element
	\begin{align*}
		\abs{V_n}^2 &= g^{2n-4} \sum_{i=1}^n \left( p_i \cdot p_i \right) ^2.
	\end{align*}
	Compared to \cref{Vn_pt}, this particular choice of diffeomorphism eliminates all combinatoric prefactors.
\end{example}

We do not claim that the QCD Slavnov-Taylor Relations or the Parke-Taylor-Formula can be derived from the scalar case or vice versa, but it is interesting to observe that similar structures exist in both cases. This implies that their presence is not exclusive to   a massless spin-1  field and its Lorentz representations, but follows rather generically from the algebraic structure of the Lagrangian being a square of an expression which is only quadratic in the fields. In this spirit it appears not surprising that the graphs of the standard model gauge theories can be generated algebraically from the graphs of a cubic scalar theory \cite{kreimer_quantization_2013,kreimer_properties_2013,prinz_corolla_2016}.

\section{Conclusion}

We have completed a survey of propagator cancelling scalar quantum field theories. Throughout, we have restricted ourselves to  those theories which are obtained by a global diffeomorphism of the field variable of a free scalar quantum field. This subclass  exhausts all massless scalar fields with quadratic propagator (\cref{thm:generality}). 
The key results are:
\begin{enumerate}
	\item We have demonstrated in \cref{sec:pos_interpretation} that the Feynman rules of a field diffeomorphism in momentum space, given by the connected perspective \cref{thm_feynmanrules}, are in accordance with the expected behaviour of diffeomorphisms in position space.
	\item We have derived the all-orders 2-point function $G_2$ for an arbitrary diffeomorphism of a free field with propagator $i/p^2$ in \cref{sec:twopoint}.
	\item We have extended the connected perspective to include also counterterm-metavertices and have computed several of them in \cref{sec_divergences}.
	\item In \cref{sec_1PI}, we have reconstructed 1PI-counterterms at one- and two-loop level from the meta-counterterms. Further, we showed that the 1PI counterterm of the 2-point-function can be computed without knowledge of higher valence counterterms (\cref{lem_1PI_2point}).
	\item In \cref{sec:rec}, we have examined in detail those theories  where the connected amplitudes are proportional to each other, $b_n \sim \lambda^{n-2}$,  with the following results:
	\begin{enumerate}
		\item The explicit field diffeomorphisms giving rise to all such theories are given by \cref{lem_rec_positionspace}. For the simplest cases the transformed Lagrangian \cref{rec_L_u1} amounts to introducing the inverse of the field variable into the kinetic term and is reminiscent of the Einstein-Hilbert-Lagrangian. 
		\item The all-order 2-point-function $G_2$ is the only amplitude which involves infinitely many graphs. The connected $n$-point loop-level amplitudes are then given by finitely many graphs on up to $n$ vertices, connected by $G_2$, both in momentum space (\cref{thm_rec_fey}) and position space (\cref{sec:pos_higher}).
		\item We computed explicitly the  position-space 2-point function (\cref{rec_lem_G2}) and its momentum-space divergence for massless fields in $4-2\epsilon$ dimensions (\cref{rec_lem_G2div}). For the simplest instance of recursion relations, $u=1$ in \cref{rec_bn}, it turns out to be the exponential superpropagator studied extensively in the early years of quantum field theory. Our result agrees with the literature up to a finite function which amounts to different renormalization conditions. Thereby we verified that the historic result is consistent with systematic momentum-space perturbation theory. 
		\item We showed in \cref{thm_c1} that if $u=1$ then all one-loop counterterms $c_{n}^{(1)}$ coincide with the tree-level vertices $iv_n$ if all offshell variables are scaled simultaneously. This is obvious in the connected perspective, see e.g. \cref{fig_rec_1L} where  each of the graphs is in 1:1 correspondence with the replacement of precisely one propagator by a multiedge. 
		\item In \cref{ex_cn_2loop_rec}, we found that at two loops, this scaling alone is not sufficient to render the theory finite. This is again plausible from the connected perspective as the graph $\Gamma_A$ in  \cref{fig_3point_2loop} is not a propagator correction. 
	\end{enumerate}
	\item In \cref{thm_st} we showed that the counterterms of the diffeomorphed field fulfil a set of equations which are analogous to Slavnov-Taylor-identities in QCD. These relations reflect the fact that the theory is invariant under diffeomorphisms at quantum level. They dramatically reduce the number of independent counterterms  but still infinitely many independent counterterms remain.
	\item We argued that propagator cancelling scalar fields show remarkable similarities to gauge theories and that at least three different sets of recursion relations, which have become indispensable in the modern understanding of perturbative quantum field theory, make an appearance:
	\begin{enumerate}
		\item Slavnov-Taylor-identities between all counterterms are realized by \cref{thm_st}. 
		\item Berends-Giele relations for the gluonic current $\hat J^x_\xi$ have an analogue in the recursive construction of tree sums $b_n$ as given in \cite{KY17}, see \cref{sec:mhv}. 
		\item BCFW relations are respected by the connected perspective \cref{thm_feynmanrules} at least at tree-level, see comment in \cref{sec:connectedperspective}.
	\end{enumerate} 
\end{enumerate} 

The  similarities between scalar field diffeomorphisms and QCD might  help to understand the behaviour of more complicated gauge theories such as quantum Einstein gravity.  There, contrary to \cref{Lqcd}, higher than quadratic terms are present in the two factors in the Lagrangian, and such behaviour can easily be included in the scalar case, see \cref{ex_st_rec}. Further, in the scalar case there is an alternative proof for the recursions \cref{thm_bn} using generating functions \cite{mahmoud_diffeomorphisms_2020}. This raises the question if a similar construction is possible in gauge theories as well.  

Given the above similarities, propagator cancelling scalar fields can also serve as a pedagogical model  to illustrate how a certain structure of the Lagrangian translates to certain behaviour of correlation functions, without obfuscating this correspondence by the extensive tensor notation of higher spin fields. 

Our discussion of Slavnov-Taylor identities for the propagator cancelling field in \cref{sec:st} has remained sketchy. To make it more precise, one would need to decompose the counterterms $c_n^{(l)}$ into kinematic form factors more systematically than we did. This is especially necessary for the case of quantum gravity, where an analogue set of identities must hold to ensure gauge invariance at quantum level.

\begin{appendix}

\section{Massless multiedge graphs}\label{sec_multiedge}
	By $M^{(l)} (s)$ we denote the $l$-loop multiedge graph without vertex- and symmetry-factors. 	
	In the case of massless fields, $m=0$, all Feynman amplitudes $M^{(l)}$ are known to evaluate to Gamma functions. To see this, consider the massless one-loop multiedge with propagator powers $\alpha, \beta$ and set $s:=p^2$, it has the amplitude (e.g.  \cite{usyukina_representation_1975,milgram_tables_1985,davydychev_recursive_1992})
	\begin{align*} 
	I(\alpha, \beta; s)&:=\int \frac{\d^D k}{(2\pi)^D} \frac{1}{(k^2)^\alpha ((k+p)^2)^{\beta}}  = \frac{s^{\frac D 2 -\alpha -\beta}}{(4 \pi)^{\frac D 2}}   \frac{ \Gamma \left( \alpha+\beta-\frac D 2 \right) \Gamma \left( \frac D 2 -\alpha  \right) \Gamma \left( \frac D 2 -\beta  \right)  }{\Gamma(\alpha) \Gamma(\beta)  \Gamma \left( D-\alpha-\beta  \right) }.
	\end{align*}
	Since this amplitude   again is a monomial in $s$ with power $\frac D 2 -\alpha -\beta$, it can be inserted recursively to produce multiedges with more than one loop and arbitrary propagator powers. If all propagators have the same power $(k^2)^{-\alpha}$, the $l$-loop multiedge graph has the Feynman amplitude
	\begin{align}\label{Ml_gamma} 
	M^{(l)} &=   \frac{s^{ l\left(\frac D 2 -\alpha\right)  -\alpha } }{ (4\pi)^{l \frac D 2}} \left(\frac{   \Gamma   \left( \frac D 2 -\alpha  \right)  }{  \Gamma (\alpha)   }\right)^{l+1}   \frac{ \Gamma \left( \alpha -l\left(\frac D 2 -\alpha\right)   \right)  }{    \Gamma \left( (l+1)\left( \frac D 2 - \alpha\right)  \right) }
	\end{align}
	These graphs need a symmetry factor $\frac{1}{(l+1)!}$ for the exchange of $l+1$ equivalent edges which is not included in \cref{Ml_gamma}.

	\begin{lemma}\label{lem_Ml}
		Let $H_n = \sum_{j=1}^n j^{-1}$ be the $n^{\text{th}}$ harmonic number and define $s:=p^2$ to be the external momentum squared.
		Then in $D=4-2\epsilon$ dimensions, the massless $l$-loop multiedge with propagators $i(k^2)^{-1}$, not including the symmetry factor, has the Feynman amplitude
		\begin{align*} 
		M^{(l)}(s)  = -   \frac{ \left(- i  s \right) ^{l-1}}{ (4\pi)^{2l} \left( l! \right) ^2} \left(  \frac{1}{\epsilon }  + (2l+1) H_l-1 +l\left( \ln(4\pi)-\gamma_E\right) - l\ln s   \right) +\mathcal O \left( \epsilon  \right)  .
		\end{align*}
	\end{lemma}
	\begin{proof}
		Expand \cref{Ml_gamma} in $\epsilon$, setting $\alpha=1$. With $D=4-2\epsilon$ we have
		\begin{align}\label{Ml_depsilon}
		M^{(l)} (s)   &=  \frac{s ^{l-1-l\epsilon}}{(4 \pi)^{l (2-\epsilon)}}   \left(\Gamma \left( 1-\epsilon \right) \right)^{l+1}  \frac{\Gamma \left(  -l+1+l\epsilon \right) }{   \Gamma \left( l+1 - (l+1)\epsilon \right)  }.
		\end{align}
		The only singular factor for $\epsilon \rightarrow 0$ is the second gamma function in the numerator. 
		Its series representation is well known,
		\begin{align*} 
		\Gamma(-l+1+l\epsilon) &= \frac{(-1)^{l-1}}{l!}\left(  \frac{1}{\epsilon } + l\psi(l) + \mathcal O \left( \epsilon \right) \right)  .
		\end{align*}
		Here, $\psi(l)$ is the digamma function, with integer argument $l>0$ it has the value \cite[§5.4]{dlmf}
		\begin{align*}
		\psi(l) &= \sum_{k=1}^{l-1} \frac 1 k -\gamma_E = H_{l-1}-\gamma_E
		\end{align*}
		where $\gamma_E$ is Euler's constant. All other factors in \cref{Ml_depsilon} are regular for $\epsilon \rightarrow 0$, consequently their $\mathcal O(\epsilon^1)$ coefficients need to be included to produce an overall $\mathcal O (\epsilon^0)$ result. These are
		\begin{align*}
		\frac{1}{\Gamma \left( l+1-(l+1)\epsilon \right) } &= \frac{1}{\Gamma(l+1)}+ \epsilon (l+1) \frac{\psi(l+1)}{\Gamma(l+1)}+ \mathcal O \left( \epsilon^2 \right) = \frac{1}{l!}\left(1 + \epsilon(l+1) \left( H_l-\gamma_E \right)  + \mathcal O \left( \epsilon^2 \right)\right) ,
		\end{align*}
		\begin{align*}
		\left(\Gamma(1-\epsilon)\right)^{l+1} &= 1 +  \epsilon (l+1) \gamma_E + \mathcal O \left( \epsilon^2 \right) ,\qquad &s^{l-1-l\epsilon} &= s^{l-1} \left( 1- \epsilon l \ln s + \mathcal O \left( \epsilon^2 \right)   \right), \\
		(4\pi)^{-2l+l\epsilon} &= (4\pi)^{-2l} \left( 1+\epsilon l\ln (4\pi)  + \mathcal O \left( \epsilon^2 \right)   \right).
		\end{align*}
		Finally, use $H_l = H_{l-1}+l^{-1}$ and include  a factor $i^{l+1}$ for $l+1$ internal propagators.
	\end{proof}
	Note that the presence of momentum-independent terms such as $l \ln (4\pi)$ depends on the choice of integration measure and varies in the literature. 
	
	Especially, given $\frac D 2 -\alpha= 1-\epsilon$ (as is fulfilled for the conventional propagator $s=p^2$ in $D=4-2\epsilon$ dimensions), all massless multiedge graphs have a simple singularity
	\begin{align}\label{Mldiv}
	M^{(l)}_{\text{div}} &=  -    \frac{ \left(- i  s \right) ^{l-1}}{ (4\pi)^{2l} \left( l! \right) ^2}   \frac{1}{\epsilon }
	\end{align}
	with no logarithmic momentum-dependence. This reflects the fact that any subgraph would give rise to a massless tadpole cograph with vanishing amplitude. In this sense, massless multiedges are primitive despite containing power-counting divergent subgraphs.
	
	\begin{example}[Multiedges]\label{ex_multiedges}
		We will use frequently the one- and two-loop multiedges which read
		\begin{align*}
		M^{(1)}(s) &= -\frac{1}{(4\pi)^2}\left( \frac 1 \epsilon + 2-\gamma_E +\ln(4\pi) -\ln s  \right), \\
		M^ {(2)}(s) &=  \  \frac{ is }{4 (4\pi)^4} \left( \frac 1 \epsilon + \frac{11}{2}-2\gamma_E  + 2 \ln(4\pi) -2 \ln s   \right) .
		\end{align*}
	\end{example}

\section{Massless connected 3-point function}\label{sec:triangles}

	The amputated connected 3-point function in the connected perspective is given by four different graph topologies: 
	\begin{enumerate}
		\item A single vertex $G_{3,(1)} = -ib_3 \left( s_1 + s_2 + s_3 \right) $
		\item One Multiedge $M^{(l)}$ where one end is adjacent to two external edges, the other  to  one,
		\item Two multiedges joined at one external vertex,
		\item Triangle-shaped graphs where the three internal edges themselves are multiedges.
	\end{enumerate}
	The first type is trivial, the second one comes in three different orientations with regard to the external momenta. For each orientation, there is a sum over all $l$-loop multiedges,
	\begin{align*}
		G_{3,(2)}&= -\sum_{l=1}^\infty \frac{b_{l+2} b_{l+3} }{(l+1)!}\left( (s_1+s_2)s_3 M^{(l)}(s_3) + (s_1+s_3)s_2 M^{(l)}(s_2) + (s_2+s_3) s_1M^{(l)}(s_1) \right) . 
	\end{align*}
	Observe that the three series do not coincide with the two-point function \cref{G2}: The latter has coefficients $\propto b_{l+2}^2$ whereas here they are $b_{l+2} b_{l+3}$. Since all $b_n$ are independent, there need not be any specific relation between the two series.

		Using \cref{lem_Ml}, $G_3^{(2)}$ for the massless theory $s=p^2$ reads
		\begin{align}\label{G32}
			G_{3,(2)} &= \left \langle 3 \right \rangle    \sum_{l=1}^\infty \frac{(-i )^{l+1} b_{l+2}b_{l+3}}{(4 \pi)^{2l}(l+1)! (l!)^2}      (s_1 + s_2)s_3^l   \left( \frac{1}{\epsilon }  + (2l+1) H_l-1 -l \gamma_E   + l \ln (4\pi) -l \ln s_3  \right) 
		\end{align} 
		where $\left \langle 3 \right \rangle $ by \cref{def_permutations} denotes the symmetric sum over three cyclic permutations of $\left \lbrace s_1,s_2,s_3 \right \rbrace  $.

	Similarly, the third topology is a product of two multiedges
	\begin{align}\label{G33}
		G_{3,(3)}&= \left \langle 3 \right \rangle i s_1s_2s_3 \sum_{l_1=1}^\infty\sum_{l_2=1}^\infty b_{l_1+2} b_{l_2+2} b_{l_1+l_2+3} \frac{M^{(l_1)}(s_1)}{(l_1+1)!} \frac{M^{(l_2)}(s_2)}{(l_2+1)!}. 
	\end{align}
	
	The fourth contribution to the 3-point-function for the massless theory is given by triangle graphs where the propagators are  multiedges on $n$ edges, $M^{(n-1)}$.Let $T^{(n_1,n_2,n_3)}$ be the triangle graph where the internal edge opposed to the external momentum $s_j$ is a multiedge $M^{(n_j-1)}$ . Like for $M^{(l)}$ in \cref{lem_Ml}, we do not include the symmetry factor into $T^{(n_1,n_2,n_3)}$, it is $\frac{1}{n_1! n_2! n_3!}$. Then, the fourth contribution to the connected 3-point amplitude is the series
	\begin{align*}
		G_{3,(4)} &= \sum_{n_1=0}^\infty \sum_{n_2=0}^\infty \sum_{n_3=0}^\infty b_{n_1+n_2+1} b_{n_2+n_3+1} b_{n_3+n_1+1} \frac{1}{n_1! n_2! n_3!} T^{(n_1,n_2,n_3)} (s_1,s_2,s_3).
	\end{align*}

	In a massless theory, the multiedge graph $M^{(l)} \propto     s^{ l\left(\frac D 2 -\alpha\right)  -\alpha }  $  (\cref{Ml_gamma}) amounts to a propagator with the power $-l\left( \frac D 2-\alpha \right) +\alpha$. This means that a triangle graph, where the edges are replaced with multiedes, reduces up to prefactors to a simple massless triangle graph where the edges carry said non-integer propagator power.
	The latter graph evaluates to Appel's hypergeometric $F_4$ functions \cite{boos_method_1987,suzuki_massless_2002}.   At this point we just quote well-known results for the lowest orders. The one-loop order is given in \cite[Eq. (2.11)]{davydychev_recursive_1992}, it is convergent in 4 dimensions. 	
	At two-loop order, the corresponding graph is called the massless dunce's cap ($\Gamma_A$ in \cref{fig_3point_2loop}), again without the symmetry factor $\frac 12$  it reads \cite{ussyukina_new_1994}
	\begin{align}\label{G34_2loop} 
	T^{(0,0,1)} (s_1,s_2,s_3)&=  \frac{1}{2 (4\pi)^4}\left(  \frac{1}{\epsilon^2} +   \left(  5 -2\gamma_E   -2\ln s_3   + 2 \ln (4 \pi) \right) \frac{1}{\epsilon } + \text{finite terms} \right).   
	\end{align}
	If two of the external edges are onshell then $T^{(n_1, n_2, n_3)}$ can be computed similarly to the multiedge graph and evaluates to a product of Gamma functions.

\end{appendix}

\printbibliography

\end{document}